\documentclass[11pt]{article}

\usepackage[utf8]{inputenc}
\usepackage[english]{babel}
\usepackage{graphicx}         
\usepackage[pdftex]{hyperref}
\usepackage{braket}
\usepackage{amsmath}
\usepackage{amssymb}
\usepackage{amsthm}
\usepackage{mathtools}
\usepackage{bbm}		
\usepackage{array}		
\usepackage{color}
\usepackage[usenames,dvipsnames]{xcolor}
\usepackage{caption}
\usepackage{enumitem}
\usepackage{subeqnarray}	
\usepackage{wrapfig}

\allowdisplaybreaks

\numberwithin{equation}{section}

\theoremstyle{plain}\newtheorem{definition}{Definition}[section]
\newtheorem{lem}[definition]{Lemma}
\newtheorem{proposition}[definition]{Proposition}

\theoremstyle{remark}\newtheorem{remark}[definition]{Remark}

\theoremstyle{plain}

\theoremstyle{plain}\newtheorem{assumption}{Assumption}

\theoremstyle{plain}\newtheorem{theorem}{Theorem}

\newcommand{\remit}[1]{\begin{enumerate}[label={(\alph*)}, ref={\theremark\alph*}]{#1}\end{enumerate}}

\newcommand{\R}{\mathbb{R}}
\newcommand{\C}{\mathbb{C}}
\newcommand{\N}{\mathbb{N}}

\newcommand{\Fock}{\mathcal{F}}
\newcommand{\Number}{\mathcal{N}}
\newcommand{\vac}{|\Omega\rangle}
\newcommand{\id}{\mathbbm{1}}

\newcommand{\cL}{\mathcal{L}}
\newcommand{\cR}{\mathcal{R}}

\newcommand{\op}{\mathrm{op}}

\let\textl\l
\renewcommand{\l}{\ell}
\renewcommand{\i}{\mathrm{i}}
\newcommand{\e}{\mathrm{e}}
\newcommand{\hc}{\mathrm{h.c.}}

\newcommand{\sym}{\mathrm{sym}}
\newcommand{\Tr}{\mathrm{Tr}}

\newcommand{\bPhi}{{\boldsymbol{\phi}}}
\renewcommand{\Im}{{\mathrm{Im}}}
\renewcommand{\Re}{{\mathrm{Re}}}

\newcommand{\bj}{\boldsymbol{j}}

\renewcommand{\hat}[1]{\widehat{#1}}
\renewcommand{\tilde}[1]{\widetilde{#1}}

\newcommand{\ls}{\lesssim}
\newcommand{\gs}{\gtrsim}

\newcommand{\lr}[1]{\left\langle #1 \right\rangle}
\newcommand{\norm}[1]{\lVert#1\rVert}

\renewcommand{\d}{\mathop{}\!\mathrm{d}}
\newcommand{\dx}{\d x}
\newcommand{\dy}{\d y}
\newcommand{\dz}{\d z}
\newcommand{\dt}{\d t}
\newcommand{\ds}{\d s}

\newcommand{\ad}{a^\dagger}

\newcommand\mydots{,\makebox[1em][c]{.\hfil.\hfil.},}
\newcommand\mycdots{\makebox[1em][c]{$\cdot$\hfil$\cdot$\hfil$\cdot$}}

\makeatletter
\DeclareFontFamily{OMX}{MnSymbolE}{}
\DeclareSymbolFont{MnLargeSymbols}{OMX}{MnSymbolE}{m}{n}
\SetSymbolFont{MnLargeSymbols}{bold}{OMX}{MnSymbolE}{b}{n}
\DeclareFontShape{OMX}{MnSymbolE}{m}{n}{
    <-6>  MnSymbolE5
   <6-7>  MnSymbolE6
   <7-8>  MnSymbolE7
   <8-9>  MnSymbolE8
   <9-10> MnSymbolE9
  <10-12> MnSymbolE10
  <12->   MnSymbolE12
}{}
\DeclareFontShape{OMX}{MnSymbolE}{b}{n}{
    <-6>  MnSymbolE-Bold5
   <6-7>  MnSymbolE-Bold6
   <7-8>  MnSymbolE-Bold7
   <8-9>  MnSymbolE-Bold8
   <9-10> MnSymbolE-Bold9
  <10-12> MnSymbolE-Bold10
  <12->   MnSymbolE-Bold12
}{}
\let\llangle\@undefined
\let\rrangle\@undefined
\DeclareMathDelimiter{\llangle}{\mathopen}
                     {MnLargeSymbols}{'164}{MnLargeSymbols}{'164}
\DeclareMathDelimiter{\rrangle}{\mathclose}
                     {MnLargeSymbols}{'171}{MnLargeSymbols}{'171}
\makeatother

\newcommand\smallO[1]{
        \mathchoice
            {
                \ensuremath{\mathop{}\mathopen{}{\scriptstyle\mathcal{O}}\mathopen{}\left(#1\right)}
            }
            {
                \ensuremath{\mathop{}\mathopen{}{\scriptstyle\mathcal{O}}\mathopen{}\left(#1\right)}
            }
            {
                \ensuremath{\mathop{}\mathopen{}{\scriptscriptstyle\mathcal{O}}\mathopen{}\left(#1\right)}
            }
            { 
                \ensuremath{\mathop{}\mathopen{}{o}\mathopen{}\left(#1\right)}
            }
    }

\newcommand{\PsiN}{\Psi_N}

\newcommand{\UN}{U_N}

\newcommand{\FNp}{{\Fock_\perp^{\leq N}}}

\newcommand{\Fp}{{\Fock_\perp}}

\newcommand{\Chi}{{\boldsymbol{\chi}}}
\newcommand{\tChi}{\tilde{\Chi}}
\newcommand{\Chiz}{\Chi_{0}}

\renewcommand{\P}{\mathbb{P}}
\newcommand{\Q}{\mathbb{Q}}

\newcommand{\Pz}{\P_{0}}
\newcommand{\Qz}{\Q_{0}}

\newcommand{\Ez}{E_{0}}

\newcommand{\Np}{\Number_\perp}

\newcommand{\FockB}{\mathbb{B}}
\newcommand{\FockA}{\mathbb{A}}
\newcommand{\FockO}{\mathbb{O}}

\newcommand{\goint}{\oint_{\gamma}}

\newcommand{\vNb}{v_{N,\beta}}
\newcommand{\hvNb}{\hat{v}_N^\beta}

\newcommand{\HNb}{H_{N,\beta}}
\newcommand{\cHNb}{\mathcal{H}_{N,\beta}}
\newcommand{\ENb}{E_{N,\beta}}

\newcommand{\FockHb}{\FockH_\beta}

\newcommand{\Ls}{{\Lambda^*}}
\newcommand{\Lsp}{{\Lambda_+^*}}

\newcommand{\boldKz}{\mathbb{K}_{0}}
\newcommand{\boldKo}{\mathbb{K}_{1}}
\newcommand{\boldKt}{\mathbb{K}_{2}}
\newcommand{\boldKth}{\mathbb{K}_{3}}
\newcommand{\boldKf}{\mathbb{K}_{4}}
\newcommand{\boldKtbar}{\mathbb{K}_2^*}
\newcommand{\boldKthbar}{\mathbb{K}_{3}^*}

\newcommand{\FockH}{\mathbb{H}}

\newcommand{\FockG}{\mathbb{G}}

\newcommand{\FockGz}{\FockG_0}

\newcommand{\FockR}{\mathbb{R}}
\newcommand{\FockRz}{\FockR_{0}}

\newcommand{\FockT}{\mathbb{T}}

\newcommand{\BogUz}{\mathbb{U}_{\tau}}

\newcommand{\aNb}{\mathfrak{a}_{N,\beta}}

\newcommand{\abs}[1]{\left| #1 \right|}
\newcommand{\scp}[2]{\left\langle #1 , #2 \right\rangle}

\newcommand{\Rz}{\frac{1}{z-\FockGz}}
\newcommand{\RQz}{\frac{\mathbb{Q}_0}{z-\FockGz}}
\newcommand{\REz}{\frac{\Qz}{E_0-\FockGz}}
\newcommand{\RG}{\frac{1}{z-\FockG}}
\newcommand{\RQG}{\frac{\mathbb{Q}}{z-\FockG}}

\usepackage[colorinlistoftodos]{todonotes}

\usepackage{ulem}

\title{Ground state of Bose gases interacting through singular potentials}

\author{Lea Boßmann\thanks{Department of Mathematics, Ludwig-Maximilians-Universität München, Theresienstr. 39, 80333 Munich, Germany. \texttt{bossmann@math.lmu.de}},\;
Nikolai Leopold\thanks{Department of Mathematics and Computer Science, University of Basel, Spiegelgasse 1, 4051 Basel, Switzerland. \texttt{nikolai.leopold@unibas.ch}},\;
Sören Petrat\thanks{School of Science, Constructor University Bremen, Campus Ring 1, 28759 Bremen, Germany. \texttt{spetrat@constructor.university}},\;
and Simone Rademacher\thanks{Department of Mathematics, Ludwig-Maximilians-Universität München, Theresienstr. 39, 80333 Munich, Germany. \texttt{simone.rademacher@math.lmu.de}}}
\date{\today}

\begin{document}
\maketitle

\begin{abstract}
\noindent We consider a system of $N$ bosons on the three-dimensional unit torus. The particles interact through repulsive pair interactions of the form $N^{3\beta-1}v(N^\beta x)$ for $\beta\in(0,1)$. We prove the next order correction to Bogoliubov theory for the ground state and the ground state energy.
\end{abstract}

\section{Introduction and main results}
We consider a system of $N$ interacting bosons on the three-dimensional unit torus, i.e., the box $\Lambda = [0,1]^{3}$ with periodic boundary conditions. The system is described by the $N$-body Hamiltonian
\begin{align}
\HNb &=   \sum_{j=1}^N  - \Delta_j + \frac{1}{N}\sum_{1 \leq i < j \leq N} \vNb (x_i - x_j) \,,
\end{align}
where
\begin{equation}
\vNb(x)=N^{3\beta}v(N^\beta x)\,,\qquad \beta\in(0,1)\,.
\end{equation}
We make the following assumption on the interaction $v$:

\begin{assumption}\label{ass}
Let $v=\kappa V$ for $\kappa\in\R$ suffiently small and let $V:\R^3\to\R$ be bounded, compactly supported, non-negative and spherically symmetric. Moreover, assume that $v$ is of positive type, i.e., that it has a non-negative Fourier transform. 
\end{assumption}

The scaling parameter $\beta$ interpolates between the mean-field regime ($\beta=0$), which describes a system of many weak and long-range interactions, and the Gross--Pitaevskii regime ($\beta=1$), featuring few and essentially on-site interactions. In this paper, we focus on the regime $\beta\in(1/2,1)$ of singular interactions well beyond the mean-field regime.

The Hamiltonian $\HNb$ acts on the Hilbert space of permutation symmetric, square integrable functions on $\Lambda^N$, which we denote as $L^2_s(\Lambda^N)\subset L^2(\Lambda^N)$.
It is well known  that $\HNb$ has a discrete spectrum and a unique ground state $\PsiN$, which is defined as
\begin{equation}\label{def:ENb}
\HNb\PsiN=\ENb\PsiN\,,\qquad \ENb:=\inf\sigma(\HNb)\,.
\end{equation}
An important property of the potential $\vNb$ is its scattering length, which can be defined as follows: Denote by $f:\Lambda\to\R$ the scattering solution on the torus, i.e., the solution of the equation
\begin{equation}\label{scattering:solution}
\left(-\Delta +\frac{1}{2N}\vNb\right) f=0\,.
\end{equation}
The scattering length of $\vNb$ on the torus is then given by
\begin{equation}\label{def:box:scattering:length}
8\pi\aNb:=
\sum_{p\in\Ls}\hat{v}\left(\frac{p}{N^\beta}\right)\hat{f}_p\,,
\end{equation}
where $\hat{f}_p:=\int_\Lambda f(x)\e^{\i p\cdot x}$ denotes the Fourier transform of the scattering solution. We furthermore introduce 
\begin{equation}
\label{def:eta1}
\eta_p:=N \left(\hat{f}_p-\delta_{p,0}\right)\,.
\end{equation}

Our main goal in this paper is to prove the first three terms in an expansion of the ground state energy $\ENb$ for large $N$. The leading and next-to-leading order have first been predicted by Bogoliubov in 1947 \cite{bogoliubov1947}; a rigorous proof was given for our model by Boccato, Brennecke, Cenatiempo and Schlein in \cite{boccato2017_2}. In our main result, we now prove the third order in the expansion:

\begin{theorem}\label{thm:energy}
Let Assumption \ref{ass} be satisfied and let $\beta\in(\frac12,1)$. Then there exists a constant $C>0$ such that
\begin{equation}\label{eqn:thm:energy}
\left|\ENb-4\pi  (N-1)  \mathfrak{a}_N^\beta - E_{0,0} - E_{\mathrm{corr}} \right|\leq CN^{\frac32(\beta-1)}
\end{equation}
for sufficiently large $N$,
where
\begin{align}
E_{0,0}&=  \frac{1}{2} \sum_{p \in \Lambda_+^*} \left( -p^2 - \widehat{v} (0) + \sqrt{\vert p \vert^4 + 2p^2 \widehat{v} (0)} + \frac{\widehat{v}(0)^2}{2p^2} \right) \,\label{E00}
\end{align}
and
\begin{align}\label{E_corr_def}
E_{\mathrm{corr}} = C_{N,\beta} \Bigg( -\frac{1}{2N} \sum_{p \in \Lsp} \frac{\widehat{v}_N^\beta (p)^2}{2p^2} \Bigg),
\end{align}
with 
\begin{subequations}
\begin{align}
C_{N,\beta} &= \frac{1}{2} \sum_{p \in \Lsp} (s_pc_p-\eta_p) + \widehat{v}(0)^2 \sum_{p \in \Lsp} \frac{1}{\sqrt{\vert p \vert^4 + 2 p^2 \widehat{v}(0)} \left( p^2 + \sqrt{\vert p \vert^4 + 2 p^2 \widehat{v}(0)} \right)} \label{thm_const_Bog} \\
&\quad + \sum_{p \in \Lsp} 4 \Big( c_p\widetilde{s}_p + 2 \widetilde{c}_p s_p \Big)^2.\label{thm_const_pert}
\end{align}
\end{subequations}
Here, we introduced the notation $s_p := \sinh( \eta_p)$ and $c_p := \cosh(\eta_p)$ with $\eta_p$ given by \eqref{def:eta1}, and $\widetilde{s}_p := \sinh( \tau_p)$ and $\widetilde{c}_p := \cosh(\tau_p)$ with $\tau_p$ given by \eqref{def:tau}.
\end{theorem}

The coefficient $E_{0,0}$, which is of order one, is the leading order of the Bogoliubov ground state energy. It has been proven rigorously for our model in \cite{boccato2017_2} under the assumption of a sufficiently small interaction potential. Previously, Bogoliubov theory had been rigorously justified in the mean-field regime ($\beta=0$) \cite{seiringer2011,grech2013,lewin2015_2,nam2014,derezinski2014}. For the mathematically very difficult Gross--Pitaevskii regime ($\beta=1)$, the Bogoliubov correction to the ground state energy was obtained in \cite{boccato2018,nam2021,brennecke2022_2,hainzl2022,caraci2023,basti2023}. Note that all these works are not only restricted to the ground state energy but extend to the full low-energy excitation spectrum.

The coefficient $E_{\mathrm{corr}}$ is of order $N^{\beta-1}$. The contribution from line \eqref{thm_const_Bog} comes from the next order of the Bogoliubov ground state energy and can therefore also be retrieved from the analysis in \cite{boccato2017_2}, where it was absorbed in the error term (see Section \ref{sec:proof:diag}). The main novelty of our result is the proof of the contribution coming from \eqref{thm_const_pert}. A comparable result beyond Bogoliubov theory has previously been obtained for $\beta=0$ in \cite{spectrum}, see also \cite{pizzo2015} for related results in a different setting.

After the first version of our work was published as a preprint, two new works appeared which improved the error estimate for the Bogolibov approximation in the Gross--Pitaevskii regime \cite{brooks2023} and proved the third-order correction to the ground state energy~\cite{caraci2023_2}.\\
 
In addition to the ground state energy, we are interested in the ground state wave function $\PsiN$, and in particular in the reduced one-body density matrix of the ground state,
\begin{equation}
\gamma_N^{(1)}:=\Tr_{L^2(\Lambda^{N-1})}|\PsiN\rangle\langle\PsiN|\,.
\end{equation} 
Since we are in the translation invariant setting on the torus, $\PsiN$ exhibits  Bose--Einstein condensation (BEC) in the state $\varphi_0\equiv 1$ (see, e.g., \cite{lieb2002,boccato2017,boccato2018_2,nam2020,brennecke2022,hainzl2020} for proofs of complete BEC in various settings). In our second main theorem, we prove the next order correction to the ground state. This allows us to show that the projection onto the condensate state approximates the reduced one-body density matrix up to an error of order $N^{\frac32(\beta-1)}$:
 
\begin{theorem}\label{thm:state}
Let Assumption \ref{ass} be satisfied, let $\beta\in(\frac12,1)$ and let $N$ be sufficiently large. Then there exists a constant $C>0$ such that
\begin{equation}
\left\|\PsiN -\Psi_{N,0}-\Psi_{N,1}-\Psi_{N,2}\right\|_{L^2(\Lambda^N)} \leq C N^{\frac32(\beta-1)}\,,
\end{equation}
where the functions $\Psi_{N,\l}\in L^2_s(\Lambda^N)$ are defined in \eqref{def:PsiNl}.
Moreover, it holds that
\begin{equation}\label{eqn:thm:RDM}
\Tr\left|\gamma_N^{(1)}-\gamma_0\right|\leq C N^{\frac32(\beta-1)}
\end{equation}
for $\gamma_{0}(x,y)=1$.
\end{theorem}

The leading order $\gamma_0$ is determined by the condensate. From our perturbative approach, we obtain a correction $\gamma_1$ to $\gamma_0$, which is, on the torus, of order $N^{-1}$ because the leading term of $\gamma_1$ vanishes by translation invariance. Hence, $\gamma_1$ is subleading compared to the overall error, while we expect to see a nonzero contribution in the inhomogeneous setting. The proof of Theorem \ref{thm:state} is given in Section \ref{sec:proof:thm:state}.
A comparable expansion of the reduced density matrix was previously obtained for the mean-field regime $\beta=0$ in \cite{spectrum,nam2020_2,proceedings}.

\begin{remark}
\remit{
\item The bounds in \eqref{eqn:thm:energy} and \eqref{eqn:thm:RDM} for the ground state energy and the reduced densities are due to the fact that we do a perturbative expansion and each step of the perturbation series gains a factor $N^{(\beta-1)/2}$. Since we truncate the expansion after two iterations, this leads to the error $N^{3(\beta-1)/2}$, which is not optimal because one can infer from parity arguments that all orders with non-integer powers of $N^{\beta-1}$ are zero. Hence, we expect the optimal errors to be of the order $N^{2(\beta-1)}$. We conjecture that our expansion could be extended to higher orders as in \cite{spectrum}, which would also yield the optimal bounds.

\item We assume that the interaction potential is sufficiently small. The only reason why we need this smallness assumption is that we use the result from \cite{boccato2017_2} to construct the ground state projector (see Section \ref{sec:intro:pert:theory}). Given the work \cite{boccato2018}, where the corresponding result was shown for $\beta=1$ without the smallness assumption, we expect that \cite{boccato2017_2} can be extended to work without this restriction. Then our result holds without restricting to small potentials.

\item In \cite{boccato2018}, $\ENb$ is approximated for $\beta=1$ by $E_{0,0}$ but with $\hat{v}(0)$ replaced by the scattering length. Given this result, we expect that we should be able to reconstruct the scattering length in $E_{0,0}$ from the corrections $E_{\mathrm{corr}}$, at least for values of $\beta$ close to one. This can nicely be seen from the formula \eqref{E_corr_def}, where the next order of the Born series of the scattering length appears. Let us also remark that in \cite{boccato2018} the formulas contain the infinite volume scattering length, while we work with the box scattering length on the torus as in \cite{boccato2017_2}. The difference between both quantities is of order $N^{-1}$ (see, e.g., \cite[Lemma 9]{hainzl2020}).

\item The restriction $\beta>1/2$ is only for technical reasons since it allows to summarize several error terms in a more efficient way. Our proof works for the full range $\beta\in(0,1)$, but for small $\beta$ we would obtain some extra contributions to the energy which are for $\beta>1/2$ subleading. 

\item We expect that our result can be extended to the full low-energy excitation spectrum similarly to \cite{spectrum}.
}
\end{remark}

Formally, the next order corrections in Theorems \ref{thm:energy} and \ref{thm:state}  can be computed by Rayleigh-Schrödinger perturbation theory. To make this rigorous, we follow the general approach introduced in \cite{spectrum} for the mean-field regime, which was inspired by the previous works \cite{corr,QF} in the dynamical context. The general idea in the mean-field regime is to first unitarily transform the $N$-body Hamiltonian to the corresponding excitation Hamiltonian $\FockH^\mathrm{mf}$ on the excitation Fock space, as is by now the standard procedure \cite{lewin2015_2}. Subsequently, $\FockH^\mathrm{mf}$ is  expanded in a power series  in $N^{-1/2}$ around the Bogoliubov Hamiltonian $\FockH_0^\mathrm{mf}$. Finally, one expresses the ground state projector as
$$\P^{\mathrm{mf}}= \frac{1}{2\pi\i}\goint\frac{1}{z-\FockH^\mathrm{mf}}\dz$$
for a suitable contour $\gamma$ encircling the ground state energy and uses the expansion of $\FockH^\mathrm{mf}$ to construct an expansion of $\P^\mathrm{mf}$
around the Bogoliubov ground state projector 
$$\P_0^{\mathrm{mf}}= \frac{1}{2\pi\i}\goint\frac{1}{z-\FockH^\mathrm{mf}_0}\dz\,.$$%
In our situation, where $\beta\in(1/2,1)$ is far beyond the mean field regime and the interaction potentials converge to $\delta$-interactions, any naive perturbative expansion must fail. The idea is therefore to first reduce the problem to a mean-field problem by conjugating the excitation Hamiltonian $\FockH$ with a suitable quadratic transformation which regularizes $\FockH$. 

There are different  suggestions for such a quadratic transformation in the literature, see e.g.\ \cite{boccato2017_2,nam2021,hainzl2022}. We choose a transformation \eqref{eq:Bogoliubov transformation} which slightly differs from these works. It is a true Bogoliubov transformation of the form 
$$\FockT=\exp\left\{\frac12\sum_{p}\eta_p(\ad_p\ad_{-p}-a_pa_{-p})\right\}\,$$
for $\eta_p$ from \eqref{def:eta1}. In this way, conjugating $\FockH$ with $\FockT$ directly reconstructs the scattering length \eqref{def:box:scattering:length} in the leading order of $\ENb$. Another advantage is that the action of $\FockT$ on creation and annihilation operators is explicit, which makes the computation of $\FockT\FockH\FockT^*$ relatively simple. On the other hand, the transformation $\FockT$ does not preserve the truncation of the excitation Fock space (in contrast to e.g.\ \cite{boccato2017_2}, where a generalized Bogoliubov transformation was used), which makes the estimates more involved.

The transformed Hamiltonian $\FockT\FockH\FockT^*$ can now be used for perturbation theory. While we follow in principle the road of \cite{spectrum}, obtaining sufficiently strong estimates is much more difficult than in the mean-field setting because we deal with very singular potentials. \\

The paper is structured as follows: In Section \ref{sec:method} we introduce all necessary notation and explain the strategy of the proof. The two main building blocks of the proof are Proposition~\ref{prop:renormalized:hamiltonian}, which concerns the renormalization of $\FockH$ and the estimates of the errors, and Proposition~\ref{prop:expansion:P}, containing the perturbation theory. Section \ref{sec:BT} collects some useful bounds on the quadratic transformations. In Section \ref{sec:proof:diag} we diagonalize the Bogoliubov Hamiltonian and extract the relevant contributions to the ground state energy. In Section \ref{sec:estimates-G}, we prove Proposition \ref{prop:renormalized:hamiltonian}. Section \ref{sec:excitation:vector} contains estimates of the kinetic energy and the number of excitations in the ground state. Finally, we prove Proposition \ref{prop:expansion:P} in Section \ref{sec:pert:theory}. In Section \ref{section:explicit calculation of E-pert}, we explicitly calculate the energy correction terms, and in Section \ref{sec:proof:thm:state} we prove Theorem \ref{thm:state}.

\subsection*{Notation}
\begin{itemize}
\item We denote the Fourier transform on $\Lambda$ as
\begin{equation}
\hat{f}(p):=\int_\Lambda f(x)\e^{-\i p\cdot x}\dx\,,\qquad
\check{f}(x):=\sum_{p\in\Lambda^*}f(p)\e^{\i p\cdot x}
\end{equation} 
for $\Lambda^*=2\pi\mathbb{Z}^3$. We also use the notation $\Lsp=\Ls\setminus\{0\}$.

\item We abbreviate $\hvNb:=\hat{v}(\cdot/N^\beta)$.

\item The notation $A\ls B$ indicates that there exists some constant $C>0$ such that $A\leq CB$.
\end{itemize}

\section{Method}\label{sec:method}
From now on we will always assume that Assumption \ref{ass} is satisfied. 

\subsection{Excitation Fock space}
To compute the higher order terms in the ground state energy, we need to focus on the excitations from the condensate. Since we are in the translation invariant setting on the torus, the condensate is described by the constant function $\varphi_0:\Lambda\to\C$, $\varphi_0\equiv 1$. 
We follow the method introduced in \cite{lewin2015_2} to split the ground state $\PsiN$ into a condensate and  excitations as
\begin{equation}
\PsiN=\sum\limits_{k=0}^N{\varphi_0}^{\otimes (N-k)}\otimes_s\tilde{\chi}^{(k)}\,,
\qquad \tilde{\chi}^{(k)}\in \bigotimes\limits_\sym^k L^2_\perp(\Lambda)\,, 
\end{equation}
with $\otimes_s$ the symmetric tensor product and where $L^2_\perp(\Lambda)$ denotes the $L^2$-orthogonal complement of $\varphi_0$.
The sequence
\begin{equation}\label{def:Chi}
\tChi:=\big(\tilde{\chi}^{(k)}\big)_{k=0}^N\oplus 0
\end{equation}
of $k$-particle excitations forms a vector in the (truncated) excitation Fock space over $L^2_\perp(\Lambda)$,
\begin{equation}\label{Fock:space}
\FNp=\bigoplus_{k=0}^N\bigotimes_\sym^k L^2_\perp(\Lambda)
\;\subset \;
\Fp=\bigoplus_{k=0}^\infty\bigotimes_\sym^k L^2_\perp(\Lambda)\,,
\end{equation}
where the direct sum in \eqref{def:Chi} is with respect to the decomposition $\id=\id^{\leq N}\oplus \id^{>N}$. In the following, direct sums are always understood in this sense, unless otherwise specified.

The creation and annihilation operators on $\Fp$ are
\begin{eqnarray}
(\ad(f)\bPhi)^{(k)}(x_1\mydots x_k)&=&\frac{1}{\sqrt{k}}\sum\limits_{j=1}^kf(x_j)\phi^{(k-1)}(x_1\mydots x_{j-1},x_{j+1}\mydots x_k)\,,\; k\geq 1\,,\\
(a(f)\bPhi)^{(k)}(x_1\mydots x_k)&=&\sqrt{k+1}\int\d x\overline{f(x)}\phi^{(k+1)}(x_1\mydots x_k,x)\,,\; k\geq 0
\end{eqnarray} 
for $f\in L^2_\perp(\Lambda)$ and $\bPhi\in\Fp$. 
They satisfy the canonical commutation relations 
\begin{equation}
[a(f),\ad(g)]=\lr{f,g}\,,\qquad [a(f),a(g)]=[\ad(f),\ad(g)]=0\,.
\end{equation}
Since we consider the translation invariant setting, we will mostly work in the momentum space representation. We denote 
\begin{equation}
\Ls:=2\pi\mathbb{Z}^3\,,\qquad \Lsp:=\Ls\setminus\{0\}
\end{equation}
and define for $p\in\Lambda^*$ the normalized plane waves $\varphi_p(x)=\e^{\i p\cdot x}\in L^2(\Lambda)$. The  operators which create/annihilate a particle in the state $\varphi_p$ are given as
\begin{equation}
\ad_p:=\ad(\varphi_p)\,,\qquad a_p:=a(\varphi_p)\,.
\end{equation}
We denote the vacuum of $\Fp$ as $\vac=(1,0,0,\dots)$. The number operator on $\Fp$ is given by
\begin{equation}
\Np:=\sum_{p\in\Lsp}\ad_p a_p\,,\qquad (\Np\bPhi)^{(k)}=k\phi^{(k)}\;\text{ for }
\bPhi\in\Fp\,.
\end{equation}
The $N$-body ground state $\PsiN$ is mapped onto its excitation vector $\tChi$  by the excitation map
\begin{eqnarray}\label{map:U}
\UN: L^2(\Lambda^N) \to  \FNp\;, \quad \tChi:=\UN \PsiN\oplus 0\,.
\end{eqnarray}
The map $\UN$ is unitary and acts as 
\begin{subequations}\label{eqn:substitution:rules}
\begin{eqnarray}
\UN \ad_0a_0\UN^*&=&N-\Np\,,\\
\UN\ad_p a_0 \UN^*&=&\ad_p\sqrt{N-\Np}\,,\\
\UN \ad_0 a_p\UN^*&=&\sqrt{N-\Np}a_p\,,\\
\UN \ad_p a_q\UN^*&=&\ad_p a_q
\end{eqnarray}
\end{subequations}
for $p,q\in\Lsp$.

\subsection{Excitation Hamiltonian}\label{sec:exc:ham}
In this section we conjugate the $N$-body Hamiltonian $\HNb$ with the excitation map $\UN$. The embedding of $\HNb$ in the Fock space is given by
\begin{equation}
\cHNb=\sum_{p\in\Lambda^*}p^2\ad_p a_p +\frac{1}{2N}\sum\limits_{p,q,r\in\Lambda^*}\hvNb(r)\ad_p\ad_q a_{q-r} a_{p+r} \,,
\end{equation}
 where we used the abbreviation
\begin{equation}\label{eqn:hvNb}
\hvNb:=\hat{v}(\cdot/N^\beta)\,.
\end{equation}
Using \eqref{eqn:substitution:rules}, we compute
\begin{align}
\UN H_{N,\beta}\UN^*
&=  \frac12(N-1)\hat{v}(0)+\boldKz + \left[\frac{N-\Np}{N}\right]_+ \boldKo\nonumber\\
&\quad+\left( \boldKt\frac{\sqrt{\left[(N-\Np)(N-\Np-1)\right]_+}}{N}+ \hc\right)\nonumber\\
&\quad+\left( \boldKth\frac{\sqrt{\left[N-\Np\right]_+}}{N}+ \hc\right) 
+\frac{1}{N} \boldKf\label{eqn:FockHN}
\end{align} 
as operator on $\FNp$, where we denoted by $[\cdot]_+$ the positive part and used the shorthand notation 
\begin{subequations}\label{eqn:K:notation}
\begin{eqnarray}
\boldKz&:=&\sum_{p\in\Ls}p^2\ad_pa_p\,,\label{eqn:K:notation:0}\\
\boldKo&:=&\sum_{p\in\Lsp}\hvNb(p)\ad_pa_p\,,\label{eqn:K:notation:1}\\
\boldKt&:=&\tfrac12\sum_{p\in\Lsp}\hvNb(p)\ad_p\ad_{-p}\,,\label{eqn:K:notation:2}\\
\boldKth&:=&\sum_{\substack{p,q\in\Lsp\\p+q\neq0}}\hvNb(p)\ad_{p+q}\ad_{-p}a_q\,,\label{eqn:K:notation:3}\\
\boldKf&:=&\tfrac12\sum_{\substack{p,q,r\in\Lsp\\ p+r\neq0,q+r\neq0}}\hvNb(r) \ad_{p+r}\ad_qa_pa_{q+r} \label{eqn:K:notation:4}\,.
\end{eqnarray}
\end{subequations}
Note that we added the positive part in the prefactors of $\boldKz$, $\boldKt$ and $\boldKth$ for free since the operator $\UN\HNb\UN^*$ is defined on the truncated Fock space where $\Np\leq N$.
To be able to renormalize it in the next step, we first need to extend $\UN\HNb\UN^*$ to an operator ${\FockH}$ on the full Fock space $\Fp$ in such a way that its low-energy spectrum coincides with the low-energy spectrum of $\HNb$. In particular, we require that $\tChi=\UN\PsiN\oplus 0$ is an eigenstate of the new operator $\FockH$ satisfying
\begin{equation}\label{eqn:eigenvalue:eqn:tilde}
{\FockH}\tChi=\ENb\tChi\
\end{equation}
for $\ENb$ from \eqref{def:ENb}.
This is satisfied by the choice ${\FockH}:\Fp\to\Fp$,
\begin{align}\label{def:H:extended}
{\FockH}&:=\frac{N-1}{2}\hat{v}(0)
+\boldKz +\frac{1}{N} \boldKf\nonumber\\
&\quad+ \boldKo\left[\frac{N-\Np}{N}\right]_+ \oplus0
+\left( \boldKt\frac{\sqrt{\left[(N-\Np)(N-\Np-1)\right]_+}}{N}\oplus 0+\hc\right)\nonumber\\
&\quad+\left( \boldKth\frac{\sqrt{\left[N-\Np\right]_+}}{N}\oplus 0+ \hc\right) \,.
\end{align}
Here we extended the particle number preserving operators $\boldKz$ and $\boldKf$ trivially to the full space, whereas all other operators are extended by zero outside $\FNp$. 
To see that \eqref{eqn:eigenvalue:eqn:tilde} holds true, one observes that 
\begin{align}
{\FockH}-4\pi(N-1)\aNb={\FockH}^<\oplus{\FockH}^>\,,
\end{align} 
where 
\begin{align}
{\FockH}^<&:=\UN\HNb\UN^*-4\pi(N-1)\aNb\,,\\
{\FockH}^>&:=\id^{>N}\left(\boldKz+\frac{1}{N}\boldKf\right)\,.
\end{align}
From \cite[Theorem 1.1]{boccato2017_2}, we know that the two lowest eigenvalues of $\FockH^<$ are of order one.
Since $\boldKf\geq 0$ by Assumption \ref{ass} and  as $|p|^2\geq 4\pi^2$ as operator on $L^2_\perp(\Lambda)$, we conclude that
\begin{equation}
\sigma\left({\FockH}^>\right)\subset(4\pi^2 N,\infty)\,.
\end{equation}
Consequently, neither the ground state energy nor the first excited eigenvalue of $\FockH^<$ are elements $\sigma(\FockH^>)$.
Since 
\begin{equation}
\sigma({\FockH})=\sigma({\FockH}^<)\cup\sigma({\FockH}^>)\,,
\end{equation}
this implies that the ground state of ${\FockH}$ is given by $\tChi=\UN\PsiN\oplus0$; in particular, \eqref{eqn:eigenvalue:eqn:tilde} is satisfied. 
Hence, we can from now on work with the operator ${\FockH}$, to which we refer as the excitation Hamiltonian.

\begin{remark}
We remark that there was a small mistake in \cite{spectrum}, where the prefactor of $\boldKo$ was defined without the positive part and the prefactors of $\boldKt$ and $\boldKth$ were expanded in Taylor series on $\Fp$ and not only on $\FNp$. Hence, there are in fact some additional remainder terms, which were not taken into account. However, these remainders are arbitrarily small because they live only on $\Fock_\perp^{>N}$ and any power of the number operator with respect to the low-energy states of both the full Hamiltonian and the Bogoliubov Hamiltonian can be bounded uniformly in $N$. Hence, this does not affect the result of \cite{spectrum}.
\end{remark}

\subsection{Quadratic transformation}
In this section we construct the Bogoliubov transformation which we will use to renormalize the excitation Hamiltonian $\FockH$.
Recall that $f$ denotes the scattering solution on the torus as defined in \eqref{scattering:solution} and 
\begin{equation*}
\eta_p=N \left(\hat{f}_p-\delta_{p,0}\right)
\end{equation*}
as in \eqref{def:eta1}.
Then it follows that
\begin{equation}
\label{def:eta}
p^2 \eta_p+\frac{1}{2N}\sum_{q\in\Lsp}\hat{v}\left(\frac{p-q}{N^\beta}\right)\eta_q=-\frac{1}{2}\hat{v}\left(\frac{p}{N^\beta}\right)
\end{equation}
for $p\in\Lsp$.
As discussed in \cite[Section 3]{hainzl2022}, there exists a solution to \eqref{def:eta}, namely $\check{\eta} \in L^2( \Lambda)$ given by 
\begin{align}
\label{def:check_eta}
\check{\eta} = - \frac{1}{2} q_0 \left[ q_0 \left( - \Delta + \frac{1}{2N} \vNb \right) q_0 \right]^{-1} q_0 \vNb\,,
\end{align}
where $q_0$ denotes the projector onto the orthogonal complement of $\varphi_0\equiv 1$. The Fourier coefficients of the solution $\check{\eta}$ are given through \eqref{def:eta}. \\

Now we can define the Bogoliubov transformation $\FockT$ on $\Fp$ as
\begin{equation}
\label{eq:Bogoliubov transformation}
\FockT:=\exp\left\{\frac12\sum_{p\in\Lsp} {\eta}_p\left(\ad_p\ad_{-p}-a_pa_{-p}\right)\right\}\,.
\end{equation}
It acts on creation and annihilation operators as
\begin{subequations}\label{eq:actionbogo}
\begin{align}
\FockT\ad_p\FockT^*&= c_p\ad_p+s_pa_{-p}\,,\\
\FockT a_p\FockT^*&= c_p a_p+s_p\ad_{-p}\,,
\end{align}
\end{subequations}
where 
\begin{equation}\label{def:s:c}
c_p:=\cosh (\eta_p)\,,\qquad s_p:=\sinh(\eta_p)\,.
\end{equation}
Since $v$ is spherically symmetric and consequently $\hat{v}(p)=\hat{v}(-p)$, it follows that 
\begin{equation}
s_{-p} = s_p\,,\qquad c_{-p} = c_p\,.
\end{equation}

\subsection{Regularized excitation Hamiltonian}
The next step is to renormalize the excitation Hamiltonian $\FockH$ by conjugating it with the Bogoliubov transformation $\FockT$. This is done in the following proposition, whose proof we postpone to Section \ref{sec:estimates-G}.

\begin{proposition}\label{prop:renormalized:hamiltonian}
For $\FockH$ as in \eqref{def:H:extended}, it holds that
\begin{align}\label{def:G}
\FockG:=\FockT \FockH\FockT^*- \mathcal{C}
 =  \FockG_0+\FockG_1+\FockG_2+\FockR_2\,,
\end{align}
where
\begin{align}
\mathcal{C}&=\frac12(N-1)\hat{v}(0)+
\sum_{p\in\Ls} \left(p^2+\hvNb(p)\right) s_p^2+\sum_{p\in\Ls}  \hvNb(p)  c_p s_p\nonumber\\
&\quad+\frac1{2N}\sum_{p\in\Lsp}(\hvNb*cs)_pc_p s_p 
-\frac{1}{N}\sum_{p\in\Lsp}\hvNb(p)\left(\frac12 c_p s_p+ c_p s_p^3\right)\,,
\label{C}\\
\FockGz &= \sum_{p \in \Lambda_+^*} F_p \ad_p a_p + \frac{1}{2}\sum_{p \in \Lambda_+^*} G_p ( \ad_p \ad_{-p}  + a_pa_{-p}) \label{G0}\,,\\
\FockG_1&:=\frac{1}{\sqrt{N}}\FockT(\boldKth+\boldKth^*)\FockT^*\,,\label{G1}\\
\FockG_2&:=\frac{1}{2N} \sum_{\substack{p,q,r\in\Lsp\\ p+r\neq0,q+r\neq0}}\hvNb(r) c_{p+r} c_q c_p  c_{q+r} \ad_{p+r} \ad_q   a_p a_{q+r} \label{G2}
\,,
\end{align}
with coefficients $F_p, G_p$ given by 
\begin{subequations}\label{def:F,G}
\begin{align}
F_p :=& ( c_p^2 + s_p^2) p^2+ ( c_p + s_p)^2 \widehat{v}_N^{\beta} (p) + \frac{2}{N} ( \widehat{v}_N^{\beta} * cs)_p c_ps_p  \,,\\
G_p :=& 2  c_p s_pp^2 + ( c_p + s_p)^2 \widehat{v}_N^\beta ( p) + \frac{1}{N} (\widehat{v}_N^{\beta} * cs)_p (c_p^2 + s_p^2) \,.
\end{align}
\end{subequations}
Here, we abbreviated
\begin{equation}\label{def:convolution}
(\hvNb*cs)_p:=\sum_{\substack{q\in\Lsp\\q\neq p}}\hvNb(p-q)c_qs_q\,.
\end{equation}
The remainder satisfies for $\ell \in \mathbb{R}$ and $\psi, \xi \in \mathcal{F}_\perp$
\begin{align}
\left|\lr{\psi,\FockR_2\xi}\right|
&\ls   N^{\frac32(\beta- 1) }  
\bigg(\norm{(\mathbb{K}_0 +1 )^{1/2}\left(\Np + 1 \right)^{3/4-\l} \psi} \; 
\norm{ \left(\Np + 1 \right)^{5/4+\l} \xi} \notag \\
&\qquad\qquad\qquad+ \norm{(\left(\Np + 1 \right)^{5/4-\l} \psi}\; 
\norm{ (\mathbb{K}_0 +1 )^{1/2}\left(\Np + 1 \right)^{3/4+\l} \xi} \bigg)\; . 
\end{align}
for all $\beta \in  (1/2,1)$. 
\end{proposition}

The decomposition of $\mathbb{G}= \mathbb{T} \mathbb{H} \mathbb{T}^*-\mathcal{C}$ in Proposition \ref{prop:renormalized:hamiltonian} allows to extract the leading-order contribution of the renormalized excitation Hamiltonian. In Section \ref{sec:estimates-G}, we analyze $\mathbb{G}$  and show that $\mathcal{C}$ is of order $N$, while $\mathbb{G}_0$ is of order one, $\mathbb{G}_1$ is of order $N^{(\beta-1)/2}$, and $\mathbb{G}_2$ is of order $N^{\beta-1}$. The remainder $\mathbb{R}_2$ is of even lower order $N^{3(\beta-1)/2}$. 
We can further decompose $\mathbb{G}_1$ as 
\begin{align}
\mathbb{G}_1 &= \frac{1}{\sqrt{N}}\sum_{\substack{p,q\in\Lsp\\p+q\neq0}}\widehat{v}_N^\beta (p) c_{p+q}c_pc_q( \ad_{p+q}\ad_{-p}a_q + {\rm h.c.} ) \notag \\
&\quad +\frac{1}{\sqrt{N}}\sum_{\substack{p,q\in\Lsp\\p+q\neq0}}\widehat{v}_N^\beta (p) c_{p+q}c_ps_q(\ad_{p+q}\ad_{-p}\ad_{-q} +
{\rm h.c.} ) \notag \\
&\quad+ \mathbb{R}_d\,,
\end{align}
where $\mathbb{R}_d$ is of order $N^{-1/2}$ and thus even smaller than the remainder $\mathbb{R}_2$  in the parameter regime $\beta \in(2/3,1)$(see Section \ref{subsec:G_1}, Lemma \ref{lem:G1}).

\subsection{Bogoliubov-Hamiltonian}
The quadratic operator $\FockGz$ is known as the Bogoliubov Hamiltonian, and its ground state energy $E_0$, satisfying
\begin{equation}\label{def:E_0}
\FockGz\Chiz=\Ez\Chiz\,,
\end{equation}
contributes the next-to-leading order in the ground state energy of $\HNb$. To extract it, we need to diagonalize $\FockGz$, which can be done explicitly by means of a Bogoliubov transformation. To construct the diagonalizing transformation, we define for $p\in\Lsp$ the sequence $\tau_p$ by
\begin{equation}\label{def:tau}
\tanh(2\tau_p)=-\frac{G_p}{F_p}\,.
\end{equation}
By Lemma \ref{lemma:propF,G} below, this is well defined. Note that $\tau_p$ is real-valued by definition. Making use of this sequence $\tau_p$, we define the Bogoliubov transformation 
\begin{align}
\label{def:Ttau}
\mathbb{U}_\tau := \exp \left[ \frac{1}{2} \sum_{p \in \Lambda_+^*} \tau_p (\ad_p \ad_{-p} - a_p a_{-p}) \right] \,.
\end{align}
It is well known that $\BogUz$ diagonalizes the quadratic Hamiltonian $\mathbb{G}_0$ as 
\begin{align}
\label{eq:G0-diag}
\mathbb{U}_\tau \mathbb{G}_0 \mathbb{U}_\tau^* = \frac{1}{2} \sum_{p \in \Lambda_+^*} \big[ - F_p + \sqrt{F_p^2 - G_p^2} \big] + \sum_{p \in \Lambda_+^*} \sqrt{F_p^2 - G_p^2} \ad_pa_p \; 
\end{align} 
(see, e.g., \cite[Lemma 5.1]{boccato2017_2}).
Consequently, 
\begin{equation}\label{eqn:Chiz}
\Chiz=\BogUz^*\vac
\end{equation}
and
\begin{equation}
E_0=\frac{1}{2} \sum_{p \in \Lambda_+^*} \big[ - F_p + \sqrt{F_p^2 - G_p^2} \big]\,.
\end{equation}
From Lemma \ref{lemma:propF,G} below, it is easy to see that $E_0=\mathcal{O}(1)$. Adding $E_0$ to the constant $\mathcal{C}$ from Proposition \ref{prop:renormalized:hamiltonian} yields the leading and next-to-leading order contribution to the ground state energy:

\begin{lem}\label{lem:diag}
Let $\beta \in (0,1)$. Then, for $\mathcal{C}$ as in \eqref{C}, it holds for every $\alpha < \beta$ that
\begin{align}
\mathcal{C}  + E_0 =4 \pi  (N-1)  \mathfrak{a}_N^\beta + E_{0,0} +E_{0,1}  + \mathcal{O}(N^{2( \beta -1)}) +\mathcal{O}(N^{-1})+ \mathcal{O}(N^{-\alpha}) \,,
\end{align}
where $\mathfrak{a}_N^\beta$ denotes the scattering length defined in \eqref{def:box:scattering:length} and with $E_{0,0}$ and $E_{0,1}$ given in \eqref{E00} and \eqref{E01}. 
Moreover,
\begin{equation}\label{eqn:orders:E00:01}
|E_{0,0}|\ls 1\,,\qquad |E_{0,1}|\ls N^{\beta-1}\,.
\end{equation}
\end{lem}

The proof of Lemma \ref{lem:diag} is given in Section \ref{sec:proof:diag}.

\subsection{Perturbation theory}\label{sec:intro:pert:theory}
Our goal is  a perturbative expansion of the ground state $\Chi$ of the operator
$$\FockG=\FockT \FockH \FockT^*-\mathcal{C}\,.$$
The ground state $\Chi$  satisfies the eigenvalue equation
\begin{equation}\label{def:E}
\FockG\Chi=(\ENb-\mathcal{C})\Chi=:E\Chi\,.
\end{equation}
Equivalently, $\Chi=\FockT\tilde{\Chi}=\FockT (\UN\PsiN\oplus0)$.
We denote the spectral projectors of $\FockG$ and $\FockGz$ corresponding to their ground state energies $E$ and $E_0$, defined in \eqref{def:E} and \eqref{def:E_0}, respectively, by
\begin{align}
\P&:=|\Chi\rangle\langle\Chi|\,,\qquad\quad\, \Q:=\id-\P\,,\label{def:P}\\
\Pz&:=|\Chiz\rangle\langle\Chiz|\,,\qquad \Qz:=\id-\Pz\,.\label{def:Pz}
\end{align}
In  \cite[Theorem 1.1]{boccato2017_2}, it is shown that
\begin{equation}
\ENb=4\pi(N-1)\aNb+E_{0,0}+\mathcal{O}(N^{-\alpha})
\end{equation}
for $\alpha<\min\{\beta,\frac{1-\beta}{2}\}$. More precisely, the scattering length $\aNb$ is constructed in \cite{boccato2017_2} via its Born series, which is truncated after sufficiently many terms.
Hence, we infer from Lemma~\ref{lem:diag}  that
\begin{equation}\label{eqn:E-E0}
\lim\limits_{N\to\infty}|\ENb-\mathcal{C}-E_0|=\lim\limits_{N\to\infty}|E-E_0|=0\,.
\end{equation}
Moreover, \cite[Theorem 1.1]{boccato2017_2} together with the reasoning in Section \ref{sec:exc:ham} implies that the spectral gap of $\FockG$ above $\ENb$ is of order one. By \eqref{eq:G0-diag}, the same holds true for the spectral gap of $\FockGz$ above $\Ez$. Hence, there exists a constant $c=\mathcal{O}(1)$  such that the closed contour
\begin{equation}\label{gamma}
\gamma:=\{\Ez+ c\,\e^{\i t} \,:\, t\in [0,2\pi)\}\subset \mathbb{C}
\end{equation}
encloses both $E$ and $\Ez$ but contains no other point of the spectra of $\FockG$ and $\FockGz$. Consequently, the projectors $\P$ and $\Pz$ can be expressed as
\begin{equation}
\P=\frac{1}{2\pi\i}\goint\frac{1}{z-\FockG}\,,\qquad
\Pz=\frac{1}{2\pi\i}\goint\frac{1}{z-\FockGz}\,.
\end{equation}
Now we follow the strategy of \cite[Lemma 3.13, Proposition 3.14 and Theorem 2]{spectrum} to expand $\P$ around $\Pz$. The proof of this proposition is given in Section  \ref{sec:pert:theory}. 

\begin{proposition}\label{prop:expansion:P}
Let $\FockA$ be an operator on $\Fp$ such that $\norm{\FockA\psi}\ls \norm{\Np\psi}$ for $\psi\in\Fp$. Then
\begin{equation}\label{eqn:expansion:TrAP}
\left|\Tr\FockA\P-\sum_{\l=0}^2\Tr\FockA\P_\l\right|
\ls N^{\frac32(\beta-1)}\,,
\end{equation}
where 
\begin{subequations}\label{def:Pl}
\begin{align}
\P_1&:=\Pz\FockG_1\REz+\hc\,,\\
\P_2&:=\left(\Pz\FockG_2\REz
+\Pz\FockG_1\REz\FockG_1\REz +\hc\right)\notag\\
&\quad-\lr{\Chiz,\FockG_1,\frac{\Qz}{(\Ez-\FockGz)^2}\FockG_1\Chiz}\Pz
+\REz\FockG_1\Pz\FockG_1\REz\,.
\end{align}
\end{subequations}
Moreover, we find that
\begin{align}
\label{eq:expansion:energy}
\left|E-E_0-E_\mathrm{pert}\right|\ls N^{\frac32(\beta-1)}\,,
\end{align}
where
\begin{equation}\label{E_pert}
E_\mathrm{pert}:= \lr{\Chiz,\FockG_2\Chiz} + \lr{\Chiz,\FockG_1\frac{\Q_0}{E_0-\FockGz}\FockG_1\Chiz}\,.
\end{equation}
\end{proposition}

Important ingredients for this proof are estimates of the operators $\FockG_1$ and $\FockG_2$ and of the remainders, which are given in Section \ref{sec:estimates-G}. Moreover, it is crucial that 
\begin{equation}
\lr{\Chiz,(\boldKz+1)(\Np+1)^\l\Chiz}\ls 1\,,\qquad \lr{\Chi,(\boldKz+1)(\Np+1)^\l\Chi}\ls 1\,,
\end{equation}
which we prove in Lemma \ref{lem:BT:K0} and Lemma \ref{lem:kinetic:energy:excitations}, respectively.

\subsection{Expansion of the ground state}
Proposition \ref{prop:expansion:P} implies that
\begin{equation}
|\ENb-\mathcal{C}-\Ez-E_\mathrm{pert}|\ls N^{\frac32(\beta-1)}
\end{equation}
by definition \eqref{def:E_0} of $E$.
In Lemma \ref{lem:diag}, we have shown that  $\mathcal{C}+\Ez=4\pi(N-1)\aNb+E_{0,0}+E_{0,1}$ up to small error terms. Finally, in Section \ref{section:explicit calculation of E-pert}, we compute $E_{0,1}+E_\mathrm{pert}$ explicitly to obtain the expression \eqref{E_corr_def} for $E_\mathrm{corr}$. This concludes the proof of Theorem \ref{thm:energy}.\\

Another consequence of Proposition \ref{prop:expansion:P} is an expansion of the ground state wave function and of its one-body reduced density matrix. Since  \eqref{eqn:expansion:TrAP} holds in particular for any bounded operator $\FockA\in\cL(\Fp)$,  we conclude analogously to \cite[Corollary 3.4]{spectrum} that
\begin{equation}
\label{eq:estimate for the projector P} 
\Tr\left|\P-\sum_{\l=0}^2\P_\l\right|\ls N^{\frac32(\beta-1)}\,.
\end{equation}
Since $\P$ is a rank-one projector, \cite[Theorem 4]{spectrum} implies that 
\begin{equation}
\left\|\PsiN -\Psi_{N,0}-\Psi_{N,1}-\Psi_{N,2}\right\|_{L^2(\Lambda^N)} \ls N^{\frac32(\beta-1)}\,,
\end{equation}
where  
\begin{equation}\label{def:PsiNl}
\Psi_{N,\l}=\UN^*\left(\id^{\leq N}\FockT^*\Chi_\l\right)
\end{equation}
for $\Chiz$ from \eqref{def:E_0} and with
\begin{align}
\Chi_1&=\REz\FockG_1\Chiz\,,\\
\Chi_2&=\left(\frac{\Qz}{(\Ez-\FockGz)^2}\FockG_2+\REz\FockG_1\REz\FockG_1-\frac12\lr{\FockG_1\frac{\Qz}{(\Ez-\FockGz)^2}\FockG_1}\right)\Chiz\,,
\end{align}
where we used the shorthand notation
$$\lr{\FockB}:=\lr{\Chiz,\FockB\Chiz}\,$$
for operators $\FockB$ on $\Fp$.
Finally, following the proof of \cite[Corollary~1.1]{proceedings}, we obtain the estimate for the reduced density matrix, which is proven in Section \ref{sec:proof:thm:state}.

\section{Bogoliubov transformations}\label{sec:BT}

In this section we collect and prove useful properties of the quadratic transformation $\FockT$, which regularizes the excitation Hamiltonian $\FockH$, and of the transformation $\BogUz$, which diagonalizes the Bogoliubov Hamiltonian $\FockGz$.

\subsection{Properties of $\widehat{v}_N^\beta$}
As a preparation, we provide two useful estimates for the interaction potential.

\begin{lem}
\label{lem:aux:potential} Recall that $\widehat{v}_N^\beta = \widehat{v} ( \cdot / N^{\beta})$. Then
\begin{align}
\label{eq:bound-pot-1}
\sum_{p \in \Lambda_+^*} \frac{\left(\widehat{v}_N^\beta (p)\right)^2}{p^2} \ls N^{\beta}
\end{align}
 and  
\begin{align}
\label{eq:auxiliary bound for the potential}
\sup_{q \in \Lsp} \bigg\{  \sum_{r \in \Lsp, r \neq - q}
\frac{ {\widehat{v}_N^\beta(r)}}{(q+r)^2} \bigg\}
\ls  N^{\beta} .
\end{align}
\end{lem}
\begin{proof} Concerning the first bound, we find 
\begin{align}
\sum_{p \in \Lambda_+^*} \left(\widehat{v}_N^\beta (p)\right)^2 / p^2 
\ls \left( \| \hvNb  \|_{\ell^\infty}^2 \| \| \vert p \vert^{-2} \chi_{\vert p \vert \leq N^\beta} \|_{\ell^1} + \| \hvNb  \|_{\ell^\infty} \| \hvNb  \|_{\ell^2} \| \vert p \vert^{-2} \chi_{\vert p \vert > N^\beta} \|_{\ell^2} \right) \ls N^\beta\label{eq:bound-2}
\end{align}
because $\| \hvNb  \|_{\ell^\infty} \ls1$ and $\| \hvNb  \|_{\ell^2} \ls N^{3\beta/2}$.
For the second estimate, we proceed similarly and find 
\begin{align}
\sum_{r \in \Lsp, r \neq - q}
\frac{ \abs{\hvNb(r)}}{(q+r)^2}
&\ls \norm{\widehat{v}}_{\ell^{\infty}} \| \chi_{\vert p \vert \leq N^\beta } \vert p \vert^{-2} \|_{\ell^1}  + \| \widehat{v}_N^\beta \|_{\ell^2} \| \chi_{\vert p \vert \geq N^\beta } \vert p \vert^{-2} \|_{\ell^2} \,,
\end{align}
which concludes the proof. 
\end{proof}

\subsection{Quadratic transformation $\FockT$}

We recall from \eqref{eq:Bogoliubov transformation} that the Bogoliubov transformation $\FockT$ is given by
$$\FockT=\exp\left\{\frac12\sum_{p\in\Lsp}\eta_p(\ad_p\ad_{-p}-a_pa_{-p})\right\},$$
for $\eta_p$ as in \eqref{def:eta1}. In Lemma \ref{lem:HST} below we summarize some useful properties of the sequence $\eta_p$. In fact this Lemma is a modification of \cite[Lemma 14]{hainzl2020}, where the Gross-Pitaevski regime is considered (i.e. $\beta=1$). The arguments presented there easily translate to our setting $\beta \in (0,1)$. 

\begin{lem}\label{lem:HST}
It holds that $\eta_p \in \mathbb{R}$, $\eta_p  = \eta_{-p}$ for all $p \in \Lambda_+^*$ and 
\begin{align}
\vert \eta_p \vert \ls  \frac{1}{p^2}, \quad \|  \eta\|_{\l^\infty}, \; \| \eta \|_{\l^2} \ls1 \; \quad \text{and} \quad  \| p \eta \|_{\l^2} \ls N^{\beta/2}.
\end{align}
\end{lem}

\begin{proof}
We follow the lines of the proof of \cite[Lemma 14]{hainzl2020}. We multiply \eqref{def:eta} with $\eta_p$ and find, summing over $p \in \Lambda_+^*$ and using $v \geq 0$,
\begin{align}
2 \| p \eta \|^2_{\l^2} = - \sum_{p \in \Lambda_+^*} \hvNb  (p) \eta_p - \frac{1}{N}\sum_{p,q \in \Lambda_+^*} \hvNb  (p-q) \eta_p \eta_q \leq - \sum_{p \in \Lambda_+^*} \hvNb  (p) \eta_p \leq \| p \eta \|_{\ell^2} \| \hvNb  /p \|_{\ell^2} \; . \label{eq:bound-1}
\end{align}
We conclude that
\begin{align}
\| p \eta \|^2_{\l^2} 
\ls \sum_{p \in \Lambda_+^*} \left(v_N^\beta (p)\right)^2 / p^2 
\ls  N^{\beta}
\end{align}
by Lemma~\ref{lem:aux:potential}. Using once more \eqref{def:eta} we get the pointwise estimate 
\begin{align}
\vert p^2 \eta_p \vert \leq \vert \hvNb  (p) \vert + \frac{1}{2N} \left(\sum_{q \in \Lambda_+^*} \frac{\vert \hvNb  (p-q) \vert^2 }{q^2} \right)^{1/2} \| q \eta \|_{\ell^2} \ls1
\end{align}
proceeding similarly with Lemma \ref{lem:aux:potential}. Consequently,
\begin{align}
\| \eta \|_{\ell^\infty} \leq \| \eta \|_{\ell^2} \ls1 \; . 
\end{align}
\end{proof}

Lemma \ref{lem:HST} immediately implies the following bounds on $c_p=\cosh(\eta_p)$ and $s_p=\sinh(\eta_p)$:
\begin{lem}\label{lem:cp:sp}
For $p\in\Lsp$,  we have the pointwise estimates
\begin{eqnarray}
|c_p|&\ls& 1\,,\qquad\qquad\qquad\;  |c_p-1|\ls\frac{1}{|p|^4}\leq 1\,,\\
|s_p|&\ls& \frac{1}{|p|^2}\ls 1\,,\qquad
|c_ps_p-\eta_p|\ls \frac{1}{|p|^6}\ls 1\,.
\end{eqnarray}
\end{lem}

Finally, we prove that the number operator conjugated with $\FockT$ can be estimated in terms of the number operator.
\begin{lem}\label{lem:T:Number}
For any $k\in\N_0$, $\l \in \mathbb{R}$ and $\psi,xi\in\Fp$, it holds that
\begin{align}\label{T_Number_T}
\left|\lr{\psi,\FockT(\Np+1)^k\FockT^*\xi}\right|
&\leq C(k)\norm{(\Np+1)^{\frac{k}{2} + \l}\psi}\norm{(\Np+1)^{\frac{k}{2} - \l}\xi} , \\
\label{N_inverse_T}
\norm{\left( \Np + 1 \right)^{-k} \psi}
&\leq C(k) \norm{\left( \Np + 1 \right)^{-k} \FockT^* \psi} . 
\end{align}
\end{lem}
\begin{proof}
By Lemma \ref{lem:cp:sp}, we compute for $k=1$
\begin{align}
&\left|\lr{\psi,\FockT(\Np+1)\FockT^*\xi}\right|
\nonumber\\
&\quad \leq \sum_{p\in\Lsp}(|c_p|^2+|s_p|^2)\left|\lr{\left( \Np + 1 \right)^{\l} \psi,\ad_p a_p \left( \Np + 1 \right)^{-\l} \xi}\right|
\nonumber \\
&\qquad +\sum_{p\in\Lsp}|c_p||s_p|\left|\lr{\psi, \big( \left( \Np + 1 \right)^{\l} \ad_p \ad_{-p}  \left( \Np + 3 \right)^{-\l} 
+ \left( \Np + 3 \right)^{\l} a_p a_{-p}  \left( \Np + 1 \right)^{-\l} \big)\xi}\right|
\nonumber \\
&\qquad + \Bigg(\sum_{p\in\Lsp}|s_p|^2+1\Bigg) \left|\lr{\left( \Np + 1 \right)^{\l} \psi, \left( \Np + 1 \right)^{-\l} \xi}\right| \nonumber\\
&\quad \ls \norm{(\Np+1)^{1/2 + \l}\psi}\norm{(\Np+1)^{1/2 - \l}\xi}\,.
\end{align}
The case $k>1$ follows from this by induction, similarly to \cite[Lemma~4.4]{QF}. In order to prove \eqref{N_inverse_T}, note that \eqref{T_Number_T} with $\l = \frac{k}{2}$ implies
\begin{align}
\norm{\left( \Np + 1 \right)^{-k} \FockT \left( \Np + 1 \right)^{k} \FockT^* \psi }^2
&\leq \scp{\left( \Np + 1 \right)^{-2k} \FockT \left( \Np + 1 \right)^{k} \FockT^* \psi }{ \FockT \left( \Np + 1 \right)^{k} \FockT^* \psi }
\nonumber \\
&\ls \norm{\left( \Np + 1 \right)^{-k} \FockT \left( \Np + 1 \right)^{k} \FockT^* \psi } \norm{\psi},
\end{align}
showing 
\begin{align}
\norm{\left( \Np + 1 \right)^{-k} \FockT \left( \Np + 1 \right)^{k} \FockT^*}_{\op} &\ls 1 .
\end{align}
Hence,
\begin{align}
\norm{\left( \Np + 1 \right)^{-k} \psi}
&= \norm{\left( \Np + 1 \right)^{-k} \FockT \left( \Np + 1 \right)^{k} \FockT^* \FockT \left( \Np + 1 \right)^{-k} \FockT^* \psi}
\nonumber \\
&\ls \norm{\FockT \left( \Np + 1 \right)^{-k} \FockT^* \psi}
\nonumber \\
&= \norm{ \left( \Np + 1 \right)^{-k} \FockT^* \psi} .
\end{align}
\end{proof}

\subsection{Quadratic transformation $\BogUz$}

Recall from \eqref{G0} that the Bogoliubov Hamiltonian is given by
$$\FockGz=\sum_{p\in\Lsp} F_p\ad_pa_p+\frac12\sum_{p\in\Lsp}G_p(\ad_p\ad_{-p}+a_pa_{-p})\,,$$
where the coefficients $F_p$ and $G_p$ are defined in \eqref{def:F,G}.
From similar arguments as given in \cite[Lemma 5.1]{boccato2017_2}, it follows that the operators $F$ and $G$ satisfy the following properties:

\begin{lem}
\label{lemma:propF,G}
For all $p \in \Lambda_+^*$ and $N$ large enough,
\begin{align}
\label{eq:lemmaFG}
p^2 /2 \leq F_p \ls (1 + p^2) , \quad \vert G_p \vert \ls p^{-2}, \quad  \vert G_p \vert / F_p  \ls p^{-4} , \quad  \vert G_p \vert / F_p \leq 1/2  \; . 
\end{align}
\end{lem}

\begin{proof}
The lemma can be proven analogously to \cite[Lemma 5.1]{boccato2017_2}. Note that in \cite{boccato2017_2} the sequence $\eta_p$ was chosen differently. However, as $\eta_p$ defined in \eqref{def:eta} satifies analogous estimates as the sequence from \cite{boccato2017_2}, the proof of \cite[Lemma 5.1]{boccato2017_2} applies here, too. We will briefly sketch the proof. First, note that
\begin{align}
\Big\vert \frac{1}{N} \sum_{q \in \Lambda_+^*} \widehat{v}_N^{\beta} (p-q) s_q c_q \Big \vert \leq \frac{C}{N} \sum_{q \in \Lambda_+^*} \frac{\vert\widehat{v}_N^{\beta} (p-q) \vert}{q^2}  \leq C N^{\beta -1}
\end{align}
by Lemmas~\ref{lem:aux:potential} and \ref{lem:HST}. 
Since $\hat{v} \geq 0$ we arrive by definition \eqref{def:F,G} of $F_p$ (note that $c_p^2 + s_p^2 \geq 1$) at 
\begin{align}
F_p \geq  p^2  - CN^{\beta -1} \geq \frac{1}{2} p^2
\end{align}
for $N$ large enough. The upper bound follows since $\widehat{v}_N^\beta, s_p, c_p$ are bounded in $\ell^\infty ( \Lambda_+^*)$ uniformly in $N$, thus 
\begin{align}
F_p \leq C ( 1 + p^2) \; . 
\end{align} 
To prove the second bound on $G_p$, we write 
\begin{align}
\label{eq:tidle-G}
G_p =2  p^2 \eta_p + \widehat{v}_N^\beta (p) + \frac{1}{N} \sum_{q \in \Lambda_+^*} \widehat{v}_N^\beta (p-q) \eta_q + \widetilde{G}_p
\end{align}
with $\vert \widetilde{G}_p \vert \leq C p^{-2}$ following from  $|s_pc_p-\eta_p|\ls p^{-6}$ by Lemma \ref{lem:HST} and similarly $ \vert (s_p + c_p)^2 - 1 \vert \ls p^{-2}$, $\vert s_p^2 \vert \ls p^{-4}, \vert c_p^2 - 1 \vert \ls p^{-4}$. The remaining three terms of the r.h.s. of \eqref{eq:tidle-G} vanish by \eqref{def:eta} and thus we arrive at \eqref{eq:lemmaFG}. 
\end{proof}

The Bogoliubov Hamiltonian is diagonalized by the quadratic transformation
$$\BogUz=\exp\left\{\frac12\sum_{p\in\Lsp}\tau_p(\ad_p\ad_{-p}-a_pa_{-p})\right\}$$
defined in \eqref{def:Ttau}, where
$\tau_p$ is given by $
\tanh( 2 \tau_p ) = - G_p / F_p$ for all $p \in \Lambda_+^*$
by \eqref{def:tau}. Note that $\tau_p=\tau_{-p}$. Equivalently, we can write 
\begin{align}
\tau_p = \frac{1}{4} \ln \frac{1-G_p/F_p}{1+G_p/F_p}   \quad \text{for all} \quad p \in \Lambda_+^* \,,
\end{align}
hence Lemma \ref{lemma:propF,G} yields the estimate
\begin{align}
\label{eq:bound-tau}
\vert \tau_p \vert \ls \vert G_p \vert / F_p \ls p^{-4} \quad \text{for all} \quad p \in \Lambda_+^* \; .
\end{align}
We now use this estimate to show that $\BogUz$ approximately preserves the number of particles and the kinetic energy:

\begin{lem}\label{lem:BT:K0}
For $\l \in \mathbb{N}_0$ it holds that
\begin{align}\label{eqn:BT:number:1}
\BogUz(\Np+1)^\l\BogUz^*&\leq C(\l) (\Np+1)^\l\,,\qquad \BogUz^*(\Np+1)^\l\BogUz\leq C(\l) (\Np+1)^\l
\end{align}
as well as
\begin{align}
\begin{split}
\label{lem:BT:number}
\BogUz^*(\Np+1)^{\l} (\boldKz+1)\BogUz
&\leq C(\l) (\Np+1)^{\l} (\boldKz+1) \,,\\
\BogUz(\Np+1)^{\l} (\boldKz+1)\BogUz^*
&\leq C(\l) (\Np+1)^{\l} (\boldKz+1) .\end{split}
\end{align}
and
\begin{align}
\begin{split}
\label{lem:BT:number:inverse}
\BogUz^*(\Np+1)^{\l} \frac{1}{\boldKz+1}\BogUz
&\leq C(\l) (\Np+1)^{\l} \frac{1}{\boldKz+1} \,,\\
\BogUz(\Np+1)^{\l} \frac{1}{\boldKz+1}\BogUz^*
&\leq C(\l) (\Np+1)^{\l} \frac{1}{\boldKz+1}\,.
\end{split}
\end{align}
\end{lem}

\begin{proof}
\noindent\textbf{Proof of \eqref{eqn:BT:number:1}.}
We compute 
\begin{align}
\BogUz\Np\BogUz^*&= \sum_{p\in\Lsp}(\sinh( \tau_p)^2+\cosh ( \tau_p)^2)\ad_pa_p +\sum_{p\in\Lsp}\sinh(\tau_p)^2\\
&\quad+\sum_{p\in\Lsp}\cosh(\tau_p)\sinh(\tau_p)(\ad_p\ad_{-p}+a_pa_{-p})\,.
\end{align}
By \eqref{eq:bound-tau}, it follows that $|\cosh(\tau_p)|\ls1$
and $|\sinh(\tau_p)|\ls|p|^{-4}$. Consequently \eqref{eqn:BT:number:1} follows analogously to the proof of Lemma \ref{lem:T:Number}.

\paragraph{Proof of  \eqref{lem:BT:number}.} Let us define the symmetric operator
$$\mathbb{A} = \frac{\i}{2} \sum_{p \in \Lambda_+^*} \tau_p (\ad_p \ad_{-p} - a_p a_{-p})\,,\qquad\mathbb{U}_\tau(\theta) = \e^{- \i \mathbb{A} \theta}$$
such that $\mathbb{U}_\tau(0) = \mathbbm{1}$ and $\mathbb{U}_\tau(1) = \mathbb{U}_\tau$. Next, we compute
\begin{subequations}
\begin{align}
&\frac{\d}{\d \theta} \scp{\psi}{\mathbb{U}_\tau^*(\theta) (\Np+1)^\l(\boldKz+1) \mathbb{U}_\tau(\theta) \psi}
\nonumber \\
\label{eq: action of of BT moments of number operator and kinetic energy estimate  1}
&\quad = \i \scp{\psi}{\mathbb{U}_\tau^*(\theta) (\Np+1)^\l  \left[ \mathbb{A} , \boldKz \right] \mathbb{U}_\tau(\theta) \psi}
\\
\label{eq: action of of BT moments of number operator and kinetic energy estimate  2}
&\qquad + \i \scp{\psi}{\mathbb{U}_\tau^*(\theta) \left[ \mathbb{A} , (\Np+1)^\l \right]  (\boldKz+1) \mathbb{U}_\tau(\theta) \psi} .
\end{align}
\end{subequations}
Note that the second term is zero if $\l =0$ and that
\begin{align}
\left[ \mathbb{A} , \boldKz \right] 
&= \i \sum_{p \in \Lambda_+^*} p^2 \tau_p \left( \ad_p \ad_{-p} + a_p a_{-p} \right)\,.
\end{align}
Together with  $\abs{\tau_p} \ls p^{-4}$ (see \eqref{eq:bound-tau}), the shifting property of the number operator and the Cauchy--Schwarz inequality, this leads to
\begin{align}
|\eqref{eq: action of of BT moments of number operator and kinetic energy estimate  1}
|&\leq\left|\sum_{p \in \Lambda_+^*} p^2 \tau_p \scp{\psi}{\mathbb{U}_\tau^*(\theta) (\Np+1)^{\frac{\l}{2}}  \ad_p \ad_{-p}  (\Np+ 3)^{\frac{\l}{2}} \mathbb{U}_\tau(\theta) \psi}\right|
\nonumber \\
&\quad  +\left| \sum_{p \in \Lambda_+^*} p^2 \tau_p \scp{\psi}{\mathbb{U}_\tau^*(\theta) (\Np+1)^{\frac{\l}{2}}  a_p a_{-p}  (\Np - 1)^{\frac{\l}{2}} \mathbb{U}_\tau(\theta) \psi}\right|
\nonumber \\
&\leq   \sum_{p \in \Lambda_+^*} p^2 \abs{\tau_p} \norm{a_p(\Np+1)^{\frac{\l}{2}} \mathbb{U}_\tau(\theta) \psi} \norm{ \ad_{-p}  (\Np+ 3)^{\frac{\l}{2}} \mathbb{U}_\tau(\theta) \psi}
\nonumber \\
&\quad  + \sum_{p \in \Lambda_+^*} p^2 \abs{\tau_p} \norm{ \ad_p (\Np+1)^{\frac{\l}{2}}\mathbb{U}_\tau(\theta)\psi} \norm{ a_{-p}  (\Np - 1)^{\frac{\l}{2}} \mathbb{U}_\tau(\theta) \psi}
\nonumber \\
&\ls
\left( \sum_{p \in \Lambda_+^*} p^2  \norm{a_p(\Np+1)^{\frac{\l}{2}} \mathbb{U}_\tau(\theta) \psi}^2 \right)^{1/2}
\left( \sum_{p \in \Lambda_+^*} p^2 \abs{\tau_p}^2  \norm{ \ad_{-p}  (\Np+ 1)^{\frac{\l}{2}} \mathbb{U}_\tau(\theta) \psi}^2 \right)^{1/2}\nonumber\\
&\ls\norm{\left( \boldKz + 1 \right)^{\frac{1}{2}} (\Np+1)^{\frac{\l}{2}} \mathbb{U}_\tau(\theta) \psi}^2\,.
\end{align}
Using the shifting property of the number operator again, we get
\begin{align}
\left[ \mathbb{A} , \left( \Np + 1 \right)^{\l} \right]
&= - \frac{\i}{2}  \sum_{p \in \Lambda_+^*} \tau_p
\bigg( \left(  \left( \Np + 1 \right)^\l - \left( \Np - 1 \right)^\l \right) \left( \Np - 1 \right)^{-\frac{\l}{2}} \ad_p \ad_{-p} \left( \Np + 1 \right)^{\frac{\l}{2}}
\nonumber \\
&\qquad \qquad  + \left( \left( \Np + 3 \right)^\l - \left( \Np + 1 \right)^\l \right)  \left( \Np + 3 \right)^{-\frac{ \l}{2}}
 a_p a_{-p}  \left( \Np + 1 \right)^{\frac{\l}{2}}
\bigg) \,.
\end{align}
Hence, 
\begin{align}
&|\eqref{eq: action of of BT moments of number operator and kinetic energy estimate  2}|\nonumber\\
&\leq\left|\scp{\psi}{\mathbb{U}_\tau^*(\theta) \left[ \mathbb{A} , (\Np+1)^\l \right]  \mathbb{U}_\tau(\theta) \psi}\right|
\nonumber \\
&\quad + \left|\scp{\psi}{\mathbb{U}_\tau^*(\theta) \left[ \mathbb{A} , (\Np+1)^\l \right]  \left( \Np + 1 \right)^{- \frac{\l}{2}}  \boldKz \left( \Np + 1 \right)^{\frac{\l}{2}}  \mathbb{U}_\tau(\theta)\psi}\right|
\nonumber \\
&\ls \norm{\left( \Np + 1 \right)^{\frac{\l}{2}}  \mathbb{U}_\tau(\theta)\psi}
\sum_{p \in \Lambda_+^*} \abs{\tau_p}
\norm{a_p a_{-p} \left(  \left( \Np + 1 \right)^\l - \left( \Np - 1 \right)^\l \right) \left( \Np - 1 \right)^{- \frac{\l}{2}} \mathbb{U}_\tau(\theta) \psi}
\nonumber \\
&\quad + \norm{\left( \Np + 1 \right)^{\frac{\l}{2}} \mathbb{U}_\tau(\theta) \psi}
\sum_{p \in \Lambda_+^*} \abs{\tau_p}
\norm{\ad_p \ad_{-p} \left(  \left( \Np + 3 \right)^\l - \left( \Np + 1 \right)^\l \right) \left( \Np +3 \right)^{- \frac{\l}{2}} \mathbb{U}_\tau(\theta) \psi}
\nonumber \\
&\quad + \sum_{k \in \Lambda_+^*} k^2 \norm{a_k \left( \Np + 1 \right)^{\frac{\l}{2}} \mathbb{U}_\tau(\theta) \psi}
\nonumber \\
&\qquad  \times 
\sum_{p \in \Lambda_+^*} \abs{\tau_p}
\norm{a_k a_p a_{-p} \bigg( \left(  \left( \Np + 1 \right)^\l - \left( \Np - 1 \right)^\l \right) \left( \Np - 1 \right)^{- \frac{\l}{2}}  \mathbb{U}_\tau(\theta)\psi}
\nonumber \\
&\quad + \sum_{k \in \Lambda_+^*} k^2 \norm{a_k \left( \Np + 1 \right)^{\frac{\l}{2}} \mathbb{U}_\tau(\theta) \psi}
\nonumber \\
&\qquad  \times 
\sum_{p \in \Lambda_+^*} \abs{\tau_p}
\norm{a_k \ad_p \ad_{-p} \bigg( \left(  \left( \Np + 3 \right)^\l - \left( \Np + 1 \right)^\l \right) \left( \Np +3 \right)^{- \frac{\l}{2}}  \mathbb{U}_\tau(\theta) \psi}\,.
\end{align}
Using $\left( \Np + 3 \right)^{\l} - \left( \Np + 1 \right)^{\l}  \ls \l \left( \Np + 3 \right)^{\l -1 }$ and $\abs{\tau_p} \ls p^{-4}$, it is easily seen that the first two lines can be bounded by $ \norm{\left( \Np + 1 \right)^{\frac{\l}{2}} \BogUz(\theta) \psi}^2$. Estimating the remaining terms then by the Cauchy--Schwarz inequality and Young's inequality for products leads to
\begin{align}
\eqref{eq: action of of BT moments of number operator and kinetic energy estimate  2}
&\ls   \norm{\left( \boldKz + 1 \right)^{1/2} \left( \Np + 1 \right)^{\frac{\l}{2}} \mathbb{U}_\tau(\theta) \psi}^2
\nonumber \\
\label{eq: action of of BT moments of number operator and kinetic energy estimate  2a}
&\quad +  \sum_{k \in \Lambda_+^*} k^2 \abs{
\sum_{p \in \Lambda_+^*} \abs{\tau_p}
\norm{a_k a_p a_{-p} \bigg( \left(  \left( \Np + 1 \right)^\l - \left( \Np - 1 \right)^\l \right) \left( \Np - 1 \right)^{- \frac{\l}{2}} \mathbb{U}_\tau(\theta) \psi} }^2
\\
\label{eq: action of of BT moments of number operator and kinetic energy estimate  2b}
&\quad + \sum_{k \in \Lambda_+^*} k^2 
\abs{\sum_{p \in \Lambda_+^*} \abs{\tau_p}
\norm{a_k \ad_p \ad_{-p} \bigg( \left(  \left( \Np + 3 \right)^\l - \left( \Np + 1 \right)^\l \right) \left( \Np +3 \right)^{- \frac{\l}{2}} \mathbb{U}_\tau(\theta) \psi} }^2 .
\end{align}
Similarly to above, we obtain
\begin{align}
\eqref{eq: action of of BT moments of number operator and kinetic energy estimate  2a}
&\ls \sum_{k \in \Lambda_+^*} k^2 \norm{a_k  \left( \Np + 1 \right)^{\frac{\l}{2}} \mathbb{U}_\tau(\theta) \psi}^2
=  \norm{\boldKz^{\frac{1}{2}} \left( \Np + 1 \right)^{\frac{\l}{2}} \mathbb{U}_\tau(\theta) \psi}^2,
\end{align}
and, using that $\left[ a_k , \ad_p \ad_{-p} \right] = (\delta_{k,p}+\delta_{k,-p}) \ad_k$, we estimate
\begin{align}
\eqref{eq: action of of BT moments of number operator and kinetic energy estimate  2b}
&\ls   \sum_{k \in \Lambda_+^*} k^2 
\abs{\sum_{p \in \Lambda_+^*} \abs{\tau_p}
\norm{ \ad_p \ad_{-p} a_k \bigg( \left(  \left( \Np + 3 \right)^\l - \left( \Np + 1 \right)^\l \right) \left( \Np +3 \right)^{- \frac{\l}{2}} \mathbb{U}_\tau(\theta) \psi} }^2
\nonumber \\
&\quad + 
\sum_{k \in \Lambda_+^*} k^2 
\abs{\tau_k}^2
\norm{ \ad_k  \bigg( \left(  \left( \Np + 3 \right)^\l - \left( \Np + 1 \right)^\l \right) \left( \Np +3 \right)^{- \frac{\l}{2}} \mathbb{U}_\tau(\theta) \psi}^2
\nonumber \\
&\ls  \norm{\boldKz^{\frac{1}{2}} \left( \Np + 1 \right)^{\frac{\l}{2}} \mathbb{U}_\tau(\theta) \psi}^2 .
\end{align}
Hence,
\begin{align}
\eqref{eq: action of of BT moments of number operator and kinetic energy estimate  2}
&\ls  \norm{\left( \boldKz + 1 \right)^{1/2} \left( \Np + 1 \right)^{\frac{\l}{2}} \mathbb{U}_\tau(\theta) \psi}^2\,.
\end{align}
Collecting the estimates, we find that
\begin{align}
\frac{\d}{\d\theta}\lr{\psi,\BogUz^*(\theta)(\Np+1)^\l(\boldKz+1)\BogUz(\theta)\psi} 
\leq C(\l) \lr{\psi,\BogUz^*(\theta)(\Np+1)^\l(\boldKz+1)\BogUz(\theta)\psi}\,,
\end{align}
hence Gronwall's lemma lead to
\begin{align}
\scp{\psi}{\mathbb{U}_\tau^*(1) (\Np+1)^\l(\boldKz+1) \mathbb{U}_\tau(1) \psi} 
&\leq e^{C (\l)} \scp{\psi}{\mathbb{U}_\tau^*(0) (\Np+1)^\l(\boldKz+1) \mathbb{U}_\tau(0) \psi}  .
\end{align}
Since $\mathbb{U}_\tau^*(1) = \mathbb{U}_\tau^*$ and $\mathbb{U}_\tau^*(0) = 1$ this shows the first inequality of \eqref{lem:BT:number}. The second one follows analogously.

\paragraph{Proof of \eqref{lem:BT:number:inverse}.} We prove this statement via a similar Gronwall argument. Analogously to above, we compute
\begin{subequations}
\begin{align}
&\frac{\d}{\d \theta} \scp{\psi}{\mathbb{U}_\tau^*(\theta) (\Np+1)^\l\frac{1}{\boldKz+1} \mathbb{U}_\tau(\theta) \psi}
\nonumber \\
\label{eqn:inverse:1}
&\quad = \i \scp{\psi}{\mathbb{U}_\tau^*(\theta) (\Np+1)^\l  \left[ \mathbb{A} , \frac{1}{\boldKz+1}\right] \mathbb{U}_\tau(\theta) \psi}
\\
\label{eqn:inverse:2}
&\qquad + \i \scp{\psi}{\mathbb{U}_\tau^*(\theta) \left[ \mathbb{A} , (\Np+1)^\l \right]  \frac{1}{\boldKz+1} \mathbb{U}_\tau(\theta) \psi} .
\end{align}
\end{subequations}
For \eqref{eqn:inverse:1}, we observe that
\begin{equation}
\left[\FockA,\frac{1}{\boldKz+1}\right]=\i\sum_{p\in\Lsp}\frac{1}{\boldKz+1}(\ad_p\ad_{-p}+a_pa_{-p})\frac{1}{\boldKz+1}\,.
\end{equation}
Analogously to the estimate of \eqref{eq: action of of BT moments of number operator and kinetic energy estimate  1}, this yields
\begin{equation}
|\eqref{eqn:inverse:1}|
\ls\lr{\psi,\BogUz^*(\theta)(\Np+1)^\l\frac{1}{\boldKz+1}\BogUz(\theta)\psi}
\end{equation}
because, for example,
\begin{align}
&\sum_{p\in\Lsp}p^2|\tau_p|\norm{a_p(\Np+1)^\l\frac{1}{\boldKz+1}\BogUz(\theta)\psi}\norm{\ad_{-p}(\Np+3)^\frac{\l}{2}\frac{1}{\boldKz+1}\BogUz(\theta)\psi}\nonumber\\
&\quad\ls \left(\sum_{p\in\Lsp}p^2\norm{a_p(\Np+1)^\frac{\l}{2}\frac{1}{\boldKz+1}\BogUz(\theta)\psi}^2\right)^\frac12 \left(\sum_{p\in\Lsp}p^2|\tau_p|^2\right)^\frac12 \norm{(\Np+3)^\frac{\l+1}{2}\frac{1}{\boldKz+1}\BogUz(\theta)\psi}\nonumber\\
&\ls\norm{(\Np+1)^\frac{\l}{2}\left(\frac{1}{\boldKz+1}\right)^\frac12\BogUz(\theta)\psi}^2\,.
\end{align}
For the second contribution \eqref{eqn:inverse:2}, recall that $(\boldKz+1)^{-\frac12} =\pi^{-1}\int_0^\infty s^{-\frac12}(\boldKz+s+1)^{-1}\ds$, which yields
\begin{equation}
\left[\left(\frac{1}{\boldKz+1}\right)^\frac12,a_pa_{-p}\right]=-\frac{2p^2}{\pi}\int_0^\infty\frac{1}{\sqrt{s}}\frac{1}{\boldKz+s+1}a_pa_{-p}\frac{1}{\boldKz+s+1}\ds\,.
\end{equation}
Consequently, abbreviating 
$$F_1(\Np):=((\Np+1)^\l-(\Np-1)^\l)(\Np+1)^{-\frac{\l}{2}}\,,\quad
F_3(\Np):=((\Np+3)^\l-(\Np+1)^\l)(\Np+3)^{-\frac{\l}{2}}\,,$$
we find
\begin{subequations}
\begin{align}
|\eqref{eqn:inverse:2}|
&\ls\sum_{p\in\Lsp}|\tau_p|\Big\|a_pa_{-p}\left(\frac{1}{\boldKz+1}\right)^\frac12 F_1(\Np)\BogUz(\theta)\psi\Big\|\nonumber\\
&\qquad\qquad\qquad\times\Big\|\left(\frac{1}{\boldKz+1}\right)^\frac12(\Np+1)^\frac{\l}{2}\BogUz(\theta)\psi\Big\|\label{eqn:inverse:2:1}\\
&\quad+\sum_{p\in\Lsp}|\tau_p|\Big\|\left(\frac{1}{\boldKz+1}\right)^\frac12(\Np+1) F_3(\Np)\BogUz(\theta)\psi\Big\|\nonumber\\
&\qquad\qquad\qquad\times\Big\|(\Np+1)^{-1}a_pa_{-p}\left(\frac{1}{\boldKz+1}\right)^\frac12(\Np+1)^\frac{\l}{2}\BogUz(\theta)\psi\Big\|\label{eqn:inverse:2:2}\\
&\quad+\int_0^\infty\ds\frac{1}{\sqrt{s}}\sum_{p\in\Lsp}|\tau_p|p^2
\Big\|a_pa_{-p}\frac{1}{\boldKz+s+1}F_1(\Np)\BogUz(\theta)\psi\Big\|\nonumber\\
&\qquad\qquad\qquad\times
\Big\|\frac{1}{\boldKz+s+1}\left(\frac{1}{\boldKz+1}\right)^\frac12(\Np+1)^\frac{\l}{2}\BogUz(\theta)\psi\Big\|\label{eqn:inverse:2:3}\\
&\quad+\int_0^\infty\ds\frac{1}{\sqrt{s}}\sum_{p\in\Lsp}|\tau_p|p^2
\Big\|\frac{1}{\boldKz+s+1}F_3(\Np)\BogUz(\theta)\psi\Big\|\nonumber\\
&\qquad\qquad\qquad\times
\Big\|a_pa_{-p}\frac{1}{\boldKz+s+1}\left(\frac{1}{\boldKz+1}\right)^\frac12(\Np+1)^\frac{\l}{2}\BogUz(\theta)\psi\Big\|\,.\label{eqn:inverse:2:4}
\end{align}
\end{subequations}
It is easily seen that 
\begin{equation}
\eqref{eqn:inverse:2:1}+\eqref{eqn:inverse:2:2}\ls \Big\|\left(\frac{1}{\boldKz+1}\right)^\frac12(\Np+1)^\frac{\l}{2}\BogUz(\theta)\psi\Big\|^2\,.
\end{equation}
For the third term, we  use the estimates $\boldKz+s+1\geq\max\{\boldKz+1,s+1\}$ to compute
\begin{align}
\eqref{eqn:inverse:2:3}
&\ls \int_0^\infty\ds\frac{1}{\sqrt{s}(s+1)}\left(\sum_{p\in\Lsp}\tau_p^2 p^2\right)^\frac12
\left(\sum_{p\in\Lsp}p^2\Big\|a_p\frac{1}{\boldKz+s+1}\Np^\frac12 F_1(\Np)\BogUz(\theta)\psi\Big\|^2\right)^\frac12\nonumber\\
&\qquad\times
\Big\|\left(\frac{1}{\boldKz+1}\right)^\frac12(\Np+1)^\frac{\l}{2}\BogUz(\theta)\psi\Big\|\nonumber\\
&\ls\int_0^\infty\ds\frac{1}{\sqrt{s}(s+1)}
\Big\|\boldKz^\frac12\frac{1}{\boldKz+s+1}\Np^\frac12 F_1(\Np)\BogUz(\theta)\psi\Big\|
\Big\|\left(\frac{1}{\boldKz+1}\right)^\frac12(\Np+1)^\frac{\l}{2}\BogUz(\theta)\psi\Big\|\nonumber\\
&\ls\Big\|\left(\frac{1}{\boldKz+1}\right)^\frac12(\Np+1)^\frac{\l}{2}\BogUz(\theta)\psi\Big\|^2\,,
\end{align}
and the estimate of \eqref{eqn:inverse:2:4} works similarly. This concludes the Gronwall argument.
\end{proof}

\section{Diagonalization of $\FockGz$}
\label{sec:proof:diag}

In this section we prove Lemma \ref{lem:diag}.
\begin{proof}
For this proof we use ideas from \cite[Proof of Lemma 5.3]{boccato2017_2}, where an expansion up to $o(1)$ is derived. Here we generalize ideas from \cite{boccato2017_2} and provide an expansion up to higher order that is $o(N^{\beta -1})$. 
First we observe from \eqref{def:F,G} and \eqref{eq:G0-diag} and the properties of the hyperbolic functions that 
\begin{align}
\mathcal{C} - \frac{1}{2} \sum_{p \in \Lambda_+^*} F_p  =&  \frac{N-1}{2} \widehat{v} (0) - \frac{1}{2} \sum_{p \in \Lambda_+^*} \Big( p^2 + \widehat{v}_N^\beta (p) 
+ \frac{1}{N} (v_N^\beta * sc )_p c_ps_p\Big) \notag \\
&- \frac{1}{N} \sum_{p \in \Lambda_+^*} v_N^\beta (p) \Big( \frac{s_pc_p}{2} + s_p^3c_p\Big)  \; 
\end{align}
for $(v_N^\beta * sc )_p $ as in \eqref{def:convolution}.
Since $\vert s_p^3 c_p \vert  \ls \vert p \vert^{-6} $ and $\vert c_p s_p - \eta_p \vert \ls \vert p \vert^{-6}$ by Lemma \ref{lem:cp:sp}, we get 
\begin{align}
\mathcal{C} - \frac{1}{2} \sum_{p \in \Lambda_+^*} F_p  =&  \frac{N-1}{2} \widehat{v} (0) - \frac{1}{2} \sum_{p \in \Lambda_+^*} \Big( p^2 + \widehat{v}_N^\beta (p) + \frac{1}{N} (v_N^\beta * sc )_p c_ps_p\Big) \notag \\
&- \frac{1}{2N} \sum_{p \in \Lambda_+^*} v_N^\beta (p) \eta_p  + \mathcal{O}(N^{-1} ) \; .
\end{align}
By definition of the scattering length on  the box in \eqref{def:box:scattering:length}, we can write the term in the last line as 
\begin{align}
\frac{1}{N} \sum_{p \in \Lambda_+^*} v_N^\beta (p) \eta_p = 8\pi\mathfrak{a}_N^\beta-\widehat{v}(0) 
\end{align}
so that 
\begin{align}
\label{eq:C-F}
\mathcal{C} - \frac{1}{2} \sum_{p \in \Lambda_+^*} F_p  =&  \frac{N}{2} \widehat{v} (0) - 4 \pi \mathfrak{a}_N^\beta - \frac{1}{2} \sum_{p \in \Lambda_+^*} \Big( p^2 + \widehat{v}_N^\beta (p) + \frac{1}{N} (v_N^\beta * sc )_p c_ps_p\Big) +  \mathcal{O}(N^{-1} ) \; .
\end{align}  
Furthermore,
\begin{align}
\label{eq:fquare-gsquare}
F_p^2 - G_p^2 = \vert p \vert^4 + 2 p^2 \widehat{v}_N^\beta (p) + A_p
\end{align}
with 
\begin{align}
\label{def:Ap}
A_p &= - \frac{1}{N}
\Big[ 2  \widehat{v}_N^\beta (p) (v_N^\beta* sc)_p
+ \frac{1}{N} (v_N^\beta* sc)_p^2 \Big] .
\end{align}
Note that $\| (v_N^\beta*\eta) \|_{\l^\infty} \ls N^{\beta}$ by Lemma \ref{lem:cp:sp} and thus 
\begin{align}
\vert A_p \vert \ls N^{\beta-1} \; .\label{eq:boundAp}
\end{align}
Note that $\widehat{v}(p) \geq 0$ for all $p \in \Lsp$. Since $A_p$ is of lower order we  obtain that $\vert p \vert^4 + 2 p^2 \widehat{v}_N^\beta (p) + A_p \geq 0$ and $\vert p \vert^4 + 2 p^2 \widehat{v}_N^\beta (p) \geq 0$ for $N$ large enough. 
Therefore, recalling that we are interested in the energy expansion for $\smallO{ N^{(\beta -1)}}$, we expand the square root of $F_p^2 - G_p^2$, i.e., the r.h.s. of \eqref{eq:fquare-gsquare}, using the identity $\sqrt{a+b}=\sqrt{b}+\frac{a}{2\sqrt{b}}-\frac{a^2}{2\sqrt{b}(\sqrt{a+b}+\sqrt{b})^2}$, we find 
\begin{align}
\sqrt{F_p^2 - G_p^2 } =& \sqrt{\vert p \vert^4 + 2 p^2 \widehat{v}_N^\beta (p)} + \frac{A_p}{2\sqrt{\vert p \vert^4 + 2 p^2 \widehat{v}_N^\beta (p)}}  \notag \\
&-
 \frac{ A_p^2 }{2 \Big( \sqrt{\vert p \vert^4 + 2 p^2 \widehat{v}_N^\beta (p)}+ \sqrt{\vert p \vert^4 
+ 2 p^2 \widehat{v}_N^\beta (p)+A_p}\Big)^2 \sqrt{\vert p \vert^4 + 2 p^2 \widehat{v}_N^\beta (p)}} \; . \label{eq:exp-squareroot}
\end{align}
We will show that the last term is $\mathcal{O}( N^{2(\beta-1)})$. To this end we observe that 
\begin{align}
C p^2 \leq \sqrt{\vert p \vert^4 + 2 p^2 \widehat{v}_N^\beta (p)}+ \sqrt{\vert p \vert^4 + 2 p^2 \widehat{v}_N^\beta (p)+A_p} \leq C p^2 \Big( 1 + C \Big( \frac{A_p}{\vert p \vert^4} + \frac{\widehat{v}_N^\beta (p)}{p^2} \Big) \Big) 
\end{align}
for sufficiently large $N$ and thus 
\begin{align}
\frac{A_p^2}{\vert p\vert^6} &\Big( 1 - C \Big( \frac{A_p}{\vert p \vert^4} + \frac{\widehat{v}_N^\beta (p)}{p^2} \Big) \Big) \notag \\
&\leq   \frac{ A_p^2 } {\Big( \sqrt{\vert p \vert^4 + 2 p^2 \widehat{v}_N^\beta (p)}+ \sqrt{\vert p \vert^4 + 2 p^2 \widehat{v}_N^\beta (p)+A_p}\Big)^2 \sqrt{\vert p \vert^4 + 2 p^2 \widehat{v}_N^\beta (p)}} 
\leq \frac{A_p^2}{\vert p\vert^6} \; . 
\end{align}
With \eqref{eq:boundAp} we find 
\begin{align}
\sum_{p \in \Lambda_+^*} \frac{A_p^2}{\vert p \vert^6} \leq CN^{2(\beta -1)}\,,
\end{align}
which leads with \eqref{eq:exp-squareroot} to 
\begin{align}\label{eqn:F_p^2-G_p^2}
\sqrt{F_p^2 - G_p^2 } =& \sqrt{\vert p \vert^4 + 2 p^2 \widehat{v}_N^\beta (p)} + \frac{A_p}{2\sqrt{\vert p \vert^4 + 2 p^2 \widehat{v}_N^\beta (p)}}  + \mathcal{O}( N^{2( \beta-1)} ) \; . 
\end{align}
Next, we further expand the second term of the r.h.s.\ of \eqref{eqn:F_p^2-G_p^2} and write  
\begin{align}
\frac{1}{\sqrt{\vert p\vert^4 + 2 p^2 \widehat{v}_N^\beta (p)}} = \frac{1}{p^2} - \frac{2 \widehat{v}_N^\beta (p)}{\sqrt{\vert p \vert^4 + 2 p^2 \widehat{v}_N^\beta(p) }  \left(  p^2 + \sqrt{\vert  p \vert^4 + 2 p^2 \widehat{v}_N^\beta(p) } \right) } \,,
\end{align}
which leads with the definition
\begin{align}
\label{def:Bp}
B_p := \frac{ A_p \widehat{v}_N^\beta (p)}{\sqrt{\vert p \vert^4 + 2 p^2 \widehat{v}_N^\beta(p) }  \left(  p^2 + \sqrt{\vert  p \vert^4 + 2 p^2 \widehat{v}_N^\beta(p) } \right) } 
\end{align}
to   
\begin{align}
\sqrt{F_p^2 - G_p^2 } = \sqrt{\vert p \vert^4 + 2 p^2 \widehat{v}_N^\beta (p)} + \frac{A_p}{2p^2} - B_p  + \mathcal{O}( N^{2( \beta -1)} )  \; . 
\end{align}
We use this expansion together with \eqref{eq:C-F} to derive an expansion 
\begin{align}
\mathcal{C} & + \frac{1}{2} \sum_{p \in \Lambda_+^*} \big[ - F_p + \sqrt{F_p^2 - G_p^2} \big] \notag \\
=& \frac{N\widehat{v} (0)}{2} - 4 \pi \mathfrak{a}_N^\beta + \frac{1}{2} \sum_{p \in \Lambda_+^*} \big[ -p^2 - v_N^\beta (p) + \sqrt{\vert p \vert^4 + 2p^2 \hat{v}_N^\beta (p)} \big] \notag \\
&+ \sum_{p \in \Lambda_+^*} \Big[ \frac{A_p}{4p^2} - \frac{1}{2N} (\hat{v}_N^\beta *sc)_p c_ps_p \Big] 
- \frac{1}{2}\sum_{p \in \Lambda_+^*}  B_p  + \mathcal{O}(N^{2(\beta -1)}) . 
\end{align}
For the two terms of the first line we proceed as in \cite{boccato2017_2} and replace $\widehat{v}_N^\beta (p)$ by $\widehat{v} (0)$ paying a price that is $\mathcal{O}(N^{-\beta})$. We add and substract $ \sum_{p \in \Lambda_+^*} (v_N^\beta (p))^2 /p^2$ and find since 
\begin{align}
\label{eq:errorNbeta}
& \Big\vert \sum_{p \in \Lambda_+^*} \big[ -p^2 - \widehat{v}_N^\beta (p) + \sqrt{\vert p \vert^4 + 2p^2 \widehat{v}_N^\beta (p)} + \frac{(\widehat{v}_N^\beta (p))^2}{2p^2} \big]  \notag \\
& \hspace{2cm} - \sum_{p \in \Lambda_+^*} \big[ -p^2 - \widehat{v} (0) + \sqrt{\vert p \vert^4 + 2p^2 \widehat{v}_N^\beta (0)} + \frac{\widehat{v}(0)^2}{2 p^2} \Big\vert \leq CN^{-\alpha} \; . 
\end{align}
for every $\alpha < \beta $ from \cite[p. 2362 (before (5.35)]{boccato2017_2} that 
\begin{align}
\mathcal{C} & + \frac{1}{2} \sum_{p \in \Lambda_+^*} \big[ - F_p + \sqrt{F_p^2 - G_p^2} \big] \notag \\
=&  \frac{1}{2} \sum_{p \in \Lambda_+^*} \big[ -p^2 - \widehat{v} (0) + \sqrt{\vert p \vert^4 + 2p^2 \widehat{v}_N^\beta (0)} + \frac{\widehat{v}(0)^2}{ 2 p^2} \big] - 4 \pi \mathfrak{a}_N^\beta\notag \\
&+ \frac{N \widehat{v} (0)}{2} -  \sum_{p \in \Lambda_+^*} \frac{(\widehat{v}_N^\beta (p))^2}{ 4 p^2} + \sum_{p \in \Lambda_+^*} \Big[ \frac{A_p}{4p^2} - \frac{1}{2N} (\widehat{v}_N^\beta *sc)_p s_pc_p  \Big]  \notag \\
&- \frac{1}{2} \sum_{p \in \Lambda_+^*}  B_p + \mathcal{O}(N^{2(\beta -1)}) + \mathcal{O}(N^{-\alpha}  ) \; . \label{eq:allorders}
\end{align}
Using \eqref{def:eta} and \eqref{def:Ap}, we compute
\begin{align}
&\sum_{p \in \Lambda_+^*}  \Big[ \frac{A_p }{4p^2}   - \frac{1}{2N} (\widehat{v}_N^\beta *sc)_p s_pc_p \Big]  
\nonumber \\
&\quad = 
- \frac{1}{2N} \sum_{p \in \Lsp} (\widehat{v}_N^\beta* \eta)_p \frac{\widehat{v}_N^\beta(p)}{2 p^2}
- \frac{1}{2N} \sum_{p \in \Lsp} (\widehat{v}_N^\beta* (sc - \eta))_p \frac{\widehat{v}_N^\beta(p)}{2 p^2}
\nonumber \\
&\qquad - \frac{1}{2N} \sum_{p \in \Lsp}  \frac{(\widehat{v}_N^\beta *sc)_p}{p^2}
\Big[
\frac{1}{2N} (\widehat{v}_N^\beta *(sc - \eta) )_p + p^2 (s_p c_p - \eta_p) 
-\frac12\hat{v}(0)\eta_p\Big] .
\end{align}
Since $\vert s_p c_p \vert \ls \vert p \vert^{-2}$ from Lemma \ref{lem:HST} we have $\| \widehat{v}_N^\beta * s c \|_{\l^\infty} \ls N^\beta$ from Lemma \ref{lem:aux:potential} and thus 
\begin{align}
\left|\frac{1}{4N^2} \sum_{p \in \Lsp}  \frac{(\widehat{v}_N^\beta *sc)_p}{p^2} (\widehat{v}_N^\beta *(sc - \eta))_p  \right|
\ls N^{\beta-2} \sum_{p,q \in \Lambda_+^*, p \not= q}  \frac{\widehat{v}_N^\beta (p)}{p^2 \vert p-q \vert^6}  \ls N^{2(\beta-1)}\,,
\end{align}
where we used once more Lemma \ref{lem:aux:potential}. 
Moreover, 
\begin{align}
\left|\frac{\hat{v}(0)}{4N^2}\sum_{p\in\Lsp}(\hvNb*sc)_p\frac{\eta_p}{p^2}\right| 
\ls N^{-2}\norm{\hvNb*sc}_{\l^\infty}\sum_{p\in\Lsp}\frac{1}{p^4}\ls N^{\beta-2}\,.
\end{align}
Therefore we get 
\begin{align}
&\sum_{p \in \Lambda_+^*}  \Big[ \frac{A_p }{4p^2}   - \frac{1}{2N} (\widehat{v}_N^\beta *sc)_p s_pc_p \Big]  
\nonumber \\
&\quad = 
- \frac{1}{2N} \sum_{p \in \Lsp} (\widehat{v}_N^\beta* \eta)_p \frac{\widehat{v}_N^\beta(p)}{2 p^2} - \frac{1}{2N} \sum_{p,q \in \Lsp, p\not=q} \widehat{v}_N^\beta (p-q) (s_pc_p-\eta_p) \bigg[ s_qc_q + \frac{\widehat{v}_N^\beta (q)}{q^2}\bigg]\notag \\
&+  \mathcal{O}(N^{2(\beta-1)}). 
\end{align}
We insert this into \eqref{eq:allorders} and get 
\begin{align}
\mathcal{C} & + \frac{1}{2} \sum_{p \in \Lambda_+^*} \big[ - F_p + \sqrt{F_p^2 - G_p^2} \big] \notag \\
=&  \frac{1}{2} \sum_{p \in \Lambda_+^*} \big[ -p^2 - \widehat{v} (0) + \sqrt{\vert p \vert^4 + 2p^2 \widehat{v}_N^\beta (0)} + \frac{\widehat{v}(0)^2}{ 2 p^2} \big] - 4 \pi \mathfrak{a}_N^\beta\notag \\
&+ \frac{N \widehat{v} (0)}{2} -  \sum_{p \in \Lambda_+^*} \frac{(\widehat{v}_N^\beta (p))^2}{ 4 p^2} -  \frac{1}{2N} \sum_{p \in \Lsp} (\widehat{v}_N^\beta* \eta)_p \frac{\widehat{v}_N^\beta(p)}{2 p^2} \notag \\
&- \frac{1}{2N} \sum_{p,q \in \Lsp, p\not=q} \widehat{v}_N^\beta (p-q) (s_pc_p-\eta_p) \bigg[ s_qc_q + \frac{\widehat{v}_N^\beta (q)}{q^2}\bigg] \notag \\
&- \frac{1}{2}\sum_{p \in \Lambda_+^*}  B_p + \mathcal{O}(N^{2(\beta -1)}) + \mathcal{O}(N^{-\alpha}  ) \; . 
\end{align}
The second line of the r.h.s. sums by \eqref{def:eta} up to 
\begin{align}
4\pi N\aNb +\frac{\hat{v}(0)}{4N}\sum_{p\in\Lsp}\frac{\hvNb(p)\eta_p}{p^2}\,,
\end{align}
where the second summand is at most of order $ N^{-1}$ by Lemma \ref{lem:cp:sp}.
Hence, 
\begin{align}
\mathcal{C} & + \frac{1}{2} \sum_{p \in \Lambda_+^*} \big[ - F_p + \sqrt{F_p^2 - G_p^2} \big] \notag \\
=&  4 \pi (N-1) \mathfrak{a}_N^\beta +  \frac{1}{2} \sum_{p \in \Lambda_+^*} \big[ -p^2 - \widehat{v} (0) + \sqrt{\vert p \vert^4 + 2p^2 \widehat{v}_N^\beta (0)} + \frac{\widehat{v}(0)^2}{ 2 p^2} \big] \notag \\
&- \frac{1}{2N} \sum_{p,q \in \Lsp, p\not=q} \widehat{v}_N^\beta (p-q) (s_pc_p-\eta_p) \bigg[ s_qc_q + \frac{\widehat{v}_N^\beta (q)}{q^2}\bigg] \notag \\
&- \frac{1}{2} \sum_{p \in \Lambda_+^*}  B_p + \mathcal{O}(N^{2(\beta -1)}) +\mathcal{O}(N^{-\alpha}  ) \; . 
\end{align}
Recalling the definition of $A_p$  in \eqref{def:Ap}, the estimates $\| \widehat{v}_N^\beta * sc \|_{\l^\infty} \ls N^{\beta} $, $\| \widehat{v}_N^\beta *(sc-\eta) \|_{\l^\infty} \ls 1 $ and the definition of $B_p$ in \eqref{def:Bp}, we can extract the leading order contribution of $-\frac12\sum_{p\in\Lsp}B_p$ and find 
\begin{align}
&\mathcal{C} + \frac{1}{2} \sum_{p \in \Lambda_+^*} \big[ - F_p + \sqrt{F_p^2 - G_p^2} \big]\notag\\
&\qquad = 4 \pi (N-1) \mathfrak{a}_N^\beta + E_{0,0}+E_{0,1}\mathcal{O}(N^{2(\beta -1)}) + \mathcal{O}(N^{-1}) + \mathcal{O}(N^{-\alpha}  ) \,,\label{eq:diag-final}
\end{align}
for every $\alpha < \beta$, where
\begin{align}
E_{0,0}=&  \frac{1}{2} \sum_{p \in \Lambda_+^*} \big[ -p^2 - \widehat{v} (0) + \sqrt{\vert p \vert^4 + 2p^2 \widehat{v}_N^\beta (0)} + \frac{\widehat{v}(0)^2}{ 2 p^2} \big]\,
\end{align}
as defined in \eqref{E00} and
\begin{align}
E_{0,1}:=&- \frac{1}{2N} \sum_{p,q \in \Lsp, p\not= q} \widehat{v}_N^\beta (p-q) (s_pc_p-\eta_p) \bigg[ s_qc_q + \frac{\widehat{v}_N^\beta (q)}{q^2}\bigg] \notag \\
 &+ \frac{1}{N} \sum_{p,q \in \Lambda_+^*, p\not= q }  \frac{ \widehat{v}_N^\beta (p)^2 \; \; \widehat{v}_N^\beta (p-q) s_qc_q }{\sqrt{\vert p \vert^4 + 2 p^2 \widehat{v}_N^\beta(p) }  \left(  p^2 + \sqrt{\vert  p \vert^4 + 2 p^2 \widehat{v}_N^\beta(p) } \right) }  \label{E01}\; .
\end{align}

\paragraph{Proof of \eqref{eqn:orders:E00:01}.} Recalling the definition of $E_{0,0}$, we find by expanding the square root that $\big|-p^2 - \widehat{v} (0) + \sqrt{\vert p \vert^4 + 2p^2 \widehat{v}_N^\beta (0)} + \frac{\widehat{v}(0)^2}{ 2 p^2} \big| \ls \vert p \vert^{-4}$, yielding the desired estimate 
\begin{align}
\vert  E_{0,0} \vert \ls 1 \; .
\end{align}
Furthermore, we conclude from the definition of $E_{0,1}$   with  the estimates $\vert s_pc_p-\eta_p \vert \ls \vert p \vert^{-6} $, $\sum_{q \in \Lambda_+^*}\frac{\widehat{v}_N^\beta (q)}{q^2} \ls N^{\beta}$ and  $\norm{\hvNb*sc}_{\l^\infty}\ls N^\beta$ that $\vert E_{0,1} \vert \ls N^{\beta-1}$. 
\end{proof}

\section{Analysis of $\FockG$}
\label{sec:estimates-G}

In this section we study properties of 
\begin{align}
\mathbb{G} = \mathbb{T} \mathbb{H} \mathbb{T}^* -\mathcal{C}
\end{align}
defined in \eqref{def:G}, which will in the end yield Proposition \ref{prop:renormalized:hamiltonian}. The idea is to compute $\mathbb{G}$ using the explicit action of $\mathbb{T}$ on creation and annihilation operators \eqref{eq:actionbogo}. For this we decompose $\mathbb{H}$ from \eqref{def:H:extended} as 
\begin{equation}
\label{eq:Hsplit}
\mathbb{H} =  \FockHb+\tilde{\FockR}_{\sqrt{}} \,,
\end{equation}
where we defined
\begin{equation}\label{def:H:tilde}
\FockHb
:=  \frac12(N-1)\hat{v}(0)+ \boldKz +  \boldKo
+  \boldKt +  \boldKtbar 
+ \frac{1}{\sqrt{N}} \left( \boldKth +  \boldKthbar\right) 
+\frac{1}{N} \boldKf 
\end{equation} 
with $\mathbb{K}_j$ given by \eqref{eqn:K:notation} and 
\begin{align}
\label{def:R_sqrt}
\tilde{\FockR}_{\sqrt{}}
:=& 
\left(\left[\frac{N-\Np}{N}\right]_+ -1\right)\boldKo \notag \\
&+\left(\boldKt    \frac{\sqrt{\left[(N-\Np)(N-\Np-1)\right]_+} - N}{N} 
+  \hc\right) \nonumber\\
&+  \left(\boldKth\frac{\sqrt{\left[N-\Np\right]_+} - \sqrt{N}}{N}+ \hc\right) \; . 
\end{align} 
We prove Proposition \ref{prop:renormalized:hamiltonian} by computing the action of the Bogoliubov transformation $\mathbb{T}$ on the individual contributions of $\mathbb{H}$ in \eqref{eq:Hsplit} and decompose the result w.r.t.\ to the order of the terms in $N$. 
The first observation is that the conjugation of $\tilde{\FockR}_{\sqrt{}}$ with $\mathbb{T}$ contributes only to sub-leading order. To be more precise, we define 
\begin{align}
\label{def:R_sqrt:new}
\mathbb{R}_c := 
\FockT\tilde{\FockR}_{\sqrt{}}\FockT^*+\frac{1}{N}\sum_{p\in\Lsp}\hvNb(p)c_ps_p\FockT(\Np+\frac12)\FockT^* \,.
\end{align}
In Section \ref{subsec:sqrt} (Lemma \ref{lem:R2}), we prove that the remainder $\mathbb{R}_c$ is  of order $N^{3(\beta-1)/2}$. Next, we observe that the explicit action of the Bogoliubov transformation on creation and annihilation operators \eqref{eq:actionbogo} yields
\begin{align}
\label{eq:TKjT}
\mathbb{T} \left( \mathbb{K}_0 + \mathbb{K}_1 + \mathbb{K}_2+ \mathbb{K}_2^* \right) \mathbb{T}^* =& \mathbb{G}_0  + \sum_{p \in \Lambda^*} \big(p^2s_p^2 + \widehat{v}_N^\beta (s_p^2 + c_ps_p)\big) \notag \\
&- \frac{2}{N}\sum_{p \in \Lambda_+^*} ( \widehat{v}_N^\beta * sc)_p c_p s_p \ad_pa_p \notag \\
&- \frac{1}{2N} \sum_{p \in \Lambda_+^*} ( \widehat{v}_N^\beta * cs)_p (c_p^2 + s_p^2 ) (\ad_p\ad_{-p} + a_p a_{-p} )
\end{align} 
for $\FockGz$ as defined in \eqref{G0}.
The last two terms of the r.h.s.\ will be compensated for by $N^{-1}\mathbb{T} \mathbb{K}_4 \mathbb{T}^*$. In fact, a lengthy computation shows that
\begin{align}
\mathbb{T}  & \left( \mathbb{K}_0 + \mathbb{K}_1 + \mathbb{K}_2+\boldKt^* \right) \mathbb{T}^* + \frac{1}{N}\FockT \boldKf \FockT^*  \notag \\
=&\mathbb{G}_0 + \FockG_2+\FockR_a+\FockR_b \notag \\ 
&+\frac1{2N}\sum_{\substack{p\in\Lsp}}(\hvNb*cs)_pc_ps_p + \sum_{p \in \Lambda^*} \big(p^2s_p^2 + \widehat{v}_N^\beta (s_p^2 + c_ps_p)\big) - \frac{1}{N}\sum_{p\in\Lsp}\hvNb(p)c_ps_p \Big( \frac{1}{2} +  s_p^2 \Big) \notag\\
&+\frac{1}{N}\sum_{p\in\Lsp}\hvNb(p)c_ps_p\FockT(\Np+\frac12)\FockT^* \notag\\
=&\mathbb{G}_0 + \FockG_2+\FockR_a+\FockR_b +\frac{1}{N}\sum_{p\in\Lsp}\hvNb(p)c_ps_p\FockT(\Np+\frac12)\FockT^* +\mathcal{C}-\frac{N-1}{2}\hat{v}(0)\,,
\end{align}
where $\FockR_a$ and $\mathbb{R}_b$ are given in \eqref{E_1} and \eqref{E_2} below.
In Section \ref{subsec:E_1} (Lemma \ref{lem:Ra}), we prove that $\mathbb{R}_a$ is of order $N^{-1+\beta/2}$. In Section \ref{subsec:E_2} (Lemma \ref{lem:Rb}), we show that $\FockR_b$ is of order $N^{-1}$. 
Recalling the definition of $\mathbb{G}_1$ in \eqref{G1} and combining this with \eqref{def:H:tilde} and \eqref{def:R_sqrt:new}, we  arrive at 
\begin{align}
\mathbb{G} = \mathbb{G}_0 + \mathbb{G}_1  + \mathbb{G}_2  + \mathbb{R}_a
+ \mathbb{R}_b + \mathbb{R}_c =: \mathbb{G}_0 + \mathbb{G}_1  + \mathbb{G}_2  + \mathbb{R}_2\,,
\end{align}
where we set 
\begin{align}
\mathbb{R}_2 := \mathbb{R}_a + \mathbb{R}_b +\mathbb{R}_c\,.
\end{align}
Then Proposition \ref{prop:renormalized:hamiltonian} follows by Lemmas \ref{lem:Ra} -- \ref{lem:R2} proven in Sections \ref{subsec:E_1} -- \ref{subsec:sqrt} below.

\subsection{Analysis of $\FockG_1$}\label{subsec:G_1}

With \eqref{eq:actionbogo} we find 
\begin{align}
\label{eq:tildeG1}
\mathbb{G}_1 := \widetilde{\mathbb{G}}_1 + \mathbb{R}_d \,,
\end{align}
where the leading order contributions of $\mathbb{G}_1$, which we will show to be of order $N^{(\beta-1)/2}$, is given by 
\begin{align}
\widetilde{\mathbb{G}}_1 :=& \frac{1}{\sqrt{N}}\sum_{\substack{p,q\in\Lsp\\p+q\neq0}}\widehat{v}_N^\beta (p) c_{p+q}c_pc_q( \ad_{p+q}\ad_{-p}a_q + {\rm h.c.} ) \notag \\
& +\frac{1}{\sqrt{N}}\sum_{\substack{p,q\in\Lsp\\p+q\neq0}}\widehat{v}_N^\beta (p) c_{p+q}c_ps_q(\ad_{p+q}\ad_{-p}\ad_{-q} +
{\rm h.c.} ) \notag \\
=:& \left(A_{d_1}  + A_{d_2} +\hc\right)  \; . 
\end{align}
 The remaining terms are of order $N^{-1/2}$ and given by 
\begin{align}
\mathbb{R}_d 
:=&\frac{1}{\sqrt{N}}\sum_{\substack{p,q\in\Lsp\\p+q\neq0}}\widehat{v}_N^\beta (p) c_{p+q}s_pc_q( \ad_{p+q}a_pa_q + {\rm h.c.} ) 
\notag \\
&+\frac{1}{\sqrt{N}}\sum_{\substack{p,q\in\Lsp\\p+q\neq0}}\widehat{v}_N^\beta (p) s_{p+q}c_pc_q ( a_{-(p+q)}\ad_{-p}a_q + {\rm h.c.}) \notag
\\
&+\frac{1}{\sqrt{N}}\sum_{\substack{p,q\in\Lsp\\p+q\neq0}}\widehat{v}_N^\beta (p) s_{p+q}s_pc_q(a_{-(p+q)}a_pa_q+  {\rm h.c.})
\notag
\\
&+\frac{1}{\sqrt{N}}\sum_{\substack{p,q\in\Lsp\\p+q\neq0}}\widehat{v}_N^\beta (p) c_{p+q}s_ps_q (\ad_{p+q}a_p\ad_{-q} +  {\rm h.c.})
\notag
\\
& +\frac{1}{\sqrt{N}}\sum_{\substack{p,q\in\Lsp\\p+q\neq0}}\widehat{v}_N^\beta (p) s_{p+q}c_ps_q (a_{-(p+q)}\ad_{-p}\ad_{-q} +  {\rm h.c.})\notag
\\
& +\frac{1}{\sqrt{N}}\sum_{\substack{p,q\in\Lsp\\p+q\neq0}}\widehat{v}_N^\beta (p) s_{p+q}s_ps_q (a_{-(p+q)}a_p\ad_{-q} +  {\rm h.c.})
\notag \\
=:& \left(A_{d_3} +  A_{d_4}  +  A_{d_5}  +  A_{d_6}  + A_{d_7}  +  A_{d_8} +  \hc\right) \,.
\end{align}

\begin{lem}\label{lem:G1}
For $\psi,\xi\in \Fp$ and $\l_1,\l_2\in \mathbb{R}$, we have 
\begin{align}
\left|\lr{\psi,\widetilde{\FockG}_1\xi}\right| 
&\ls  N^{\frac{\beta -1}{2}}   \norm{\left( \boldKz  + 1 \right)^{1/2} \left(\Np + 1 \right)^{\l_1} \psi}
\norm{\left(\Np + 1 \right)^{- \l_1 +1} \xi}
\nonumber \\
&\quad + 
N^{\frac{\beta -1}{2}}   \norm{\left(\Np + 1 \right)^{- \l_2 +1} \psi}
\norm{\left( \boldKz  + 1 \right)^{1/2} \left(\Np + 1 \right)^{\l_2} \xi} 
\end{align}
and, moreover, 
\begin{align}
\left|\lr{\psi,\mathbb{R}_d \xi}\right| 
&\ls   N^{-\frac12}   \norm{ \left(\Np + 1 \right)^{\l_1} \psi}
\norm{\left(\Np + 1 \right)^{- \l_1 + \frac{3}{2}} \xi} \; . 
\end{align}
\end{lem}

\begin{proof}
We start with the first bound for $\widetilde{\mathbb{G}}_1$ defined in \eqref{eq:tildeG1} and consider both its contributions separately. 

\paragraph{Term $\widetilde{\mathbb{G}}_1$.} Since $\norm{c}_{\l^\infty}\ls 1$ by Lemma \ref{lem:cp:sp}, 
\begin{align}
&\abs{\lr{\psi,A_{d_1}\xi} }\notag \\
&\quad \ls N^{-1/2} \sum_{\substack{p,q\in\Lsp}} \widehat{v}_N^\beta (p)  \abs{\lr{\psi, \left(\Np + 1 \right)^{\l}  \ad_{p+q}\ad_{-p} a_q \left( \Np + 2 \right)^{-\l} \xi} }
\nonumber\\
&\quad \ls N^{-1/2}\left( \sum_{\substack{p,q\in\Lsp}} |p|^2
\norm{a_{p+q}a_{-p} \left(\Np + 1 \right)^{\l} \psi}^2\right)^\frac12
\| \vert p \vert^{-1} \widehat{v}_N^\beta \|_{\ell^2( \Lambda_+^*)} 
\left( \sum_{\substack{q \in\Lsp}} \norm{a_q \left( \Np + 3 \right)^{-\l} \xi}^2\right)^\frac12 \,.
\end{align}
With Lemma \ref{lem:aux:potential} we arrive at the desired bound 
\begin{align}
\abs{\lr{\psi,A_{d_1}\xi} }
&\ls N^\frac{\beta-1}{2} \norm{ \mathbb{K}_0 \left(\Np + 1 \right)^{\l+\frac{1}{2}} \psi}
\norm{\left(\Np + 1 \right)^{- \l + \frac{1}{2}} \xi} .
\end{align}
Using similar ideas we continue with 
\begin{align}
&\abs{ \lr{\psi,A_{d_2}\xi} }
\nonumber \\
&\quad \ls N^{-1/2} 
\norm{ \mathbb{K}_0^{1/2} \left(\Np + 1 \right)^{\l+ \frac{1}{2} } \psi} \, \norm{\abs{p}^{-1} \widehat{v}_N^\beta}_{\ell^2} 
\Big(  \sum_{\substack{q\in\Lsp}} 
\left|s_q \right|^2 \norm{\ad_{-q} \left( \Np + 4 \right)^{-\l} \xi}^2
\Big)^{1/2} 
\end{align}
Using once more Lemmas \ref{lem:aux:potential} and \ref{lem:cp:sp}, we obtain
\begin{align}
\abs{ \lr{\psi,A_{d_2}\xi} }
\nonumber \ls  N^{\frac{\beta-1}{2}} \norm{\mathbb{K}_0^{1/2}  \left(\Np + 1 \right)^{\l+ \frac{1}{2} } \psi}
\norm{\left( \Np + 1 \right)^{-\l + \frac{1}{2}} \xi}  . 
\end{align}
The hermitian conjugates can be estimated simiarly.

\paragraph{Term $\mathbb{R}_d$.} Next we prove the bound on $\mathbb{R}_d$. For this we consider all its contributions separately. We start with 
\begin{align}
&\abs{\lr{\psi, A_{d_3}\xi} }\notag \\
&\quad \ls N^{-1/2} \sum_{\substack{p,q\in\Lsp}} \widehat{v}_N^\beta
\abs{s_p} \abs{\scp{\psi}{\left( \Np + 2 \right)^{\l} \ad_{p+q} a_p a_q \left( \Np + 1 \right)^{-\l}\xi} }
\nonumber \\
&\quad \ls N^{-1/2} \norm{\hat{v}}_{\l^{\infty}}
\left( \sum_{\substack{p,q\in\Lsp}} 
\abs{s_p}^2  \norm{a_{p+q} \left( \Np + 2 \right)^{\l} \psi}^2 \right)^{1/2} 
\left( \sum_{\substack{p,q\in\Lsp}}  \norm{a_p a_q  \left( \Np + 1 \right)^{-\l} \xi}^2  \right)^{1/2} 
\nonumber \\
&\quad \ls N^{-1/2}  \norm{\left( \Np + 1 \right)^{\l + \frac{1}{2}}  \psi} \norm{\left( \Np + 1 \right)^{-\l + 1} \xi} \,,
\end{align}
by Lemma \ref{lem:cp:sp}. In order to estimate the next term, note that
\begin{align}
 A_{d_4}
&= N^{-1/2}  \sum_{\substack{p,q\in\Lsp\\p+q\neq0}}\hvNb(p) s_{p+q} c_p c_q a_{-(p+q)} \ad_{-p} a_q
\nonumber \\
&= 
N^{-1/2}  \sum_{\substack{p,q\in\Lsp\\p+q\neq0}}\hvNb(p) s_{p+q} c_p c_q  \ad_{-p} a_{-(p+q)} a_q
\end{align}
because $q \in \Lsp$. This allows us to estimate with Lemma \ref{lem:cp:sp} 
\begin{align}
&\abs{\scp{\psi}{ A_{d_4} \xi}}
\nonumber \\
&\quad \ls N^{-1/2}   \norm{\hat{v}}_{\l^{\infty}}
\bigg(  \sum_{\substack{p,q\in\Lsp\\p+q\neq0}} \abs{s_{p+q}}^2 \norm{a_{-p} \left( \Np + 2 \right)^\l \psi}^2 \bigg)^{1/2}
\bigg(  \sum_{\substack{p,q\in\Lsp\\p+q\neq0}}  \norm{a_{-(p+q)} a_q  \left( \Np + 1 \right)^{-\l} \xi}^2 \bigg)^{1/2}
\nonumber \\
&\quad \ls N^{-1/2}   \norm{\left( \Np + 1 \right)^{\l + \frac{1}{2} } \psi} \norm{\left( \Np + 1 \right)^{-\l + 1} \xi} .
\end{align}
For the next term, we find with similar arguments 
\begin{align}
&\abs{\scp{\psi}{ A_{d_5} \xi}}
\nonumber \\
&\quad \ls N^{-1/2} \norm{\hat{v}}_{\l^\infty}
\bigg(  \sum_{\substack{p,q\in\Lsp\\p+q\neq0}} \abs{s_{p+q}}^2 \abs{s_p}^2 \norm{\ad_{-(p+q)}  \left( \Np + 4 \right)^{\l} \psi}^2 \bigg)^{1/2}
\bigg(  \sum_{\substack{p,q\in\Lsp\\p+q\neq0}} \norm{a_p a_q \left( \Np + 1 \right)^{-\l} \xi}^2 \bigg)^{1/2}
\nonumber \\
&\quad \ls N^{-1/2}   \norm{ \left( \Np + 1 \right)^{\l+ \frac{1}{2}} \psi} \norm{\left( \Np + 1 \right)^{-\l + 1} \xi} 
\end{align}
and, moreover, by Lemma \ref{lem:cp:sp} 
\begin{align}
&\abs{\scp{\psi}{ A_{d_6} \xi}}\notag \\
&\quad \ls N^{-1/2}  \norm{\hat{v}}_{\l^\infty(\Lsp)}
\bigg(  \sum_{\substack{p,q\in\Lsp\\p+q\neq0}} \abs{s_p}^2
\norm{a_{p+q} \left(\Np + 1 \right)^\l \psi}^2 \bigg)^{1/2}
\bigg(  \sum_{\substack{p,q\in\Lsp\\p+q\neq0}} \abs{s_q}^2
\norm{a_{p} \ad_{-q}  \left( \Np + 2 \right)^{-\l} \xi}^2 \bigg)^{1/2}
\nonumber \\
&\quad \ls N^{-1/2}   \norm{\left(\Np + 1 \right)^{\l + \frac{1}{2} } \psi}
\norm{\left( \Np + 1 \right)^{-\l +1} \xi} .
\end{align}
For the next term, we put the creation and annihilation operators in normal order using that $\delta_{-(p+q), - p} =0$ for all $q \neq 0$ and $\delta_{-(p+q), - q} =0$ for all $p \neq 0$ so that 
\begin{align}
 A_{d_7}
&= N^{-1/2}  \sum_{\substack{p,q\in\Lsp\\p+q\neq0}} \widehat{v}_N^\beta(p)  s_{p+q} c_p s_q  \ad_{-p} \ad_{-q} a_{-(p+q)} .
\end{align}
Then we can proceed analogously as for $A_{d_3}$ and thus get 
\begin{align}
&\abs{\scp{\psi}{ A_{d_7} \xi}}  \ls N^{-1/2} \norm{\left(\Np + 1 \right)^{\l+1} \psi} \norm{\left( \Np + 1 \right)^{-\l + \frac{1}{2}}  \xi} .
\end{align}
For the last term, we find, using that  $[a_{- (p+q)}, \ad_{-q}]=0$,
\begin{align}
&\abs{\scp{\psi}{ A_{d_8}\xi}}
\nonumber \\
&\quad \ls  N^{-1/2}  \norm{\hat{v}}_{\l^\infty} \norm{s}_{\l^\infty}
\bigg( \sum_{\substack{p,q\in\Lsp\\p+q\neq0}} \abs{s_p}^2  \norm{a_{-q} \ad_p \left( \Np + 2 \right)^{\l} \psi}^2  \bigg)^{1/2}
\nonumber \\
&\quad \times
\bigg( \sum_{\substack{p,q\in\Lsp\\p+q\neq0}} \abs{s_q}^2 \norm{a_{-(p+q)} \left( \Np + 1 \right)^{-\l} \xi}^2  \bigg)^{1/2}
\nonumber \\ 
&\quad \ls  N^{-1/2}  \norm{\left( \Np + 1 \right)^{\l+1} \psi} 
\norm{\left( \Np + 1 \right)^{-\l + \frac{1}{2}} \xi} .
\end{align}
The hermitian conjugates can be estimated similarly. 
\end{proof}

\subsection{Analysis of $\FockG_2$}

Here we show that the operator $\mathbb{G}_2$, given by 
\begin{align}
\mathbb{G}_2 = \frac{1}{2N} \sum_{\substack{p,q,r\in\Lsp\\ p+r\neq0,q+r\neq0}}\hvNb(r) c_{p+r} c_q c_p  c_{q+r} \ad_{p+r} \ad_q   a_p a_{q+r}\,,
\end{align}
is of order $N^{\beta-1}$. The proof of the following Lemma closely follows ideas of \cite[Lemma 7.3]{boccato2017_2}. For completeness, we nevertheless sketch the proof here.

\label{subsec:G_2}
\begin{lem}\label{lem:G2}
Let $\psi,\xi\in\Fp$ and $\l \in\R$. Then, 
\begin{equation}
\left|\lr{\psi,\FockG_2\xi}\right| \ls N^{\beta-1}
\norm{ \boldKz^\frac12 \left( \Np +1 \right)^{\frac12+\l} \psi }
\norm{ \boldKz^\frac12 \left( \Np +1 \right)^{\frac12-\l} \xi }  \; . 
\end{equation}
\end{lem}
\begin{proof}
Since $\vert c_p \vert \ls 1$ by Lemma \ref{lem:cp:sp}, we find with Lemma \ref{lem:aux:potential} 
\begin{align}
\abs{\scp{\psi}{ \FockG_2 \xi}}
&\leq \frac{1}{2N} \bigg( \sum_{\substack{p,q,r\in\Lsp\\ p+r\neq0,q+r\neq0}} \frac{\abs{\widehat{v}_N^\beta (r) }}{(q+r)^2} (p+r)^2 \norm{a_q a_{p+r} \left( \Np +1 \right)^{\l} \psi}^2 \bigg)^{1/2}
\nonumber \\
&\qquad \qquad \times 
\bigg( \sum_{\substack{p,q,r\in\Lsp\\ p+r\neq0,q+r\neq0}} \frac{\abs{\widehat{v}_N^\beta (r)  }}{(p+r)^2} (q+r)^2 \norm{a_p a_{q+r} \left( \Np +1 \right)^{-\l} \xi}^2 \bigg)^{1/2}
\nonumber \\
&\ls N^{\beta - 1}
\norm{ \mathbb{K}_0^{1/2} \left( \Np +1 \right)^{\l + \frac{1}{2}} \psi }
\norm{  \mathbb{K}_0^{1/2} \left( \Np +1 \right)^{-\l + \frac{1}{2}} \xi } .
\end{align}
\end{proof}

\subsection{Analysis of $\FockR_a$}\label{subsec:E_1}

In this section, we show that the remainder term $\mathbb{R}_a$, given by 
\begin{align}\label{E_1}
\FockR_a
&:=\frac1{N}\sum_{\substack{p,q,r\in\Lsp\\p+r\neq0, q+r\neq 0}}\hvNb(r)c_p c_q c_{q+r} s_{p+r}\left(\ad_{q+r}\ad_p\ad_{-(p+r)}a_q+\hc\right)\notag\\
&\quad+ \frac1{2N}\sum_{\substack{p,q,r\in\Lsp\\p+r\neq0, q+r\neq 0}}\hvNb(r) c_{p+r} c_q s_p s_{q+r}\left(\ad_{p+r} \ad_q  \ad_{-p}  \ad_{-(q+r)}+ \hc\right)\nonumber\\
&=:\left(A_{a1}+A_{a2}+\hc\right)\,,
\end{align}
is of order $N^{\beta/2-1}$. 

\begin{lem}\label{lem:Ra}
Let $\psi,\xi\in\Fp$ and $\l_1,\l_2 \in\R$. Then 
\begin{align}
\left|\lr{\psi,\FockR_a\xi}\right| 
&\ls
N^{\frac{\beta}{2}-1}\norm{ \boldKz^{1/2} \left(\Np + 1 \right)^{\frac12 - \l_1} \psi}\norm{\left( \Np + 1 \right)^{1 + \l_1} \xi}
\nonumber \\
&\quad +
N^{\frac{\beta}{2}-1} \norm{\left( \Np + 1 \right)^{1 + \l_2} \psi} 
\norm{ \boldKz^{1/2} \left( \Np + 1 \right)^{\frac12 - \l_2} \xi} .
\end{align}
\end{lem}

\begin{proof}
For $A_{a1}$, we compute with Lemma \ref{lem:cp:sp} and \eqref{eq:auxiliary bound for the potential}
\begin{align}
&\left|\lr{\psi,A_{a1}\xi}\right|\nonumber\\
&\ls \frac{\norm{\hat{v}}_{\ell^\infty}^\frac12}{N}\bigg(\sum_{p,q,r\in\Lsp}|q+r|^2\norm{a_p a_{-(p+r)} a_{q+r} \left(\Np +1 \right)^{-\frac12 +\l } \psi}\bigg)^\frac12 \notag \\
&\hspace{1cm} \times 
\bigg(\sum_{\substack{p,q,r\in\Lsp\\p+r\neq0, q+r\neq 0}}\frac{\hvNb(r)}{|q+r|^2}|s_{p+r}|^2\norm{a_{q}(\Np+3)^{\frac12 - \l} \xi}^2\bigg)^\frac12 
\nonumber\\
&\ls N^{\frac{\beta}{2}-1}\norm{\mathbb{K}_0^{1/2} \left( \Np + 1 \right)^{\frac12 + \l} \psi}\norm{\left( \Np + 1 \right)^{1 - \l} \xi}
\end{align}
and similarly 
\begin{align}
&\left|\lr{\psi,A_{a2}\xi}\right|\nonumber\\
&\ls \frac{\norm{\hat{v}}_{\ell^\infty}^\frac12}{N}\bigg(\sum_{\substack{p,q,r\in\Lsp\\p+r\neq0, q+r\neq 0}}\frac{\hvNb(r)}{|p+r|^2}|s_p|^2|s_{q+r}|^2\norm{\ad_{-p}(\Np+1)^{\frac12 - \l}\xi}^2\bigg)^\frac12
\nonumber \\
&\quad \times
\bigg(\sum_{\substack{p,q,r\in\Lsp\\p+r\neq0, q+r\neq 0}}|p+r|^2\norm{a_q a_{-(q+r)} a_{p+r} \left(\Np +5 \right)^{ - \frac12 + \l}  \psi}^2\bigg)^\frac12\nonumber\\ 
&\ls N^{\frac{\beta}{2}-1}\norm{ \left( \Np + 1 \right)^{1- \l} \psi}\norm{\mathbb{K}_0^{1/2} \left( \Np + 1 \right)^{\frac{1}{2} + \l}\xi}\,.
\end{align}
The hermitian conjugates can be estimated similarly and we arrive at the desired bounds. 
\end{proof}

\subsection{Analysis of $\FockR_b$}\label{subsec:E_2}

In this section we show that the remainder term $\mathbb{R}_b$ 
\begin{align}
\FockR_b \label{E_2}
&=\frac1{2N}\sum_{\substack{p,q,r\in\Lsp\\p+r\neq0, q+r\neq 0}}\hvNb(r)c_q s_p s_{p+r} s_{q+r} \left(\ad_{-p}\ad_q\ad_{-(q+r)} a_{-(p+r)}+\hc\right)\notag \\
&\quad
+ \frac1{2N}\sum_{\substack{p,q,r\in\Lsp\\p+r\neq0, q+r\neq 0}}\hvNb(r) c_{p+r} s_p s_q  s_{q+r}\left(\ad_{p+r} a_{- q}   \ad_{-p}  \ad_{-(q+r)}+ \hc\right) \notag \\
&\quad+\frac1{2N}\sum_{\substack{p,q,r\in\Lsp\\p+r\neq0, q+r\neq 0}}\hvNb(r)c_pc_qs_{p+r}s_{q+r}\ad_q  \ad_{-(q+r)}a_{-(p+r)} a_p\notag \\
&\quad+\frac1{2N}\sum_{\substack{p,q,r\in\Lsp\\p+r\neq0, q+r\neq 0}}\hvNb(r) c_p c_{p+r} s_q s_{q+r} \left(\ad_{p+r} \ad_{-(q+r)} a_p a_{-q}+\hc\right)\notag \\
&\quad 
+ \frac1{2N}\sum_{\substack{p,q,r\in\Lsp\\p+r\neq0, q+r\neq 0}}\hvNb(r)c_{p+r} c_{q+r}  s_p s_q \ad_{p+r} a_{- q}   \ad_{-p} a_{q+r}\notag \\
&\quad 
+\frac1{2N}\sum_{\substack{p,q,r\in\Lsp\\p+r\neq0, q+r\neq 0}}\hvNb(r) s_p s_{p+r} s_q s_{q+r} a_{-(p+r)} a_{- q}   \ad_{-p}  \ad_{-(q+r)}\notag \\
&\quad +\frac1{2N}\sum_{\substack{p,q\in\Lsp\\p\neq q}}\hvNb(p-q)c_p^2s_q^2\ad_pa_p\notag
\\
&\quad+\frac1{2N}\sum_{\substack{p,q\in\Lsp\\p\neq q}}\hvNb(p-q) c_p  s_ps_{q}^2\left(\ad_{-p}\ad_p+\hc\right)\notag\\
&\quad-\frac1N\sum_{p\in\Lsp}\hvNb c_p s_p\left((c_p^2+s_p^2)\ad_pa_p+c_ps_p(\ad_p\ad_{-p}+\hc)\right)  \notag \\
&=:A_{b1}+A_{b1}^*+A_{b2}+A_{b2}^*+A_{b3}+A_{b4}+A_{b4}^*+A_{b5}+A_{b6}+A_{b7}+A_{b8}+A_{b8}^*+A_{b9}+A_{b9}^*\,.
\end{align}
is of order $N^{-1}$ and thus subleading for our discussion.

\begin{lem}\label{lem:Rb}
Let $\psi,\xi\in\Fp$ and $\l \in\R$. Then,
\begin{equation}
\left|\lr{\psi,\FockR_b\xi}\right| \ls \frac{1}{N}\norm{(\Np+1)^{1+\l}\psi}\norm{(\Np+1)^{1-\l}\xi}
\end{equation}
\end{lem}

\begin{proof}
We estimate all the contributions of $\mathbb{R}_b$ separately. Using Lemma \ref{lem:cp:sp}, we start with 
\begin{align}
\left|\lr{\psi,A_{b1}\xi}\right|
&\leq \frac{\norm{\hat{v}}_{\ell^\infty}\norm{s}_{\ell^\infty}}{2N}
\bigg(\sum_{p,q,r\in\Lsp}\norm{a_{-p}a_qa_{-(q+r)} \left( \Np + 1 \right)^{\l-\frac12}\psi}^2\bigg)^\frac12
\nonumber \\
&\quad \times 
\bigg(\sum_{\substack{p,q,r\in\Lsp\\q+r\neq 0}}|s_p|^2|s_{q+r}|^2\norm{a_{-(p+r)}(\Np+3)^{-\l + \frac12}\xi}^2\bigg)^\frac12\nonumber\\
&\ls \frac{1}{N}\norm{(\Np+1)^{\l+1} \psi}\norm{(\Np+1)^{-\l + 1} \xi}\,.
\end{align}
Since $[a_q,\ad_{-(q+r)}]=0$ for $r\neq 0$,   normal ordering yields
\begin{align}
\left|\lr{\psi,A_{b2}\xi}\right|
&\leq\frac{\norm{\hat{v}}_{\ell^\infty}}{2N}\sum_{\substack{p,q,r\in\Lsp\\q+r\neq 0}}|s_p| |s_q| |s_{q+r}|\left|\lr{\psi, \left( \Np + 1 \right)^{\l} \ad_{p+r} \ad_{-p} \ad_{-(q+r)} a_{-q} \left( \Np + 3 \right)^{-\l} \xi}\right| \nonumber\\
&\quad+ \frac{\norm{\hat{v}}_{\ell^\infty}\norm{s}_{\ell^\infty}}{2N}\sum_{p,r\in\Lsp}|s_p|^2 \left|\lr{\psi,\left( \Np + 1 \right)^{\l} \ad_{p+r}\ad_{-(p+r)} \left( \Np + 3 \right)^{-\l} \xi}\right|\nonumber\\
&\ls \frac{1}{N}\bigg(\sum_{p,q,r\in\Lsp}|s_q|^2\norm{a_{-p}a_r \left( \Np + 1 \right)^{\l} \psi}^2\bigg)^\frac12 \bigg(\sum_{p,q,r\in\Lsp}|s_p|^2|s_r|^2\norm{(\Np+1)^{-\l + \frac{1}{2}} a_q\xi}^2\bigg)^\frac12\nonumber\\
&\quad + \frac{1}{N}\sum_{p,q\in\Lsp}\norm{a_q \left( \Np + 1 \right)^\l \psi}\norm{(\Np+1)^{- \l + \frac{1}{2}} \xi}\nonumber\\
&\ls \frac{1}{N}\norm{(\Np+1)^{\l+ 1} \psi}\norm{(\Np+1)^{-\l+1} \xi}\,.
\end{align}
Moreover,
\begin{align}
&\left|\lr{\psi,A_{b3}\xi}\right| \notag \\
&\quad \ls\frac{\norm{\hat{v}}_{\ell^\infty}}{N}\sum_{\substack{p,q,r\in\Lsp\nonumber\\ p+r\neq0,q+r\neq0}}|s_{q+r}||s_{p+r}|\norm{a_{-(q+r)}a_q \left( \Np + 1 \right)^{\l} \psi}\norm{a_{-(p+r)}a_p \left( \Np + 1 \right)^{-\l} \xi}\nonumber\\
&\quad \ls \frac{1}{N}\bigg(\sum_{\substack{p,q,r\in\Lsp\\ p+r\neq0}}|s_{p+r}|^2\norm{a_{-(q+r)}a_q \left( \Np + 1 \right)^\l \psi}^2\bigg)^\frac12 
\bigg(\sum_{\substack{p,q,r\in\Lsp\\ q+r\neq0}}|s_{q+r}|^2\norm{a_{-(p+r)}a_p ( \Np + 1 )^{-\l} \xi}^2\bigg)^\frac12\nonumber\\
&\quad \ls \frac{1}{N}\norm{( \Np + 1 )^{\l+1} \psi}\norm{\left( \Np + 1 \right)^{-\l + 1} \xi}\, . 
\end{align}
The next term can be estimated as
\begin{align}
& \left|\lr{\psi,A_{b4}\xi}\right| \notag \\
&\quad\ls \frac{\norm{\hat{v}}_{\ell^\infty}}{2N}\bigg(\sum_{p,q,r\in\Lsp}|s_q|^2\norm{a_{p+r} a_{-(q+r)} \left( \Np + 1 \right)^{\l} \psi}^2\bigg)^\frac12 \bigg(\sum_{\substack{p,q,r\in\Lsp\\q+r\neq 0}}|s_{q+r}|^2\norm{a_pa_{-q}\left( \Np + 1 \right)^{-\l} \xi}^2\bigg)^\frac12\nonumber\\
&\quad\ls \frac{1}{N}\norm{\left( \Np + 1 \right)^{\l+1} \psi}\norm{\left( \Np + 1 \right)^{-\l+1}\xi}\,.
\end{align}
Normal ordering yields
\begin{align}
&\left|\lr{\psi,A_{b5}\xi}\right| \notag \\
&\quad \ls \frac{\norm{\hat{v}}_{\ell^\infty}}{2N}\bigg(\sum_{p,q,r\in\Lsp}|s_q|^2\norm{a_{p+r}a_{-p} \left( \Np + 1 \right)^{\l} \psi}^2\bigg)^\frac12 \bigg(\sum_{p,q,r\in\Lsp}|s_p|^2\norm{a_{-q}a_{q+r} \left( \Np + 1 \right)^{-\l}\xi}^2\bigg)^\frac12\nonumber\\
&\quad \quad+ \frac{\norm{\hat{v}}_{\ell^\infty}}{2N}\sum_{p\in\Lsp}|s_p|^2\sum_{r\in\Lsp}\norm{a_{p+r} \left( \Np + 1 \right)^{\l} \psi}\norm{a_{p+r}  \left( \Np + 1 \right)^{-\l} \xi}\nonumber\\
&\quad\ls\frac{1}{N}\norm{\left( \Np + 1 \right)^{\l+1}\psi}\norm{\left( \Np + 1 \right)^{-\l+1} \xi}\,
\end{align}
and 
\begin{align}
&\left|\lr{\psi,A_{b6}\xi}\right| \notag \\
&\quad \ls \frac{\norm{\hat{v}}_{\ell^\infty}\norm{s}^2_{\ell^\infty}}{2N}\bigg(\sum_{p,q,r\in\Lsp}|s_q|^2\norm{a_{-p}a_{-(q+r)} \left( \Np + 1 \right)^{\l} \psi}^2\bigg)^\frac12 
\bigg(\sum_{p,q,r\in\Lsp}|s_p|^2\norm{a_{-q} a_{-(p+r)} \left( \Np + 1 \right)^{-\l} \xi}^2\bigg)^\frac12\nonumber\\
&\quad +\frac{\norm{\hat{v}}_{\ell^\infty}}{2N}\sum_{p,r\in\Lsp}|s_p|^2|s_{p+r}|^2\norm{(\Np+1)^{\l+\frac{1}{2}} \psi}\norm{(\Np+1)^{-\l+\frac{1}{2}} \xi}\nonumber\\
&\quad \ls\frac{1}{N}\norm{(\Np+1)^{\l+1} \psi}\norm{(\Np+1)^{-\l + 1} \xi}\,.
\end{align}
The remaining contributions of $\mathbb{R}_b$ can be estimated similarly.
\end{proof}

\subsection{Analysis of $\FockR_c$}\label{subsec:sqrt}
We will show that the remainder $\mathbb{R}_c$, given by 
\begin{align*}
\mathbb{R}_c = 
\FockT\tilde{\FockR}_{\sqrt{}}\FockT^*+\frac{1}{N}\sum_{p\in\Lsp}\hvNb(p)c_ps_p\FockT(\Np+\frac12)\FockT^* 
\end{align*}
with
\begin{align*}
\tilde{\FockR}_{\sqrt{}}
=& 
\left(\left[\frac{N-\Np}{N}\right]_+ -1\right)\boldKo \notag \\
&+\left(\boldKt    \frac{\sqrt{\left[(N-\Np)(N-\Np-1)\right]_+} - N}{N} 
+  \hc\right) \nonumber\\
&+  \left(\boldKth\frac{\sqrt{\left[N-\Np\right]_+} - \sqrt{N}}{N}+ \hc\right)\notag \\
=:& A_{c_1} + A_{c_2} +A_{c_2}^* + A_{c_3} + A_{c_3}^* 
\end{align*} 
is of order $N^{3(\beta-1)/2}$ for all $\beta>1/2$.  

\begin{lem}\label{lem:R2}
Let $\psi, \xi \in \mathcal{F}_\perp$ and $\l_1,\l_2,\l_3 \in \mathbb{R}$. Then
\begin{align}
\abs{\scp{\psi}{\FockR_c \xi }}
&\quad \ls  N^{ \beta/2-1 } \norm{(\mathbb{K}_0 +1 )^{1/2} \left(\Np + 1 \right)^{\l_1 + 1/2} \psi} \norm{\left( \Np + 1 \right)^{-\l_1+1} \xi} 
\nonumber \\
&\qquad +  N^{\beta/2-1 } 
\norm{\left( \Np + 1 \right)^{-\l_2+1} \psi}
\norm{(\mathbb{K}_0 +1 )^{1/2} \left(\Np + 1 \right)^{\l_2 + 1/2} \xi}
\nonumber \\
&\qquad +  \left( N^{-1} + N^{\frac32(\beta- 1) } \right)
\norm{\left( \Np + 1 \right)^{\l_3+ \frac{5}{4}} \psi }
\norm{\left( \Np + 1 \right)^{- \l_3 + \frac{5}{4}}  \xi }\,,
\end{align}
which, for $\beta\in(1/2,1)$, simplifies to
\begin{align}
\left|\lr{\psi,\FockR_c\xi}\right|
&\ls   N^{\frac32(\beta- 1) }  
\norm{(\mathbb{K}_0 +1 )^{1/2}\left(\Np + 1 \right)^{3/4} \psi}
\norm{(\mathbb{K}_0 +1 )^{1/2} \left(\Np + 1 \right)^{3/4} \xi} .
\end{align}
\end{lem}

\begin{proof}
In order to estimate the first contribution $\mathbb{T} A_{c_1} \mathbb{T}^*$, we first observe that the function $g_0 (\mathcal{N}_\perp):= \left[\frac{N-\Np}{N}\right]_+ -1$ satisfies 
\begin{align}
g_0^2 (\mathcal{N}_\perp) \ls \mathcal{N}_\perp /N 
\end{align}
as operator inequality on $\mathcal{F}_\perp$. Thus, 
\begin{align}
\abs{\scp{\psi}{\mathbb{T} A_{c_1} \mathbb{T}^* \xi}}
&\leq  \sum_{p\in\Lsp} \hvNb(p) \abs{\scp{\psi}{ \mathbb{T}\left( \Np +1 \right)^{\l - \frac{1}{2}}   g_0 (\mathcal{N}_\perp ) \ad_p a_p  \left( \Np +1 \right)^{-\l + \frac{1}{2}} \mathbb{T}^*\xi}}
\nonumber \\
&\leq \norm{\widehat{v}}_{\l^{\infty} (\Lsp)}  \sum_{p\in\Lsp} \norm{ a_pg_0( \mathcal{N}_\perp) \left( \Np +1 \right)^{\l - \frac{1}{2}}\mathbb{T}^* \psi} \norm{ a_p \left( \Np +1 \right)^{-\l + \frac{1}{2}} \xi} 
\nonumber \\
&\ls \norm{ \left( \Np +1 \right)^{\l+1/2} \psi} \norm{\left( \Np +1 \right)^{-\l + 1}  \xi}\,,
\end{align}
where we used Lemma \ref{lem:T:Number} for the last line. 
For the next term $A_{c_2}$, we compute the action of the Bogoliubov transformation on the annihilation and creation operators 
\begin{align}
\FockT \boldKt \FockT^*
&=  \frac{1}{2}
\sum_{p\in\Lsp} \widehat{v}_N^\beta (p)
\FockT  \ad_p\ad_{-p} \FockT^*
\nonumber \\
&= \frac{1}{2}
\sum_{p\in\Lsp} \widehat{v}_N^\beta (p)
\left(
c_p s_p + c_p^2  \ad_p \ad_{-p} 
+ 2 c_p s_p  \ad_p a_p 
+ s_p^2 a_p a_{-p}
\right) .
\end{align}
We denote $g_1(\Np):=\frac{\sqrt{\left[(N-\Np)(N-\Np-1)\right]_+} }{N} -1$ and observe that $g_1^2(\Np)\ls N^{-1}\Np$. 
Since $\sum_{p\in\Lsp} \widehat{v}_N^\beta (p)c_p s_p \ls  N^{\beta}$
by Lemma \ref{lem:aux:potential}, we obtain the bound
\begin{align}\label{eq:estimate l-1 norm potential with c-p and s-p}
&\abs{\scp{\psi}{\left(\mathbb{T} A_{c_2}\mathbb{T}^* - \frac{1}{2}
\sum_{p\in\Lsp} \hvNb(p)
c_p s_p \FockT g_1(\Np) \FockT^* \right)  \xi}}
\nonumber \\
&\quad=\abs{\scp{\psi}{\bigg( \FockT \boldKt \FockT^* - \frac{1}{2}
\sum_{p\in\Lsp} \widehat{v}_N^\beta (p) 
c_p s_p \bigg)\FockT g_1(\Np)\FockT^* \xi}}
\nonumber \\
&\quad \ls
\sum_{p\in\Lsp} \widehat{v}_N^\beta (p)
\abs{\scp{a_p a_{-p} \left( \Np +1 \right)^\l \psi}{ \left( \Np + 3 \right)^{-\l} \FockT g_1(\Np)\FockT^*\xi}}\nonumber\\
&\qquad + \sum_{p\in\Lsp} \abs{s_p} \widehat{v}_N^\beta (p)
\abs{\scp{\ad_p a_p \left(\Np + 1 \right)^\l \psi}{ \left(\Np + 1 \right)^{-\l} \FockT g_1(\Np)\FockT^*\xi}}
\nonumber \\
&\qquad
+ \sum_{p\in\Lsp} \abs{s_p}^2  \widehat{v}_N^\beta (p)
\abs{\scp{\ad_p \ad_{-p} \left( \Np + 3 \right)^{\l} \psi}{ \left( \Np + 1 \right)^{-\l} \FockT g_1(\Np)\FockT^*\xi}}
\nonumber \\
&\quad \ls 
\sum_{p\in\Lsp} \abs{p}^{-1} \widehat{v}_N^\beta (p)
\norm{ \abs{p} a_p \left( \Np +1 \right)^{\l + \frac{1}{2}} \psi} \norm{\left( \Np + 1 \right)^{-\l} \FockT g_1(\Np)\FockT^*\xi}
\nonumber \\
&\qquad
+\norm{\left( \Np + 1 \right)^{\l+1} \psi} \norm{\left( \Np + 1 \right)^{-\l} \FockT g_1(\Np)\FockT^*\xi}
\nonumber \\
&\ls N^{\beta/2} \norm{\mathbb{K}_0^{1/2} \left( \Np +1 \right)^{\l + \frac{1}{2}} \psi} \norm{\left( \Np + 1 \right)^{-\l} \FockT g_1(\Np)\FockT^*\xi}\nonumber\\
&\ls N^{\beta/2-1}\norm{\mathbb{K}_0^{1/2} \left( \Np +1 \right)^{\l + \frac{1}{2}} \psi} \norm{\left( \Np + 1 \right)^{1-\l}\xi}
\end{align}
by Lemma \ref{lem:T:Number}. Next we show that the additional term subtracted from $A_{c_2}$ above together with the last term of the r.h.s. of \eqref{def:R_sqrt:new} is small. For this we observe that 
\begin{align}
g_1(\Np) = -\frac{\Np + \frac{1}{2}}{N} +N^{-2}\widetilde{R}_1^{(2)}
\quad \text{with} \quad 
\norm{\widetilde{R}_1^{(2)} \psi} &\leq C \norm{\left( \Np + 1 \right)^2 \psi}
\end{align}
 and by Lemma \ref{lem:T:Number} 
\begin{align}
\abs{\scp{\psi}{\left( \FockT g_1(\Np) \FockT^* + N^{-1} \FockT \left( \Np + 1/2 \right) \FockT^* \right) \xi }}
 &= \abs{\scp{\psi}{\FockT \left(  g_1(\Np)  + N^{-1}  \left( \Np + 1/2 \right) \right) \FockT^* \xi }}
\nonumber \\
& \ls N^{-2} \norm{\left( \Np + 1 \right)^{\l+1}  \psi}
\norm{\left( \Np + 1 \right)^{-\l+1}   \xi} \; . 
\end{align}
In combination with \eqref{eq:estimate l-1 norm potential with c-p and s-p} this leads to
\begin{align}
&\frac{1}{2} \abs{\scp{\psi}{ \sum_{p\in\Lsp} \hvNb(p)
c_p s_p \left(  \FockT g_1(\Np) \FockT^* +  N^{-1} \FockT \left( \Np + 1/2 \right) \FockT^* \right) \xi}}
\nonumber \\
&\quad \ls N^{-2 + \beta} \norm{\left( \Np + 1 \right)^{\l+1}  \psi}
\norm{\left( \Np + 1 \right)^{-\l+1}  \xi}
\end{align} 
and thus altogether 
\begin{align}
&\abs{\scp{\psi}{\left( \mathbb{T}A_{c_2}\mathbb{T}^* + \frac{1}{2} N^{-1}
\sum_{p\in\Lsp} \hvNb(p)
c_p s_p \FockT \left( \Np + 1/2 \right) \FockT^* \right)  \xi}}
\nonumber \\
&\ls  N^{-1 + \beta/2} \norm{ \mathbb{K}_0^{1/2} \left(\Np + 1 \right)^{\l + 1/2} \psi} \norm{\left( \Np + 1 \right)^{-\l+1} \xi} .
\end{align}
The hermitian conjugate can be bounded similarly. 

Finally, observing that $g_2(\Np)^2\ls N^{-3/2}\Np$ for $g_2(\Np):=\frac{\sqrt{\left[N-\Np\right]_+} - \sqrt{N} }{N}$ and using Lemma \ref{lem:T:Number} again, we obtain 
\begin{align}
&\abs{\scp{\psi}{ \mathbb{T} A_{c_3}\mathbb{T}^*\xi}}
\nonumber \\
&\quad \ls  \sum_{p,q \in \Lsp} 
\hvNb(p)
\abs{\scp{\psi}{ \FockT \left( \Np + 1 \right)^{\l + 1/4} \ad_{p+q} \ad_{-p} a_q  g_2(\Np) (\Np + 2)^{-\l -1/4} \FockT^* \xi}}
\nonumber \\
&\quad \ls \left( \sum_{p,q \in \Lsp} \norm{a_{-p} a_{p+q} \left( \Np + 1 \right)^{\l + 1/4} \FockT^* \psi }^2 \right)^{1/2}\nonumber\\
&\quad \times
\left( \sum_{p,q \in \Lsp} \hvNb(p)
\norm{a_q g_2(\Np) (\Np + 2)^{-\l -1/4} \FockT^* \xi}^2 \right)^{1/2}
\nonumber \\
&\quad \ls C N^{3(- 1 + \beta)/2 }
\norm{\left( \Np + 1 \right)^{\l + 5/4} \psi }
\norm{\left( \Np + 1 \right)^{- \l + 5/4}  \xi }  \; .
\end{align}
As the hermitian conjugate can be estimated similarly, this concludes the proof of the lemma. 
\end{proof}

\section{Estimates for the excitation vector}\label{sec:excitation:vector}
The goal of this section is to prove that the quantity
$$\lr{\Chi,(\boldKz+1)(\Np+1)^\l\Chi}$$
is bounded uniformly in $N$. To show this, we first need a bound on the kinetic energy in terms of the generator $\FockG$.

\subsection{Estimate of the kinetic energy in terms of $\FockG$}

\begin{lem}\label{lem:G}
\begin{align}
\boldKz\ls \FockG + \left( \Np + 1 \right)^4 .
\end{align}
\end{lem}
\begin{proof}
Let $\psi\in\Fp$. Since $a_0\psi=0$, a straightforward computation yields
\begin{align}
\lr{\psi,\FockG_2\psi}
&=\lr{\psi,\left(\sum_{p,q,r\in\Lambda}\frac{\hvNb(p)}{2N} c_{p+r}c_qc_pc_{q+r}\ad_{p+r}\ad_qa_pa_{q+r}
-frac{\hat{v}(0)}{2N} \sum_{p,q\in\Lsp}c_p^2c_q^2\ad_p\ad_qa_pa_q\right)\psi}\nonumber\\
&=\lr{\psi,\int\dx\dy v(x-y)\ad(\check{c}_x)\ad(\check{c}_y) a(\check{c}_x) a(\check{c}_y)\psi}\nonumber\\
&\quad-\frac{\hat{v}(0)}{2N}\sum_{p,q\in\Lsp}c_p^2c_q^2\lr{\psi,\ad_p\ad_qa_pa_q\psi}\nonumber\\
&\geq -\frac{\hat{v}(0)}{2N}\sum_{p,q\in\Lsp}c_p^2c_q^2\lr{\psi,\ad_p\ad_qa_pa_q\psi}\,.\label{eqn:lower:bound:G2}
\end{align}
 We consequently have that
\begin{align}\label{eqn:lower:bound:G}
\FockG = \FockG_0 + \FockG_1 + \FockG_2 + \FockR_2
\geq  \FockG_0 + \FockG_1 + \FockR_2-\frac{\hat{v}(0)}{2N}\sum_{p,q\in\Lsp}c_p^2c_q^2\ad_p\ad_qa_pa_q
\end{align}
in the sense of operators on $\Fp$.

\paragraph{Estimate for $\FockG_0$.}

Using Lemma \ref{lemma:propF,G}, we obtain
\begin{align}
\frac12\scp{\psi}{\boldKz \psi} 
&\leq \scp{\psi}{\sum_{p \in \Lambda_+^*} F_p \ad_pa_p \psi}
\nonumber \\
&\leq  \lr{\psi,\FockGz\psi}+
\abs{\scp{\psi}{\sum_{p \in \Lambda_+^*} G_p ( \ad_p \ad_{-p}  + a_pa_{-p})  \psi}}\nonumber\\
&\leq \scp{\psi}{\FockG_0 \psi}+C\norm{\left( \Np + 1 \right)^{1/2} \psi}^2 \,. 
\end{align}

\paragraph{Estimate for the other terms.}
From Lemma \ref{lem:G1} and Proposition \ref{prop:renormalized:hamiltonian}, we obtain
\begin{align}
\abs{\scp{\psi}{\FockG_1 \psi}}
&\leq C N^{(\beta - 1)/2} \norm{(\boldKz+1)^{1/2}  \psi} \norm{\left( \Np + 1 \right) \psi} \,,\\
\abs{\scp{\psi}{\FockR_2 \psi}}
&\leq C N^{\frac32(\beta-1)}\norm{\left( \boldKz + 1 \right)^{1/2} \psi} \norm{\left( \Np + 1 \right)^2 \psi} .
\end{align}
Moreover,
\begin{align}
\left|\frac{\hat{v}(0)}{2N}\sum_{p,q\in\Lsp}c_p^2c_q^2\lr{\psi,\ad_p\ad_qa_pa_q\psi}\right|\leq C N^{-1}\norm{(\Np+1)\psi}^2\,.
\end{align}

\paragraph{Final estimate.}
Combining these estimates with \eqref{eqn:lower:bound:G} yields
\begin{align}
\frac{1}{2} \scp{\psi}{\boldKz \psi}  
&\leq \scp{\psi}{\FockG \psi}
+ \abs{\scp{\psi}{\left( \FockG_1 + \FockR_2 \right) \psi}}
+\left|\frac{\hat{v}(0)}{2N}\sum_{p,q\in\Lsp}c_p^2c_q^2\lr{\psi,\ad_p\ad_qa_pa_q\psi}\right|\nonumber\\
&\quad
+ C \norm{\left( \Np + 1 \right)^{1/2} \psi}^2
\nonumber \\
&\leq \scp{\psi}{\FockG \psi} + C \norm{\left( \Np + 1 \right)^2 \psi}^2
+ C \norm{\boldKz^{1/2} \psi} \norm{\left( \Np + 1 \right)^{2} \psi}
\nonumber \\
&\leq \scp{\psi}{\FockG \psi} + C \norm{\left( \Np + 1 \right)^2 \psi}^2
+ \frac{1}{4} \scp{\psi}{\boldKz \psi}  .
\end{align}
Subtracting the third summand on the right hand side from both sides of the inequality then shows the claim.
\end{proof}

\subsection{Kinetic energy of the excitations}
With this tool, we can now prove the main result of this section.

\begin{lem}\label{lem:kinetic:energy:excitations}
Recall that  $\Chi=\FockT(U_N\PsiN\oplus 0)$  and let $\l\geq 0$. Then 
\begin{equation}
\label{eq:lemma-excitation}
\lr{\Chi,(\Np+1)^\l\left(\boldKz+1\right)\Chi}\ls 1\,.
\end{equation}
\end{lem}

\begin{proof} We first show that any moment of the number of particles operator is bounded. We then use this bound to show that the kinetic energy is bounded, and as a last step estimate products of the kinetic energy and any power of the number of particles operator. 

\paragraph*{Powers of the number of particles.} We recall an a priori bound on powers of the number of excitations in the ground state $\PsiN$ that was recently proven in \cite[Remark 1.2]{nam2023}. More precisely, we have 
\begin{align}
\langle \PsiN,  \mathcal{N}_\perp^\ell  \PsiN \rangle_{L^2(\Lambda^N)} \leq C(\l)
\end{align}
for any power $\ell  \in \mathbb{N}$. Since $U_N \mathcal{N}_\perp U_N^* = \mathcal{N}_\perp$ and 
\begin{align}
\mathbb{T} \mathcal{N}_\perp^\ell \mathbb{T}^* \ls (\mathcal{N}_\perp + 1)^\ell
\end{align}
from Lemma \ref{lem:T:Number}, we consequently get
\begin{align}
\langle \Chi,  (\mathcal{N}_\perp+1)^\ell  \Chi \rangle \ls 1 \; . \label{eq:moments-chi}
\end{align}

\paragraph*{Kinetic energy.} From \eqref{eq:moments-chi}, we now derive a bound on the kinetic energy of $\Chi$. For this we recall that $\Chi$ is the ground state of $\mathbb{G}$, i.e. it satisfies the eigenvalue equation 
$\mathbb{G}  \Chi =  E \Chi  $
for $E=\ENb-\mathcal{C}$ as defined in \eqref{def:E}. From Lemma \ref{lem:G} we have the upper bound 
\begin{align}
\label{eq:formbound-1}
\mathbb{G}   \gtrsim \boldKz - (\mathcal{N}_\perp+1)^4 \,,
\end{align}
which we use to estimate 
\begin{align}
\langle\Chi, \boldKz \chi \rangle \ls \langle \Chi, \big( \mathbb{G} + (\mathcal{N}_\perp+1)^4 \big) \Chi \rangle = \langle \Chi, ( E + (\mathcal{N}_\perp+1)^4 ) \Chi \rangle\,.
\end{align}
Since $|E|\ls 1$ by \eqref{eqn:E-E0},  we conclude with \eqref{eq:moments-chi} that 
\begin{align}\label{eqn:K0}
\langle\Chi, \boldKz \Chi \rangle  \leq 1 \;  . 
\end{align}

\paragraph*{Moments.} Next we prove \eqref{eq:lemma-excitation}. For this we find with \eqref{eq:formbound-1} and \eqref{eq:moments-chi} 
\begin{align}
\langle \Chi, ( \mathcal{N}_\perp + 1)^{2(\ell+1)} \boldKz \Chi \rangle 
\ls& \langle \Chi, ( \mathcal{N}_\perp + 1)^{\ell+1} \big( \mathbb{G}   + (\mathcal{N}_\perp+1)^4 \big) ( \mathcal{N}_\perp + 1)^{\ell+1} \Chi \rangle \notag \\
\ls& \langle \Chi, ( \mathcal{N}_\perp + 1)^{\ell+1} \mathbb{G}    ( \mathcal{N}_\perp + 1)^{\ell+1} \Chi \rangle + 1  \notag \\
\ls& \langle \Chi,  ( \mathcal{N}_\perp + 1)^{\ell+1} \big[  \mathbb{G} , ( \mathcal{N}_\perp + 1)^{\ell+1}\big]  \Chi \rangle +  E \langle \Chi, ( \mathcal{N}_\perp + 1)^{2 (\ell+1)}  \Chi \rangle +1 \notag \\ 
\ls& \langle \Chi,  ( \mathcal{N}_\perp + 1)^{\ell+1} \big[  \mathbb{G} , ( \mathcal{N}_\perp + 1)^{\ell+1}\big]  \Chi \rangle +1 \; . 
\end{align}
Since 
\begin{align}
 \big[  \mathbb{G} , ( \mathcal{N}_\perp + 1)^{\ell+1}\big] =& \sum_{j=0}^{\ell +1} (\mathcal{N}_\perp+1)^j  \big[  \mathbb{G} ,  \mathcal{N}_\perp \big] (\mathcal{N}_\perp+1)^{\ell -j}  \notag \\
 =& \sum_{j=1}^{\ell +1} (\mathcal{N}_\perp+1)^j  \left[  \FockG_1+\tfrac12\sum_{p\in\Lsp}G_p(\ad_p\ad_{-p}+\hc)+\FockR_2,  \mathcal{N}_\perp \right] (\mathcal{N}_\perp+1)^{\ell -j}  \,,
\end{align}
we find from Lemma \ref{lem:G1}, Proposition \ref{prop:renormalized:hamiltonian} and \ref{lemma:propF,G}
\begin{align}
\lr{\Chi,(\Np+1)^j\left[\FockG_1,\Np\right](\Np+1)^{\l+1-j}\Chi}&\ls N^\frac{\beta-1}{2}\norm{(\boldKz+1)^\frac12\Chi}\norm{(\Np+1)^{\l+2}\Chi}\,,\\
\lr{\Chi,(\Np+1)^j\left[\FockR_2,\Np\right](\Np+1)^{\l+1-j}\Chi}&\ls N^{\frac32(\beta-1)}\norm{(\boldKz+1)^\frac12\Chi}\norm{(\Np+1)^{3+\l}\Chi}\,,
\end{align}
and
\begin{align}
&\tfrac12\sum_{p\in\Lsp}|G_p|\lr{\Chi,(\Np+1)^j\left[(\ad_p\ad_{-p}+\hc),\Np\right](\Np+1)^{\l+1-j}\Chi}\\
&\qquad\ls\norm{(\Np+1)^{\l+1}\Chi}\norm{(\Np+1)\Chi}\,.
\end{align}
Hence, by \eqref{eqn:K0} and \eqref{eq:moments-chi}, we  conclude  \eqref{eq:lemma-excitation}. 

\end{proof}

\section{Perturbation theory}\label{sec:pert:theory}
In this final section we complete the proof of the perturbative expansion of the ground state projector $\P=|\Chi\rangle\langle\Chi|$. As a first step, we expand $\P$ around $\Pz$; subsequently, we prove estimates on the remainders of this expansion (Proposition \ref{prop:expansion:P}).

\subsection{Proof of Proposition \ref{prop:expansion:P}}\label{sec:proof:prop:expansion:P}
Recall that Proposition \ref{prop:expansion:P} states that
$$\left|\Tr\FockA\P-\sum_{\l=0}^2\Tr\FockA\P_\l\right|\leq N^{\frac32 ( \beta-1)} $$
for $\P_1$ and $\P_2$ as in \eqref{def:Pl}. Our goal is now to compute these quantities and also the remainder term explicitly. In Sections \ref{sec:proof:prop:perturbation:estimates:P}-- \ref{sec:proof:prop:perturbation:estimates:Q}, we show that the remainder satisfies  the bound above. 

Since the strategy is very close to \cite[Proposition 3.13, Proposition 3.14 and Theorem 2]{spectrum}, we only sketch the main ideas.\\

\noindent\textbf{Proof of Proposition \ref{prop:expansion:P}.}
The first step is to express $\FockG$ as
\begin{equation}
\label{def:Ri}
\FockG=\FockGz+\FockG_1+\FockG_2+\FockR_2=:\FockGz+\FockG_1+\FockR_1=:\FockGz+\FockRz
\end{equation}
(see Proposition \ref{prop:renormalized:hamiltonian}). The remainder terms $\mathbb{R}_i$ correspond to errors of of different order in $N$. Lemmas \ref{lem:Ra}-\ref{lem:R2} show that $\mathbb{R}_2 = \mathbb{R}_a + \mathbb{R}_b + \mathbb{R}_c $ satisfies for $\l,\l_1,\l_2 \in \mathbb{R}$
\begin{align}
\label{eq:R2-estimate}
\left|\lr{\psi,\FockR_2\xi}\right|
&\ls   N^{\frac32(\beta- 1) }  
\bigg(\norm{(\mathbb{K}_0 +1 )^{1/2}\left(\Np + 1 \right)^{3/4-\l_1} \psi} \; 
\norm{ \left(\Np + 1 \right)^{5/4+\l_1} \xi} \notag \\
&\qquad\qquad\qquad+ \norm{(\left(\Np + 1 \right)^{5/4-\l_2} \psi}\; 
\norm{ (\mathbb{K}_0 +1 )^{1/2}\left(\Np + 1 \right)^{3/4+\l_2} \xi} \bigg)\; . 
\end{align}
From Lemma \ref{lem:G2}-\ref{lem:R2}, it follows that $\mathbb{R}_1 = \mathbb{G}_2 + \mathbb{R}_2$ scales as 
\begin{align}
\label{eq:R1-estimate}
\left|\lr{\psi,\FockR_1\xi}\right|
&\ls   N^{\beta- 1}  
\norm{(\mathbb{K}_0 +1 )^{1/2}\left(\Np + 1 \right)^{3/4-\l} \psi} \; 
\norm{(\mathbb{K}_0 +1 )^{1/2} \left(\Np + 1 \right)^{3/4+\l} \xi} \,.
\end{align}
The remainder term $\mathbb{R}_0 = \mathbb{G}_1 + \mathbb{G}_2 + \mathbb{R}_2$ is by Lemmas \ref{lem:G1}-\ref{lem:R2} of order 
\begin{align}
\label{eq:R0-estimate}
\left|\lr{\psi,\FockR_0\xi}\right|
&\ls   N^{(\beta- 1)/2}  
\norm{(\mathbb{K}_0 +1 )^{1/2}\left(\Np + 1 \right)^{3/4-\l} \psi} \; 
\norm{ (\mathbb{K}_0 +1 )^{1/2}\left(\Np + 1 \right)^{3/4+\l} \xi}\; . 
\end{align}

We now expand $(z-\FockG)^{-1}$ around $(z-\FockGz)^{-1}$ as
\begin{align}
\RG
&=\Rz+\Rz\FockG_1\Rz + \Rz\left(\FockG_2+\FockG_1\Rz\FockG_1\right)\Rz\nonumber\\
&\quad +\RG\left(\FockR_0\Rz(\FockG_2+\FockG_1\Rz\FockG_1)+\FockR_1\Rz\FockG_1+\FockR_2\right)\Rz\,.
\end{align}
Then \eqref{def:Pl}  follows from an application of the residue theorem analogously to \cite[Proposition 3.14]{spectrum}. 
Note that $\FockG_1$ is cubic in the number of creation/annihilation operators and $\Chiz$ is a quasi-free state, hence $\lr{\Chiz,\FockG_1\Chiz}=0$ by Wick's rule.
Analogously to \cite[Proposition 3.14]{spectrum}, we find that
\begin{align}
\Tr\FockA\P=\sum_{\l=0}^2\Tr\FockA\P_\l +\Tr\FockA\FockB_P +\Tr\FockA\FockB_Q\,,
\end{align}
where the remainder $\Tr\FockA\FockB_P$ is given as
\begin{align}
\Tr\FockA\FockB_P
&=\goint'\frac{1}{z-E}\bigg(\lr{\Chi,\FockR_0\Rz\left(\FockG_2+\FockG_1\Rz\FockG_1\right)\Rz\FockA\Chi}\notag\\
&\qquad+\lr{\Chi,\FockR_1\Rz\FockG_1\Rz\FockA\Chi}
+\lr{\Chi,\FockR_2\Rz\FockA\Chi}\Bigg)\dz\,,\label{TrAB_P}
\end{align}
where we used the abbreviation $\goint':=\frac{1}{2\pi\i}\goint$. 
For  $\Tr\FockA\FockB_Q $, we find analogously to the computation in \cite{spectrum} that
\begin{align}
\Tr\FockA\FockB_Q
&=\goint'\frac{\dz}{z-\Ez}\Bigg(
\lr{\FockG_2\RQz\FockA\RQG\FockRz} + \lr{\FockA\RQG\FockRz\RQz\FockG_2} \notag\\
&\qquad\qquad\qquad+ \lr{\FockG_1\RQz\FockG_1\RQz\FockA\RQG\FockRz} 
+ \lr{\FockG_1\RQz\FockA\RQG\FockRz\RQz\FockG_1}\notag\\
&\qquad\qquad\qquad+  \lr{\FockA\RQG\FockRz\RQz\FockG_1\RQz\FockG_1} 
+\lr{\FockG_1\RQz\FockA\RQG\FockR_1}\notag\\
&\qquad\qquad\qquad+ \lr{\FockA\RQG\FockR_1\RQz\FockG_1} +\lr{\FockA\RQG\FockR_2}
\Bigg)\notag\\
&+\goint'\frac{\dz}{(z-\Ez)^2}\Bigg(
\lr{\FockG_2}\lr{\FockA\RQG\FockRz} + \lr{\FockG_1}\lr{\FockG_1\RQz\FockA\RQG\FockRz}\notag\\
&\qquad\qquad\qquad+  \lr{\FockG_1\RQz\FockG_1}\lr{\FockA\RQG\FockRz}
+\lr{\FockG_1}\lr{\FockA\RQG\FockRz\RQz\FockG_1}\notag\\
&\qquad\qquad\qquad+\lr{\FockG_1}\lr{\FockA\RQG\FockR_1}\Bigg)\notag\\
&+\goint'\frac{\dz}{(z-\Ez)^3}\lr{\FockG_1}^2\lr{\FockA\RQG\FockRz}\,,\label{TrAB_Q}
\end{align}
where we used the notation
\begin{equation}
\lr{\FockB}:=\lr{\Chiz,\FockB\Chiz}\,.
\end{equation}
Recall that $\lr{\FockG_1}=0$ by Wick's rule, hence some of the terms in \eqref{TrAB_Q} vanish.
Since $z\in\gamma$ with $\gamma$ as defined in \eqref{gamma}, it follows that $|z-\Ez|= c$ and $|z-E|\geq |z-\Ez|-|E-\Ez|\geq c/2$ for sufficiently large $N$, hence 
\begin{equation}\label{eqn:z-E}
\sup_{z\in\gamma}\left|\frac{1}{z-\Ez}\right|\ls 1\,,\qquad \sup_{z\in\gamma}\left|\frac{1}{z-E}\right|\ls 1\,.
\end{equation} 
We use this argument to estimate $\Tr \mathbb{A}\mathbb{B}_P$ by \eqref{RPA} in Section \ref{sec:proof:prop:perturbation:estimates:P} and $\Tr \mathbb{A}\mathbb{B}_Q$  by \eqref{RQA} in Section \ref{sec:proof:prop:perturbation:estimates:Q}. The final estimates of Section \ref{sec:proof:prop:perturbation:estimates:P} and Section \ref{sec:proof:prop:perturbation:estimates:Q} (see \eqref{eq:RP-final} and \eqref{eq:RQ-final}) show that 
\begin{align}
\vert \Tr \mathbb{A}\mathbb{B}_P \vert +\vert  \Tr \mathbb{A}\mathbb{B}_Q \vert \ls N^{\frac32(\beta-1)}  \;, \label{eq:estimate:TRAP,TRAQ}
\end{align}
which proves the first statement \eqref{eqn:expansion:TrAP} of Proposition \ref{prop:expansion:P}. 
\\

\noindent\textbf{Ground state energy.} We observe analogously to \cite[Theorem 2]{spectrum} that
\begin{equation}
E=\Tr\,\FockG\,\P=\frac{1}{2\pi\i}\Tr\,\goint\frac{\FockG}{z-\FockG}\dz=\frac{1}{2\pi\i}\Tr\,\goint\frac{z}{z-\FockG}\dz =E_0+\frac{1}{2\pi\i}\Tr\goint\frac{z-\Ez}{z-\FockG}\dz\,.
\end{equation}
Consequently, expanding $(z-\FockG)^{-1}$ around $(z-\FockGz)^{-1}$ as above yields
\begin{align}
E -\Ez&=\Tr\goint'\Rz\FockG_1\frac{z-\Ez}{z-\FockGz}\dz\notag\\
&\quad + \Tr\goint'\Rz\FockG_2\frac{z-\Ez}{z-\FockGz} +\Tr\goint'\Rz\FockG_1\Rz\FockG_1\frac{z-\Ez}{z-\FockGz}\dz\notag\\
&\quad +\sum_{\nu=0}^2\sum_{m=1}^{2-\nu}\sum_{\substack{\bj\in\N^m\\|\bj|=2-\nu}}\Tr\goint'\RG\FockR_\nu\Rz\FockG_{j_1}\Rz\FockG_{j_2}\mycdots \FockG_{j_m}\frac{z-\Ez}{z-\FockGz}\dz\label{eqn:E-Ez:long}
\end{align}
For the first term, we observe that
\begin{equation}
\Tr\goint'\Rz\FockG_1\frac{z-\Ez}{z-\FockGz}\dz=\goint'\frac{1}{z-\Ez}\lr{\Chiz,\FockG_1\Chiz}\dz =0
\end{equation}
by Wick's rule. Applying the residue theorem, we find that the second and third term give $E_\mathrm{pert}$ as in \eqref{E_pert}. This is of the same form as $\Tr\FockA \FockB_P+\Tr\FockA\FockB_Q$ but with $\FockA$ replaced by $(z-\Ez)$ in each term in \eqref{TrAB_P} and \eqref{TrAB_Q}. Together with the estimate \eqref{eqn:z-E}, this implies that the absolute value of the third line in \eqref{eqn:E-Ez:long} is  bounded by 
\begin{align}
\vert \Tr \id\mathbb{B}_P \vert +\vert  \Tr \id\mathbb{B}_Q \vert \leq N^{\frac32(\beta-1)} 
\end{align}
following from \eqref{eq:estimate:TRAP,TRAQ}. This proves the second statement \eqref{eq:expansion:energy} of Proposition \ref{prop:expansion:P}. 
\qed

\subsection{Bounds on the resolvent of $\FockGz$}
In this section we prove estimates on the resolvent of the quadratic Hamiltonian $\mathbb{G}_0$.

\begin{lem}\label{lem:aux}
Let $\l\in\mathbb{R}$, $\psi\in\Fp$ and $z\in\gamma$ for $\gamma$ as in \eqref{gamma} and let 
$$\FockO\in\left\{\Rz\,,\,\frac{\Qz}{\Ez-\FockGz}\right\}.$$ 
Then
\begin{align}
\norm{(\boldKz+1)^\frac12 (\Np+1)^\l \FockO (\mathbb{K}_0+1)^{\frac{1}{2}}\psi}\ls\norm{(\Np+1)^{\l}\psi}\label{eqn:aux:1}\,.
\end{align}
\end{lem}

\begin{proof}
We consider first the case $\FockO=(z-\FockGz)^{-1}$.
First note that by diagonalizing $\mathbb{G}_0$ using $\mathbb{U}_\tau$ we get 
\begin{align}
\mathbb{U}_\tau \mathbb{G}_0\mathbb{U}_\tau^* \geq c \mathcal{N}_\perp - C   \; .\label{eq:boung-G01}
\end{align}
We multiply \eqref{eq:boung-G01} with $\mathbb{U}_\tau^*$ from the left and $\mathbb{U}_\tau$ from the right and arrive with $\mathbb{U}_\tau^* (\mathcal{N}_\perp +1)\mathbb{U}_\tau \geq \tfrac1c \mathcal{N}_\perp $ from Lemma \ref{lem:BT:K0} at 
\begin{align}
\label{eq:bound-G02}
\mathbb{G}_0 \geq c_1 \mathcal{N}_\perp  - c_2 
\end{align}
for some  constants $c_1,c_2>0$. Moreover, Lemma \ref{lemma:propF,G} shows that 
\begin{align}
\label{eq:bound-G03}
\mathbb{G}_0 \geq c_3  \boldKz - c_4 (\mathcal{N}_\perp + 1)  \; . 
\end{align}
Interpolating both from bounds \eqref{eq:bound-G02}, \eqref{eq:bound-G03} we arrive at 
\begin{align}
\label{eq:bound-G04}
\mathbb{G}_0 \geq c_5 ( \boldKz + \mathcal{N}_\perp) - c_6 \,,
\end{align}
which implies \eqref{eqn:aux:1} for $\ell =0$ and $\Im z =0$. For $\ell =0$ and $\Im z \not=0$ we find 
\begin{align}
\label{eq:G04-1}
\| & ( \mathbb{K}_0 + 1)^{\frac{1}{2}} \frac{1}{z-\mathbb{G}_0} ( \mathbb{K}_0 + 1)^{\frac{1}{2}} \psi \|^2 \notag \\
=& \left\langle ( \mathbb{K}_0 + 1)^{\frac{1}{2}}  \psi, \frac{\mathbb{G}_0 - \Re z + \i \Im z}{( \mathbb{G}_0 - \Re z )^2 + (\Im z)^2 } ( \mathbb{K}_0 + 1) \frac{\mathbb{G}_0 - \Re z + \i \Im z}{( \mathbb{G}_0 - \Re z )^2 + (\Im z)^2 }   ( \mathbb{K}_0 + 1)^{\frac{1}{2}}  \psi\right\rangle  \, . 
\end{align} 
From \eqref{eq:bound-G04} we get 
\begin{align}
\Big\| ( \mathbb{K}_0 + 1)^\frac{1}{2}  & \frac{\mathbb{G}_0 - \Re z}{( \mathbb{G}_0 - \Re z )^2 + (\Im z)^2 }   ( \mathbb{K}_0 + 1)^{\frac{1}{2}}  \psi \Big\| \notag \\
\ls&  \Big\|   \frac{(\mathbb{G}_0 - \Re z)^{3/2} + 1}{( \mathbb{G}_0 - \Re z )^2  }   ( \mathbb{K}_0 + 1)^{\frac{1}{2}}  \psi \Big\| 
\ls \| \psi \| \; 
\end{align}
and
\begin{align}
\Big\| ( \mathbb{K}_0 + 1)^\frac{1}{2}  & \frac{1}{( \mathbb{G}_0 - \Re z )^2 + (\Im z)^2 }   ( \mathbb{K}_0 + 1)^{\frac{1}{2}}  \psi \Big\| \notag \\
\ls&  \Big\|   \frac{(\mathbb{G}_0 - \Re z)^{1/2} + 1}{( \mathbb{G}_0 - \Re z )^2 }   ( \mathbb{K}_0 + 1)^{\frac{1}{2}}  \psi \Big\| 
\ls \| \psi \| \; 
\end{align}
which implies \eqref{eqn:aux:1} for $\l=0$ and $z\in\gamma$.
To prove \eqref{eqn:aux:1} for general $\ell \in \mathbb{R}$, we use once more that $\BogUz$ diagonalizes $\FockGz$ and compute with Lemma \ref{lem:BT:K0}
\begin{align} 
\norm{\boldKz^\frac12\Np^\l \Rz (\boldKz+1)^\frac12 \psi}
&=\norm{ \boldKz^\frac12 \Np^\l\BogUz^* \frac{1}{z- \mathbb{U}_\tau\mathbb{G}_0\mathbb{U}_\tau^*} (\mathbb{U}_\tau \boldKz \mathbb{U}_\tau^* +1)^{\frac{1}{2}}\BogUz\psi}\nonumber\\
&\ls\norm{\boldKz^\frac12 \Np^\l \frac{1}{z- \mathbb{U}_\tau\mathbb{G}_0\mathbb{U}_\tau^*} (\mathbb{U}_\tau \boldKz \mathbb{U}_\tau^* +1)^{\frac{1}{2}} \BogUz\psi} \; . 
\end{align}
As the number of particles operator commutes with the diagonal Hamiltonian $\mathbb{U}_\tau\mathbb{G}_0\mathbb{U}_\tau^*$,  we thus get, by \eqref{eq:bound-G04} and using Lemma \ref{lem:BT:K0} several times, that
\begin{align}
\norm{\boldKz^\frac12\Np^\l \Rz (\boldKz+1)^\frac12 \psi}
&\ls\norm{\boldKz^\frac12 \frac{1}{z- \mathbb{U}_\tau\mathbb{G}_0\mathbb{U}_\tau^*}  \Np^\l\mathbb{U}_\tau (\boldKz+1)^{\frac{1}{2}} \psi}  \notag \\
& \ls\norm{\boldKz^\frac12 \frac{1}{z- \mathbb{G}_0}(\boldKz+1)^\frac12\Big(\frac{1}{\boldKz+1}\Big)^{\frac12}\BogUz^*\Np^\l  \mathbb{U}_\tau (\boldKz+1)^\frac{1}{2} \psi}  \notag \\
&\ls \Big\|\Big(\frac{1}{\boldKz+1}\Big)^\frac12\BogUz^*\Np^\l\BogUz(\boldKz+1)^\frac12\psi\Big\|\nonumber\\
&\ls \norm{ \Np^\l\BogUz\psi} \,.
\end{align}

For $\FockO=\Qz(\Ez-\FockGz)^{-1}$ and $\l=0$, we find similarly to above that
\begin{align}
\norm{(\boldKz+1)^\frac12\frac{\Qz}{\Ez-\FockGz}(\boldKz+1)^\frac12\psi}^2
&\ls \lr{\psi,(\boldKz+1)^\frac12\frac{\Qz}{\Ez-\FockGz}(\boldKz+1)^\frac12\psi}\nonumber\\
&\ls\lr{\psi,(\boldKz+1)^\frac12\Qz\frac{1}{\boldKz-1}\Qz(\boldKz+1)^\frac12\psi}\nonumber\\
&\ls\norm{\psi}^2
\end{align}
where the last step follows from decomposing $\Qz=\id-\Pz$ and using that $\norm{(\boldKz+1)^{1/2}\Chiz}\ls 1$ by Lemma \ref{lem:BT:K0}. The case $\l\neq 0$ is analogous to the computation for $\FockO=(z-\FockGz)^{-1}$.
\end{proof}

Next, we combine Lemma \ref{lem:aux} and the estimates on the remainder terms $\mathbb{R}_i$ and $\mathbb{G}_1$ im \eqref{eq:R0-estimate}-\eqref{eq:R2-estimate} resp. Lemma \ref{lem:G1} to derive the following weigthed norm estimates. 

\begin{lem}
\label{lem:aux:2}
Let $\l\in\mathbb{Z}/2$, $\psi\in\Fp$ and $z\in\gamma$ for $\gamma$ as in \eqref{gamma} and let 
$$\FockO\in\left\{\Rz\,,\,\frac{\Qz}{\Ez-\FockGz}\right\}\,.$$
Then
\begin{align}
\norm{(\boldKz+1)^\frac12(\Np+1)^\l\FockO\,\FockRz\psi}&\ls N^\frac{\beta-1}{2}\norm{(\boldKz+1)^\frac12(\Np+1)^{\l+\frac32}\psi}\,,\label{eqn:aux:2:1}\\
\norm{(\boldKz+1)^\frac12(\Np+1)^\l\FockO\,\FockG_1\psi}&\ls N^\frac{\beta-1}{2}\norm{(\boldKz+1)^\frac12(\Np+1)^{\l+1}\psi}\,,\label{eqn:aux:2:3}\\
\norm{(\boldKz+1)^\frac12(\Np+1)^\l\FockO\,\FockR_1\psi}&\ls N^{\beta-1}\norm{(\boldKz+1)^\frac12(\Np+1)^{\l+\frac32}\psi}\,.\label{eqn:aux:2:2} \\
\norm{(\boldKz+1)^\frac12(\Np+1)^\l\FockO\,\FockG_2\psi}&\ls N^{\beta-1}\norm{(\boldKz+1)^\frac12(\Np+1)^{\l+\frac32}\psi}\,.\label{eqn:aux:2:4}
\end{align}
\end{lem}
\begin{proof}
We prove the statement for $\FockO=(z-\FockGz)^{-1}$, the case $\FockO=\Qz(\Ez-\FockGz)^{-1}$ works analogously.
We start with the proof of \eqref{eqn:aux:2:1}. By \eqref{eq:R0-estimate} we find that
\begin{align}
&\norm{(\boldKz+1)^\frac12(\Np+1)^\l\Rz\FockRz\psi}^2\nonumber\\
&\quad=\lr{\Rz(\boldKz+1)(\Np+1)^{2\l}\Rz\FockRz\psi,\FockRz\psi}\nonumber\\
&\quad\ls N^\frac{\beta-1}{2}\norm{(\boldKz+1)^\frac12(\Np+1)^{\frac32+\l}\psi} \nonumber\\
&\qquad\times
\norm{(\boldKz+1)^\frac12(\Np+1)^{-\l}\Rz(\boldKz+1)(\Np+1)^{2\l}\Rz\FockRz\psi}\,.
\end{align}
From Lemma \ref{lem:aux} we get 
\begin{align}
&\norm{(\boldKz+1)^\frac12(\Np+1)^{-\l}\Rz(\boldKz+1)(\Np+1)^{2\l}\Rz\FockRz\psi}\nonumber\\
&\quad \ls \norm{(\Np+1)^\l(\boldKz+1)^\frac12\Rz\FockRz\psi}
\end{align}
and thus arrive at 
\begin{align}
\norm{(\boldKz+1)^\frac12(\Np+1)^\l\Rz\FockRz\psi} \ls N^\frac{\beta-1}{2}\norm{(\boldKz+1)^\frac12(\Np+1)^{\frac32+\l}\psi}\,.
\end{align}
The other two estimates follows same with similar ideas using Lemma \ref{lem:G1} and \eqref{eq:R2-estimate} instead of \eqref{eq:R0-estimate}. 
\end{proof}

\subsection{Bounds on the resolvent of $\FockG$}
Finally, we also need to control expressions involving a resolvent of the full generator $\FockG$.

\begin{lem}\label{lem:resolventG}
Let $a \geq 5/2$, $\psi\in\Fp$ and $z\in\gamma$ for $\gamma$ as in \eqref{gamma}. Then
\begin{align}
\norm{(\boldKz+1)^\frac12 \frac{1}{(\Np + 1 )^{a}} \frac{\Q}{z - \FockG} \psi}
&\ls \norm{\psi} .
\end{align}
\end{lem}

\begin{proof}
Let $a \geq 5/2$.
From \eqref{eq:bound-G04} and $\FockG_2 \geq -\tilde{\FockR}$  with $\tilde{\FockR}:=\frac{\hat{v}(0)}{2N}\sum_{p,q\in\Lsp}c_p^2c_q^2\ad_p\ad_qa_pa_q$ as in \eqref{eqn:lower:bound:G2}, we get 
\begin{align}
&\norm{\boldKz^{1/2} \frac{1}{(\Np + 1 )^{a}} \frac{\Q}{z - \FockG} \psi}^2
\nonumber \\
&\quad \ls  
\scp{\frac{1}{(\Np + 1 )^{a}} \frac{\Q}{z - \FockG} \psi}{\left(  \mathbb{G}_0 + \FockG_2 +\tilde{\FockR} \right) \frac{1}{(\Np + 1 )^{a}} \frac{\Q}{z - \FockG} \psi}
\nonumber \\
&\quad = 
\scp{\frac{1}{(\Np + 1 )^{2 a}} \frac{\Q}{z - \FockG} \psi}{\left(  \FockG_0 + \FockG_2 +\tilde{\FockR} \right)  \frac{\Q}{z - \FockG} \psi}
\nonumber \\
&\qquad + 
\abs{\scp{\frac{1}{(\Np + 1 )^{a}} \frac{\Q}{z - \FockG} \psi}{\left(  \sum_{p \in \Lambda_+^*} G_p ( \left[ (\mathcal{N}_+ +3)^{-a} - (\mathcal{N}_+ +1)^{-a} \right]a_pa_{-p} + {\rm h.c.})   \right)  \frac{\Q}{z - \FockG} \psi}}\,.
\end{align}
Since $G_p \in \ell^2(\Lsp)$ from Lemma \ref{lemma:propF,G}, the last term of the r.h.s. is bounded by a constant times $\norm{\psi}^2$. Moreover, since $\FockG_0  + \FockG_2 = \FockG - \FockG_1 -  \FockR_2 =  \FockG - z + z - \FockG_1   - \FockR_2$,  we get
\begin{align}
\norm{\boldKz^{1/2} \frac{1}{(\Np + 1 )^{a}} \frac{\Q}{z - \FockG} \psi}^2
& \ls  \norm{\psi}^2 
\nonumber \\
\label{eq:estimate for full resolvent 1}
&\quad +  
\abs{\scp{\frac{1}{(\Np + 1 )^{2 a}} \frac{\Q}{z - \FockG} \psi}{\FockG_1 \frac{\Q}{z - \FockG} \psi}}
\\
\label{eq:estimate for full resolvent 2}
&\quad +  
\abs{\scp{\frac{1}{(\Np + 1 )^{2 a}} \frac{\Q}{z - \FockG} \psi}{\FockR_2  \frac{\Q}{z - \FockG} \psi}}\\
\label{eq:estimate for full resolvent 3}
&\quad + \abs{\scp{\frac{1}{(\Np + 1 )^{2 a}} \frac{\Q}{z - \FockG} \psi}{\tilde{\FockR} \frac{\Q}{z - \FockG} \psi}}\; . 
\end{align}
We estimate the three terms separately and start with Lemma \ref{lem:G1} by 
\begin{align}
\eqref{eq:estimate for full resolvent 1}
&\ls  N^{\frac{\beta -1}{2}} 
\norm{\left( \boldKz + 1 \right)^{1/2} (\Np + 1 )^{1 - 2 a} \frac{\Q}{z - \FockG} \psi}
\norm{ \frac{\Q}{z - \FockG} \psi}
\nonumber \\
&\quad +  N^{\frac{\beta -1}{2}} 
\norm{ (\Np + 1 )^{1 + a - 2 a} \frac{\Q}{z - \FockG} \psi}
\norm{\left( \boldKz + 1 \right)^{1/2} (\Np + 1 )^{- a} \frac{\Q}{z - \FockG} \psi}
\nonumber \\
&\ls N^{\frac{\beta -1}{2}} 
\norm{\left( \boldKz + 1 \right)^{1/2} (\Np + 1 )^{- a} \frac{\Q}{z - \FockG} \psi}
\norm{ \psi} .
\end{align}
From \eqref{eq:R2-estimate} we furthermore get 
\begin{align}
\eqref{eq:estimate for full resolvent 2}
&\ls  N^{\frac32(\beta- 1) }  
\bigg[ \norm{\left( \boldKz + 1 \right)^{1/2} (\Np + 1 )^{2 - 2 a} \frac{\Q}{z - \FockG} \psi}
\norm{ \frac{\Q}{z - \FockG} \psi}
\nonumber \\
&\hspace{3cm}  +
\norm{ (\Np + 1 )^{2 + a - 2 a} \frac{\Q}{z - \FockG} \psi}
\norm{\left( \boldKz + 1 \right)^{1/2} (\Np + 1 )^{- a} \frac{\Q}{z - \FockG} \psi}
\bigg]
\nonumber \\
&\ls  N^{\frac32(\beta- 1) }  
\norm{\left( \boldKz + 1 \right)^{1/2} (\Np + 1 )^{- a} \frac{\Q}{z - \FockG} \psi} 
\norm{\psi} .
\end{align}
Finally,
\begin{align}
\eqref{eq:estimate for full resolvent 3}
&\ls N^{-1}\sum_{p,q\in\Lsp}\norm{(\Np+1)a_pa_q(\Np+1)^{-a}\frac{\Q}{z-\FockG}\psi}\norm{(\Np+1)^{-1}a_pa_q\frac{\Q}{z-\FockG}\psi}\nonumber\\
&\ls N^{-1}\norm{\psi}^2\,.
\end{align}
In total, this shows
\begin{align}
\norm{\boldKz^{1/2} \frac{1}{(\Np + 1 )^{a}} \frac{\Q}{z - \FockG} \psi}^2
\ls& \norm{\psi}^2 
+   N^{\frac{\beta- 1}{2} } 
\norm{\left( \boldKz + 1 \right)^{1/2} (\Np + 1 )^{- a} \frac{\Q}{z - \FockG} \psi} 
\norm{\psi} \notag \\
&\ls \norm{\psi}^2  + \frac{1}{2}
\norm{\boldKz^{1/2} (\Np + 1 )^{- a} \frac{\Q}{z - \FockG} \psi}^2 .
\end{align}
\end{proof}

\subsection{Estimate of $\Tr \mathbb{A}\mathbb{B}_P $}\label{sec:proof:prop:perturbation:estimates:P}

By definition of $\Tr \mathbb{A} \mathbb{B}_P$ in \eqref{TrAB_P}, it follows from \eqref{eqn:z-E} that
\begin{align}
\label{RPA}
\vert \Tr \mathbb{A} \mathbb{B}_P \vert 
\ls&\sup_{z\in\gamma}\bigg(\left|\lr{\Chi,\FockR_0\Rz\FockG_2\Rz\FockA\Chi}\right| + \left|\lr{\Chi,\FockR_0\Rz\FockG_1\Rz\FockG_1\Rz\FockA\Chi}\right|\notag\\
&+\left|\lr{\Chi,\FockR_1\Rz\FockG_1\Rz\FockA\Chi} \right| +\left|\lr{\Chi,\FockR_2\Rz\FockA\Chi}\right|\Bigg)\,. \notag \\
&=: \cR_{P1} +\cR_{P2}+\cR_{P3}+\cR_{P4} \; . 
\end{align}
For the first term, Lemmas \ref{lem:G1} and\ref{lem:aux} lead to the estimate
\begin{align}
|\cR_{P1}|&\ls \sup_{z\in\gamma}\left|\lr{\Chi,\FockG_1\Rz\FockG_2\Rz\mathbb{A}\Chi}\right|\nonumber\\
&\ls N^{\beta-1}\sup_{z\in\gamma}\norm{\Np^\frac32\boldKz^\frac12\Rz\Chi}\norm{\Np^{-\frac12}\boldKz^\frac12\Rz\FockG_1\mathbb{A}\Chi}\nonumber\\
&\ls N^{\beta-1}\sup_{z\in\gamma}\norm{(\Np+1)^2 (\mathbb{K}_0 + 1)^{-\frac{1}{2}}\Chi} \norm{(\boldKz+1)^{-\frac12}(\Np+1)^{-\frac12} \FockG_1 \mathbb{A} \Chi}\nonumber\\
&\ls N^{\frac32(\beta-1)}\norm{(\boldKz+1)^\frac12(\Np+1)^2\Chi}^2 \,,
\end{align}
where we used for the last two estimates Lemma \ref{lem:aux:2}, the bound $\| \mathbb{A} \psi \| \ls \| A \| \| \mathcal{N}_\perp \psi \|$ and Lemma \ref{lem:kinetic:energy:excitations}. 
For the second term, we find similarly with Lemmas \ref{lem:G1}, \ref{lem:aux} and \ref{lem:aux:2} 
\begin{align}
|\cR_{P2}|&\ls\sup_{z\in\gamma}\left|\lr{\Chi,\FockRz\Rz\FockG_1\Rz\FockG_1\Rz\mathbb{A}\Chi}\right|\nonumber\\
&\ls N^\frac{\beta-1}{2}\sup_{z\in\gamma}\norm{(\Np+1)^\frac12(\boldKz+1)^\frac12\Rz\FockRz\Chi}\notag\\
&\qquad\qquad\times\norm{(\Np+1)^\frac12(\boldKz+1)^\frac12\Rz\FockG_1\Rz\mathbb{A}\Chi}\nonumber\\
&\ls N^{\frac32(\beta-1)}\norm{(\boldKz+1)^\frac12(\Np+1)^3\Chi}^2\ls N^{3(\beta-1)/2}\,,
\end{align}
where the last estimate is again a consequence of Lemma \ref{lem:kinetic:energy:excitations}. For the third term we proceed analogously and find 
\begin{align}
|\cR_{P3}|&\ls\sup_{z\in\gamma}\left|\lr{\Chi,\FockR_1\Rz\FockG_1\Rz\mathbb{A}\Chi}\right|\nonumber\\
&\ls N^\frac{\beta-1}{2}\sup_{z\in\gamma}\norm{(\Np+1)(\boldKz+1)^\frac12\Rz\Chi}\norm{(\boldKz+1)^\frac12\Rz\FockR_1\mathbb{A}\Chi}\nonumber\\
&\ls N^{\frac32(\beta-1)}\norm{(\boldKz+1)^\frac12(\Np+1)^3\Chi}^2\ls N^{3( \beta-1)/2} \; . 
\end{align}
Finally, the last term is given by
\begin{align}
|\cR_{P3}|\ls\sup_{z\in\gamma}\left|\lr{\Chi,\FockR_2\Rz\mathbb{A}\Chi}\right|\ls N^{\frac32(\beta-1)}\norm{(\boldKz+1)^\frac12(\Np+1)^\frac52 \Chi}^2\,.
\end{align}
In summary, we find that
\begin{align}
\left|\Tr\FockA\FockB_P\right|\ls N^{3(\beta-1)/2} \; . \label{eq:RP-final}
\end{align}

\subsection{Estimate of $\Tr \mathbb{A}\mathbb{B}_Q$}\label{sec:proof:prop:perturbation:estimates:Q}

In this section we estimate $\Tr \mathbb{A}\mathbb{B}_P$ defined in \eqref{TrAB_Q}. It follows from \eqref{eqn:z-E} that
\begin{subequations}\label{RQA}
\begin{align}
|\Tr \mathbb{A}\mathbb{B}_Q|& \ls \sup_{z\in\gamma} \Bigg(
\left| \lr{\Chiz,\mathbb{A}\RQG\FockR_0\Chiz}\lr{\Chiz,\FockG_2\Chiz}\right| \label{eq:RQa}\\
&\hspace{3cm} + \left|\lr{\Chiz,\mathbb{A}\RQG\FockR_0\RQz\FockG_2\Chiz}\right|  \label{eq:RQb} \\
&\hspace{3cm} + \left| \lr{\Chiz,\FockG_2\RQz\mathbb{A}\RQG\FockR_0\Chiz}\right|\label{eq:RQc} \\
&\hspace{3cm} + \left| \lr{\Chiz,\mathbb{A}\RQG\FockR_0\Chiz} \lr{\Chiz\FockG_1\RQz\FockG_1\Chiz} \right|\label{eq:RQf} \\
&\hspace{3cm} + \left| \lr{\Chiz,\mathbb{A}\RQG\FockR_0\RQz\FockG_1\RQz\FockG_1\Chiz} \right| \label{eq:RQh}\\
&\hspace{3cm} + \left| \lr{\Chiz,\FockG_1\RQz\mathbb{A}\RQG\FockR_0\RQz\FockG_1\Chiz} \right|\label{eq:RQi} \\
&\hspace{3cm} + \left| \lr{\Chiz,\FockG_1\RQz\FockG_1\RQz\mathbb{A}\RQG\FockR_0\Chiz} \right| \label{eq:RQj} \\
&\hspace{3cm} + \left| \lr{\Chiz,\FockG_1\RQz\mathbb{A}\RQG\FockR_1\Chiz} \right| \label{eq:RQl}\\
&\hspace{3cm} + \left| \lr{\Chiz,\mathbb{A}\RQG\FockR_1\RQz\FockG_1\Chiz} \right| \label{eq:RQm}\\ 
&\hspace{3cm} + \left|\lr{\Chiz,\mathbb{A}\RQG\FockR_2\Chiz}\right| \Bigg). \label{eq:RQn}
\end{align}
\end{subequations}
We estimate the terms separately: 

\paragraph{Term \eqref{eq:RQa}.} For the first line of the r.h.s. of \eqref{RQA} we find with \eqref{eq:R0-estimate} and Lemma \ref{lem:G2}
\begin{align}
\vert \eqref{eq:RQa} \vert \ls& \sup_{z \in \gamma} {N^{(\beta-1)/2}}\Big\| ( \boldKz + 1)^{1/2} ( \mathcal{N}_\perp + 1)^{-5/2} \frac{\mathbb{Q}}{z-\mathbb{G}} \mathbb{A} \Chiz \Big\|  \; \Big\| ( \boldKz + 1)^{1/2} ( \mathcal{N}_\perp + 1)^{4} \Chiz \| \notag \\
&\hspace{2cm}\times N^{\beta -1}\Big\| ( \boldKz + 1)^{1/2} ( \mathcal{N}_\perp + 1)^{1/2} \Chiz \|^2 \; . 
\end{align}
With Lemma \ref{lem:resolventG} and using that $|z|\ls1$ for $z\in\gamma$, we get for the first term of the r.h.s. 
\begin{align}
\vert \eqref{eq:RQa} \vert \ls&   N^{3(\beta -1)/2} \; \| \mathbb{A} \Chiz \|   \| ( \boldKz + 1)^{1/2} ( \mathcal{N}_\perp + 1)^{4} \Chiz \| \; \| ( \boldKz + 1)^{1/2} ( \mathcal{N}_\perp + 1)^{1/2} \Chiz \| \,,
\end{align}
which we can finally bound since 
\begin{align}
\label{eq:bound-A}
\| \mathbb{A} \psi \| \ls \| \mathcal{N}_\perp \psi \|
\end{align}
 for any $\psi \in \mathcal{F} $.
With Lemma \ref{lem:BT:K0}, recalling that $\Chiz = \mathbb{U}_\tau \Omega$, we find that 
\begin{align}
\vert \eqref{eq:RQa} \vert \ls&  N^{3(\beta -1)/2}  \; .  
\end{align}

\paragraph{Term \eqref{eq:RQb}.} We proceed similarly and find with \eqref{eq:R0-estimate} 
\begin{align}
\vert \eqref{eq:RQb} \vert \ls & \sup_{z \in \gamma}N^{(\beta -1)/2} \Big\| ( \boldKz + 1)^{1/2} ( \mathcal{N}_\perp + 1)^{-5/2} \frac{\mathbb{Q}}{z-\mathbb{G}} \mathbb{A} \Chiz \Big\| \notag \\
& \hspace{2cm} \times \; \Big\| ( \boldKz + 1)^{1/2} ( \mathcal{N}_\perp + 1)^{4} \frac{\mathbb{Q}_0}{z-\mathbb{G}_0} \mathbb{G}_2 \Chiz \|   \; . 
\end{align}
For the first term of the r.h.s. we use Lemma \ref{lem:resolventG} and \eqref{eq:bound-A} and for the second term Lemma~\ref{lem:aux}. This yields
\begin{align}
\vert \eqref{eq:RQb} \vert  \ls N^{(\beta -1)/2}  \; \| \mathcal{N}_\perp \Chiz \| \; \Big\|  ( \mathcal{N}_\perp + 1)^{4}( \mathbb{K}_0 + 1)^{\frac{1}{2}} \RQz \mathbb{G}_2 \Chiz \| \; . \label{eq:RQb-1}
\end{align}
With Lemmas \ref{lem:aux:2} and \ref{lem:BT:K0}, we find that
\begin{align}
\vert \eqref{eq:RQb} \vert  \ls N^{3(\beta -1)/2} \; \| \mathcal{N}_\perp \Chiz \| \;  \| ( \mathcal{N}_\perp + 1)^{11/2}( \boldKz + 1)^{1/2} \Chiz \| 
 \ls N^{3( \beta-1)/2} \; . 
\end{align}

\paragraph{Term \eqref{eq:RQc}} From \eqref{eq:R0-estimate} we find 
\begin{align}
\vert \eqref{eq:RQc} \vert \ls& \sup_{z \in \gamma} N^{(\beta-1)/2} \Big\| ( \boldKz+1)^{1/2} ( \mathcal{N}_\perp +1)^{-5/2} \frac{\mathbb{Q}}{z-\mathbb{G}} \mathbb{A}\frac{\mathbb{Q}_0}{z-\mathbb{G}_0} \mathbb{G}_2 \Chiz \Big\| \notag \\
&\hspace{2cm} \times \; \Big\| ( \boldKz+1)^{1/2} ( \mathcal{N}_\perp +1)^{4}  \Chiz \|  \; . 
\end{align}
We use Lemmas \ref{lem:resolventG}, \ref{lem:aux} and \eqref{eq:bound-A} for the first and Lemma \ref{lem:BT:K0} for the second term and find (similarly as before) 
\begin{align}
\vert \eqref{eq:RQc} \vert \ls& \sup_{z \in \gamma} N^{(\beta-1)/2}\Big\|\mathcal{N}_\perp \frac{\mathbb{Q}_0}{z-\mathbb{G}_0} \mathbb{G}_2 \Chiz \Big\|  \; .
\end{align}
The remaining term can be estimated with Lemma \ref{lem:aux:2} and thus we finally get 
\begin{align}
\vert \eqref{eq:RQc} \vert \ls  N^{3(\beta-1)/2} \Big\| (\mathcal{N}_\perp+1)^2 ( \boldKz + 1)^{1/2} \Chiz \Big\| \ls N^{3(\beta-1)/2} \; . 
\end{align}

\paragraph{Term \eqref{eq:RQf}} It follows immediately from the estimate in \eqref{eq:RQa} and Lemma \ref{lem:aux:2} that
\begin{equation}
|\eqref{eq:RQf}|\ls N^{\frac32(\beta-1)}\,.
\end{equation}

\paragraph{Term \eqref{eq:RQh}. } We use \eqref{eq:R0-estimate} first and afterwards as above Lemmas \ref{lem:resolventG}, \ref{lem:aux}, \ref{lem:aux:2}, \ref{lem:BT:K0} and \eqref{eq:bound-A} to obtain
\begin{align}
\vert \eqref{eq:RQh} \vert \ls& \sup_{z \in \gamma} N^{(\beta-1)/2} \Big\|  ( \boldKz + 1)^{1/2}   ( \mathcal{N}_\perp + 1)^{-5/2} \frac{\mathbb{Q}}{z-\mathbb{G} } \mathbb{A} \Chiz \Big\| \; \notag \\
& \hspace{3cm} \times \|  ( \boldKz + 1)^{1/2}   ( \mathcal{N}_\perp + 1)^{4}  \frac{\mathbb{Q}_0}{z-\mathbb{G}_0} \mathbb{G}_1 \frac{\mathbb{Q}_0}{z-\mathbb{G}_0} \mathbb{G}_1  \Chiz \| \notag \\
\ls& \sup_{z \in \gamma} N^{\beta-1} \|  ( \boldKz + 1)^{1/2}   ( \mathcal{N}_\perp + 1)^{5}   \frac{\mathbb{Q}_0}{z-\mathbb{G}_0} \mathbb{G}_1  \Chiz \| \; \| \mathcal{N}_\perp \Chiz \| \ls N^{3(\beta-1)/2} \; .
\end{align}

\paragraph{Term \eqref{eq:RQi}.} Formula \eqref{eq:R0-estimate} with Lemmas \ref{lem:resolventG}, \ref{lem:aux}, \ref{lem:aux:2},  \ref{lem:BT:K0} and well as \eqref{eq:bound-A} lead to 
\begin{align}
\vert \eqref{eq:RQi} \vert \ls&  \sup_{z \in \gamma} N^{(\beta-1)/2} \Big\|  ( \boldKz + 1)^{1/2}   ( \mathcal{N}_\perp + 1)^{-5/2} \frac{\mathbb{Q}}{z-\mathbb{G}} \mathbb{A}\frac{\mathbb{Q}_0}{z-\mathbb{G}_0} \mathbb{G}_1 \Chiz \Big\| \; \notag \\
& \hspace{3cm} \times \|  ( \boldKz + 1)^{1/2}   ( \mathcal{N}_\perp + 1)^{4}  \frac{\mathbb{Q}_0}{z-\mathbb{G}_0} \mathbb{G}_1 \ \Chiz \| \notag \\
\ls& N^{(\beta-1)/2} \|  \mathcal{N}_\perp \RQz  \mathbb{G}_1 \Chiz \|  \; \|  ( \boldKz + 1)^{1/2}   ( \mathcal{N}_\perp + 1)^{4}\RQz \mathbb{G}_1 \ \Chiz \| \notag\\
\ls& N^{3(\beta-1)/2}\,.
\end{align}

\paragraph{Term \eqref{eq:RQj}.} Similarly as before we use \eqref{eq:R0-estimate} first, and then Lemmas \ref{lem:resolventG}, \ref{lem:aux}, \ref{lem:aux:2} and \ref{lem:BT:K0}. We obtain 
\begin{align}
\vert \eqref{eq:RQj} \vert \ls&  \sup_{z \in \gamma} N^{(\beta-1)/2} \Big\|  ( \boldKz + 1)^{1/2}   ( \mathcal{N}_\perp + 1)^{-5/2} \frac{\mathbb{Q}}{z-\mathbb{G}} \mathbb{A} \frac{\mathbb{Q}_0}{z-\mathbb{G}_0} \mathbb{G}_1 \frac{\mathbb{Q}_0}{z-\mathbb{G}_0} \mathbb{G}_1 \Chiz \Big\| \; \notag \\
\ls& N^{3( \beta-1)/2} \; . 
\end{align}

\paragraph{Term \eqref{eq:RQl}.} As before we use \eqref{eq:R1-estimate}, Lemmas \ref{lem:resolventG}, \ref{lem:aux}, \ref{lem:aux:2}  and \ref{lem:BT:K0} as well as the estimate \eqref{eq:bound-A} to compute
\begin{align}
\vert \eqref{eq:RQl} \vert \ls& \sup_{z \in \gamma} N^{\beta-1} \Big\|  ( \boldKz + 1)^{1/2}   ( \mathcal{N}_\perp + 1)^{-5/2} \frac{\mathbb{Q}}{z-\mathbb{G}} \mathbb{A} \frac{\mathbb{Q}_0}{z-\mathbb{G}_0}\mathbb{G}_1 \Chiz \Big\| \;\notag \\
& \hspace{2cm} \times \|  ( \boldKz + 1)^{1/2}   ( \mathcal{N}_\perp + 1)^{4} \Chiz \| \notag \\
\ls&  N^{3(\beta-1)/2} \; . 
\end{align}

\paragraph{Term \eqref{eq:RQm}.} We combine \eqref{eq:R1-estimate}, \ref{lem:resolventG}, \ref{lem:aux}, \ref{lem:aux:2} and \eqref{eq:bound-A}, which yields 
\begin{align}
\vert \eqref{eq:RQm} \vert \ls& \sup_{z \in \gamma} N^{\beta-1}\Big\|  ( \boldKz + 1)^{1/2}   ( \mathcal{N}_\perp + 1)^{-5/2} \frac{\mathbb{Q}}{z-\mathbb{G}} \mathbb{A} \Chiz \Big\| \; \notag \\
&\hspace{3cm} \times \|  ( \boldKz + 1)^{1/2}   ( \mathcal{N}_\perp + 1)^{4} \frac{\mathbb{Q}_0}{z-\mathbb{G}_0}\mathbb{G}_1 \Chiz \| \notag \\
\ls&  N^{3(\beta-1)/2} \; . 
\end{align}

\paragraph{Term \eqref{eq:RQn}.} For the last term we proceed similarly as before. With \eqref{eq:R2-estimate}, Lemma \ref{lem:resolventG} and \ref{lem:BT:K0}, we conclude that
\begin{align}
\vert \eqref{eq:RQn} \vert \ls& \sup_{z \in \gamma} N^{3(\beta-1)/2} \Big\|  ( \boldKz + 1)^{1/2}   ( \mathcal{N}_\perp + 1)^{-5/2} \frac{\mathbb{Q}}{z-\mathbb{G}}  \mathbb{A} \Chiz \Big\| \; \notag \\
&\hspace{3cm} \times \|  ( \boldKz + 1)^{1/2}   ( \mathcal{N}_\perp + 1)^{4} \Chiz \| \notag \\
\ls&  N^{\frac32(\beta-1)} \; . 
\end{align}

\paragraph{Summary.} Thus, in summary, we have proven that 
\begin{align}
\vert  \Tr\FockA\FockB_Q\vert \ls N^{3(\beta-1)/2} \; .  \label{eq:RQ-final}
\end{align}

\section{Explicit calculation of energy correction}
\label{section:explicit calculation of E-pert}

\subsection{Calculation for $E_{0,1}$}
Recall that $E_{0,1}$ is defined in \eqref{E01}.
Here, we prove that for $\beta \in \big(\frac{1}{2},1\big)$ we have
\begin{align}
E_{0,1} = C_{N,\beta}^{(1)} \Bigg( -\frac{1}{2N} \sum_{q \in \Lsp} \frac{\widehat{v}_N^\beta (q)^2}{2q^2} \Bigg) + O(N^{2(\beta-1)}),
\end{align}
with 
\begin{align}
C_{N,\beta}^{(1)} = \frac{1}{2} \sum_{p \in \Lsp} (s_pc_p-\eta_p) + \widehat{v}(0)^2 \sum_{p \in \Lsp} \frac{1}{\sqrt{\vert p \vert^4 + 2 p^2 \widehat{v}(0)} \left( p^2 + \sqrt{\vert p \vert^4 + 2 p^2 \widehat{v}(0)} \right)}.
\end{align}
Note that $|C_{N,\beta}^{(1)}| \ls 1$. We define 
\begin{align}
E^{(1)}_{0,1}& := -\frac{1}{2N} \sum_{\substack{p,q \in \Lsp \\ p\not= q}} \widehat{v}_N^\beta (p-q) (s_pc_p-\eta_p) \bigg[ s_qc_q + \frac{\widehat{v}_N^\beta (q)}{q^2}\bigg],
\end{align}
and
\begin{align}
E^{(2)}_{0,1}& := \frac{1}{N} \sum_{\substack{p,q \in \Lsp \\ p\not= q}} \frac{ \widehat{v}_N^\beta (p)^2 \; \; \widehat{v}_N^\beta (p-q) s_qc_q }{\sqrt{\vert p \vert^4 + 2 p^2 \widehat{v}_N^\beta(p) }  \left(  p^2 + \sqrt{\vert  p \vert^4 + 2 p^2 \widehat{v}_N^\beta(p) } \right) }\,,
\end{align}
then it follows that $E_{0,1} = E^{(1)}_{0,1} + E^{(2)}_{0,1}$. Note that the scattering equation \eqref{def:eta} implies that
\begin{align}\label{another_eta_bound}
\bigg| \eta_p + \frac{\widehat{v}_N^\beta(p)}{2p^2} \bigg| = \frac{1}{2N} |p|^{-2} \bigg| \sum_{q \in \Lsp} \widehat{v}_N^\beta(p-q) \eta_q \bigg| \ls N^{\beta-1} |p|^{-2}
\end{align}
by Lemmas~\ref{lem:aux:potential} and \ref{lem:HST}. Also, since $v$ has compact support, we have the estimate
\begin{align}\label{Taylor_of_v_hat}
\Big| \widehat{v}_N^\beta(p-q) - \widehat{v}_N^\beta(q) \Big| \leq N^{-\beta} |p| \int_0^1 |\nabla\widehat{v}_N^{\beta}(tp-q)| \dt \ls N^{-\beta} |p|.
\end{align}
By an expansion of the square root as in \cite[Sec.~5]{boccato2017_2}, this implies the pointwise estimate
\begin{equation}\label{Taylor_of_v_hat_in_square_root}
\left| \sqrt{\vert p \vert^4 + 2 p^2 \widehat{v}_N^\beta(p)} - \sqrt{\vert p \vert^4 + 2 p^2 \widehat{v}(0)} \right| \ls N^{-\beta} |p|.
\end{equation}

\noindent \textbf{Simplification of $E^{(1)}_{0,1}$.} We write
\begin{subequations}
\begin{align}
E^{(1)}_{0,1} &= \Bigg(\frac{1}{2} \sum_{p \in \Lsp} (s_pc_p-\eta_p)\Bigg) \Bigg( -\frac{1}{2N} \sum_{q \in \Lsp} \frac{\widehat{v}_N^\beta (q)^2}{q^2} \Bigg) \\
&\quad - \frac{1}{2N} \sum_{\substack{p,q \in \Lsp \\ p\not= q}} \widehat{v}_N^\beta (p-q) (s_pc_p-\eta_p) \bigg[ s_qc_q -\eta_q + \eta_q + \frac{\widehat{v}_N^\beta (q)}{2q^2}\bigg] \label{error_r_1} \\
&\quad + \frac{1}{2N} \sum_{p \in \Lsp} (s_pc_p-\eta_p) \frac{\widehat{v}_N^\beta(p)^2}{2p^2} \label{error_r_2} \\
&\quad - \frac{1}{2N} \sum_{\substack{p,q \in \Lsp \\ p\not= q}} \Big( \widehat{v}_N^\beta(p-q) - \widehat{v}_N^\beta(q) \Big) (s_pc_p-\eta_p) \frac{\widehat{v}_N^\beta (q)}{2q^2}. \label{error_r_3}
\end{align}
\end{subequations}
Then, using Lemma~\ref{lem:aux:potential}, Lemma~\ref{lem:cp:sp}, and \eqref{another_eta_bound}, we find that
\begin{equation}
|\eqref{error_r_1}| \ls N^{-1} \sum_{p,q \in \Lsp, p\not= q} |\widehat{v}_N^\beta(p-q)| \, |p|^{-6} \Big( |q|^{-6} + N^{\beta-1}|q|^{-2} \Big) \ls N^{-1} + N^{2(\beta-1)},
\end{equation}
and
\begin{equation}
|\eqref{error_r_2}| \ls N^{-1} \sum_{p \in \Lsp} \frac{\widehat{v}_N^\beta(p)^2}{2p^2} |p|^{-6} \ls N^{-1}.
\end{equation}
Furthermore, by Lemma~\ref{lem:aux:potential}, Lemma~\ref{lem:cp:sp}, and \eqref{Taylor_of_v_hat} we find
\begin{equation}
|\eqref{error_r_3}| \ls N^{-1} \sum_{\substack{p,q \in \Lsp \\ p\not= q}} N^{-\beta} |p| |p|^{-6} \frac{|\widehat{v}_N^\beta(q)|}{q^2} \ls N^{-1}.
\end{equation}
In conclusion, for $\beta \in \big(\frac{1}{2},1\big)$ we have
\begin{align}
E^{(1)}_{0,1} &= \Bigg(\frac{1}{2} \sum_{p \in \Lsp} (s_pc_p-\eta_p)\Bigg) \Bigg( -\frac{1}{2N} \sum_{q \in \Lsp} \frac{\widehat{v}_N^\beta (q)^2}{2q^2} \Bigg) + \mathcal{O}(N^{2(\beta-1)}).
\end{align}

\noindent \textbf{Simplification of $E^{(2)}_{0,1}$.} We start by writing
\begin{subequations}
\begin{align}
E^{(2)}_{0,1}& = - \frac{1}{N} \sum_{p,q \in \Lsp} \frac{ \widehat{v}_N^\beta(0)^2 \; \; \widehat{v}_N^\beta (q)}{\sqrt{\vert p \vert^4 + 2 p^2 \widehat{v}_N^\beta(p) }  \left(  p^2 + \sqrt{\vert  p \vert^4 + 2 p^2 \widehat{v}_N^\beta(p) }\right)} \left( \frac{\widehat{v}_N^\beta(q)}{2q^2} \right) \\
&\quad + \frac{1}{N} \sum_{p \in \Lsp} \frac{ \widehat{v}_N^\beta(0)^2 \; \; \widehat{v}_N^\beta(p)}{\sqrt{\vert p \vert^4 + 2 p^2 \widehat{v}_N^\beta(p) }  \left(  p^2 + \sqrt{\vert  p \vert^4 + 2 p^2 \widehat{v}_N^\beta(p) }\right)} \left( \frac{\widehat{v}_N^\beta(p)}{2p^2} \right) \label{error_r_4} \\
&\quad - \frac{1}{N} \sum_{\substack{p,q \in \Lsp \\ p\not= q}} \frac{ \left( \widehat{v}_N^\beta(p)^2 \; \; \widehat{v}_N^\beta(p-q) - \widehat{v}_N^\beta(0)^2 \; \; \widehat{v}_N^\beta(q) \right)}{\sqrt{\vert p \vert^4 + 2 p^2 \widehat{v}_N^\beta(p) }  \left(  p^2 + \sqrt{\vert  p \vert^4 + 2 p^2 \widehat{v}_N^\beta(p) }\right)} \left( \frac{\widehat{v}_N^\beta(q)}{2q^2} \right) \label{error_r_5} \\
&\quad + \frac{1}{N} \sum_{\substack{p,q \in \Lsp \\ p\not= q}} \frac{ \widehat{v}_N^\beta (p)^2 \; \; \widehat{v}_N^\beta (p-q)}{\sqrt{\vert p \vert^4 + 2 p^2 \widehat{v}_N^\beta(p) }  \left(  p^2 + \sqrt{\vert  p \vert^4 + 2 p^2 \widehat{v}_N^\beta(p) }\right)} \left( s_qc_q - \eta_q + \eta_q + \frac{\widehat{v}_N^\beta(q)}{2q^2} \right).\label{error_r_6}
\end{align}
\end{subequations}
We directly find $|\eqref{error_r_4}| \ls N^{-1}$. Next, using \eqref{Taylor_of_v_hat} for momenta $|p| \leq N^{\beta}$ and $\sup_{p \in \Lsp} \big|\widehat{v}_N^\beta(p)\big| \ls 1$ for momenta $|p| \geq N^{\beta}$, together with Lemma~\ref{lem:aux:potential}, we find
\begin{align}\label{r_5_estimate}
|\eqref{error_r_5}| &\ls \frac{1}{N} \sum_{\substack{p,q \in \Lsp, p\not= q \\ |p| \leq N^{\beta}}} \frac{ N^{-\beta} |p|}{\sqrt{\vert p \vert^4 + 2 p^2 \widehat{v}_N^\beta(p) }  \left(  p^2 + \sqrt{\vert  p \vert^4 + 2 p^2 \widehat{v}_N^\beta(p) }\right)} \left( \frac{\widehat{v}_N^\beta (q)^2}{2q^2} \right) \nonumber\\
&\quad + \frac{1}{N} \sum_{\substack{p,q \in \Lsp, p\not= q \\ |p| \geq N^{\beta}}} \frac{1}{\sqrt{\vert p \vert^4 + 2 p^2 \widehat{v}_N^\beta(p) }  \left(  p^2 + \sqrt{\vert  p \vert^4 + 2 p^2 \widehat{v}_N^\beta(p) }\right)} \left( \frac{\big|\widehat{v}_N^\beta(q)\big|}{2q^2} \right) \nonumber\\
&\ls N^{-1} \ln(N).
\end{align}
Furthermore, by Lemma~\ref{lem:cp:sp} and \eqref{another_eta_bound},
\begin{equation}
|\eqref{error_r_6}| \ls N^{-1} + N^{2(\beta-1)}.
\end{equation}
Finally, we use \eqref{Taylor_of_v_hat_in_square_root} to deduce that
\begin{align}
&\sum_{p \in \Lsp} \left| \frac{1}{\sqrt{\vert p \vert^4 + 2 p^2 \widehat{v}_N^\beta(p) }  \left(p^2 + \sqrt{\vert  p \vert^4 + 2 p^2 \widehat{v}_N^\beta(p)}\right)} - \frac{1}{\sqrt{\vert p \vert^4 + 2 p^2 \widehat{v}(0) }  \left(p^2 + \sqrt{\vert  p \vert^4 + 2 p^2 \widehat{v}(0) }\right)} \right| \nonumber\\
&\quad = \sum_{p \in \Lsp} \left| \frac{p^2 \Big( \sqrt{\vert p \vert^4 + 2 p^2 \widehat{v}_N^\beta(p)} - \sqrt{\vert p \vert^4 + 2 p^2 \widehat{v}(0)} \Big) + 2p^2 \Big( \widehat{v}_N^\beta(p) - \widehat{v}(0) \Big)}{\sqrt{\vert p \vert^4 + 2 p^2 \widehat{v}_N^\beta(p) }  \left(p^2 + \sqrt{\vert  p \vert^4 + 2 p^2 \widehat{v}_N^\beta(p) }\right)\sqrt{\vert p \vert^4 + 2 p^2 \widehat{v}(0) }  \left(p^2 + \sqrt{\vert  p \vert^4 + 2 p^2 \widehat{v}(0) }\right)} \right| \nonumber\\
&\quad \ls \sum_{p \in \Lsp} \frac{|p|^3 N^{-\beta}}{p^8} \ls N^{-\beta},
\end{align}
which proves, for $\beta \in \big(\frac{1}{2},1\big)$, that
\begin{align}
E^{(2)}_{0,1} &= \left(\widehat{v}(0)^2 \sum_{p \in \Lsp} \frac{1}{\sqrt{\vert p \vert^4 + 2 p^2 \widehat{v}(0)} \left( p^2 + \sqrt{\vert p \vert^4 + 2 p^2 \widehat{v}(0)} \right)} \right) \Bigg( -\frac{1}{2N} \sum_{q \in \Lsp} \frac{\widehat{v}_N^\beta (q)^2}{2q^2} \Bigg) \nonumber\\ 
&\quad + \mathcal{O}(N^{2(\beta-1)}).
\end{align}

\subsection{Calculation for $E_\mathrm{pert}$}

The goal of this section is to estimate and explicitly calculate the leading order contribution of
\begin{equation}
E_\mathrm{pert}:= \lr{\Chiz,\FockG_2\Chiz} + \lr{\Chiz,\FockG_1\frac{\Q_0}{E_0-\FockGz}\FockG_1\Chiz}
\end{equation}
with $\Chiz=\BogUz^*\vac$. To this end, we introduce the shorthand notation 
\begin{align}
\widetilde{c}_p &= \cosh (\tau_p)
\quad \text{and} \quad
\widetilde{s}_p = \sinh (\tau_p) .
\end{align}
By the definition of  $\FockG_2$ we get
\begin{align}
\scp{\Chiz}{\FockG_2 \Chiz}
&= \frac{1}{2N} \sum_{\substack{p,q,r\in\Lsp\\ p+r\neq0,q+r\neq0}}\hvNb(r) c_{p+r} c_q c_p  c_{q+r} 
\scp{\Omega}{\BogUz \ad_{p+r} \ad_q   a_p a_{q+r} \BogUz^* \Omega} .
\end{align}
Note that 
\begin{align}
&\scp{\Omega}{\BogUz \ad_{p+r} \ad_q   a_p a_{q+r} \BogUz^* \Omega}
\nonumber \\
&\quad = \widetilde{s}_{p+r} \widetilde{s}_{q+r}
\scp{\Omega}{ a_{-(p+r)} \BogUz \ad_q   a_p \BogUz^* \ad_{-(q+r)}  \Omega}
\nonumber \\
&\quad = \widetilde{s}_{p+r} \widetilde{s}_{q+r}
\scp{\Omega}{ a_{-(p+r)} \left( \widetilde{c}_q \ad_q + \widetilde{s}_q a_{-q} \right) \left( \widetilde{c}_p a_p + \widetilde{s}_p \ad_{-p} \right) \ad_{-(q+r)}  \Omega}
\nonumber \\
&\quad = \widetilde{s}_{p+r} \widetilde{s}_{q+r}\widetilde{c}_q \widetilde{c}_p
\scp{\Omega}{ a_{-(p+r)} \ad_q  a_p  \ad_{-(q+r)}  \Omega}
 + \widetilde{s}_{p+r} \widetilde{s}_{q+r}\widetilde{s}_q \widetilde{s}_p
\scp{\Omega}{ a_{-(p+r)} a_{-q} \ad_{-p}  \ad_{-(q+r)}  \Omega}
\nonumber \\
&\quad = \widetilde{s}_{p+r} \widetilde{s}_{q+r}\widetilde{c}_q \widetilde{c}_p
\delta_{-(p+r),q} \delta_{-(q+r),p}
 + \widetilde{s}_{p+r} \widetilde{s}_{q+r}\widetilde{s}_q \widetilde{s}_p
\left( \delta_{p,q} \delta_{r,0} + \delta_{r,0} \right) .
\end{align}
Since $\delta_{-(q+r),p} = \delta_{-(p+r),q}$,$\delta_{r,0} = 0$ for all $r \in \Lsp$, $\widetilde{s}_{-p} = \widetilde{s}_p$ and $\widetilde{c}_{-p} = \widetilde{c}_p$ we obtain
\begin{align}
\scp{\Chiz}{\FockG_2 \Chiz}
&= \frac{1}{2N} \sum_{\substack{p,q,r\in\Lsp\\ p+r\neq0,q+r\neq0}}\hvNb(r) c_{p+r} c_q c_p  c_{q+r} 
\widetilde{s}_{p+r} \widetilde{s}_{q+r}\widetilde{c}_q \widetilde{c}_p \delta_{-(p+r),q} 
\nonumber \\
&= \frac{1}{2N} \sum_{\substack{p,r\in\Lsp\\ p+r\neq0}}\hvNb(r) c_{p+r}^2 c_p^2   
\widetilde{s}_{p+r} \widetilde{s}_{p}\widetilde{c}_{p+r} \widetilde{c}_p  .
\end{align}
Using \eqref{eq:bound-tau} we estimate
\begin{align}
\abs{\scp{\Chiz}{\FockG_2 \Chiz} }
&\leq C N^{-1} \norm{\widehat{v}}_{\infty}  \sum_{\substack{p,r\in\Lsp\\ p+r\neq0}} \abs{\tau_{p+r}} \abs{\tau_p} 
\leq C N^{-1} \norm{\widehat{v}}_{\infty}  \sum_{\substack{p,r\in\Lsp\\ p+r\neq0}} p^{-4} (p+r)^{-4}
\leq C N^{-1} ,
\end{align}
showing that the contribution of $\scp{\Chiz}{\FockG_2 \Chiz}$ is subleading.
Next, we use the splitting \eqref{eq:tildeG1} and consider
\begin{align}
\lr{\Chiz,\FockG_1\frac{\Q_0}{E_0-\FockGz}\FockG_1\Chiz}
&= \lr{\Chiz, \widetilde{\mathbb{G}}_1 \frac{\Q_0}{E_0-\FockGz} \widetilde{\mathbb{G}}_1 \Chiz}
\\
&\quad +  \lr{\Chiz, \mathbb{R}_d  \frac{\Q_0}{E_0-\FockGz} \mathbb{R}_d  \Chiz}
\\
&\quad +2 \Re  \lr{\Chiz, \widetilde{\mathbb{G}}_1 \frac{\Q_0}{E_0-\FockGz} \mathbb{R}_d  \Chiz} .
\end{align}

\paragraph{The first term.}

A direct computation leads to
\begin{align}
\BogUz \widetilde{\mathbb{G}}_1 \BogUz^* \vac &= \frac{1}{\sqrt{N}}\sum_{\substack{p,q\in\Lsp\\p+q\neq 0}} f(p,q) \ad_{q} \ad_{p} \ad_{-p-q} \vac,
\end{align}
with
\begin{align}\label{f_def}
f(p,q) &= \frac{1}{6} c_{p+q} c_p c_q \bigg( 
\widehat{v}_N^\beta(p) (\widetilde{c}_p + \widetilde{s}_p) (\widetilde{c}_{p+q}\widetilde{s}_q + \widetilde{c}_{q}\widetilde{s}_{p+q})
+ \widehat{v}_N^\beta(q) (\widetilde{c}_q + \widetilde{s}_q) (\widetilde{c}_{p+q}\widetilde{s}_p + \widetilde{c}_{p}\widetilde{s}_{p+q}) \nonumber\\
&\qquad + \widehat{v}_N^\beta(p+q) (\widetilde{c}_{p+q} + \widetilde{s}_{p+q}) (\widetilde{c}_p\widetilde{s}_q + \widetilde{c}_{q}\widetilde{s}_p) \bigg) \nonumber\\
&\quad + \frac{1}{3} \Big( \widetilde{c}_p\widetilde{c}_q\widetilde{c}_{p+q} + \widetilde{s}_p\widetilde{s}_q\widetilde{s}_{p+q} \Big) \bigg( \widehat{v}_N^\beta(p) c_p (c_{p+q}s_q + c_qs_{p+q}) + \widehat{v}_N^\beta(q) c_q (c_{p+q}s_p + c_ps_{p+q}) \nonumber\\
&\qquad + \widehat{v}_N^\beta(p+q) c_{p+q} (c_ps_q + c_qs_p) \bigg).
\end{align}
Note that this function is written in a symmetric way, i.e., $f(p,q) = f(q,p) = f(-p-q,q) = f(p,-p-q)$.
Together with
\begin{align}
\BogUz \frac{\Q_0}{E_0-\FockGz}  \BogUz^* \ad_{p_1} \ldots \ad_{p_n} \vac
=  - \frac{1}{ \mathfrak{e}(p_1) + \ldots + \mathfrak{e}(p_n)} \ad_{p_1} \ldots \ad_{p_n} \vac,
\end{align}
where $\mathfrak{e}(p) = \sqrt{F_p^2 - G_p^2}$, this leads to
\begin{align}
&\lr{\Chiz, \widetilde{\mathbb{G}}_1 \frac{\Q_0}{E_0-\FockGz} \widetilde{\mathbb{G}}_1 \Chiz}
\nonumber \\
&\quad = \scp{ \BogUz \widetilde{\mathbb{G}}_1 \BogUz^* \Omega}{\BogUz \frac{\Q_0}{E_0-\FockGz} \BogUz^* \BogUz \widetilde{\mathbb{G}}_1 \BogUz^* \Omega}
\nonumber \\
&\quad = - \frac{1}{N} \sum_{\substack{p,q\in\Lsp\\p+q\neq0}} \sum_{\substack{r,s\in\Lsp\\r+s\neq0}}   \frac{f(p,q) f(r,s)}{ \mathfrak{e}(p+q) + \mathfrak{e}(-p) + \mathfrak{e}(-q) }
\scp{\ad_{p+q} \ad_{-p}   \ad_{-q} \Omega}{\ad_{r+s} \ad_{-r}   \ad_{-s} \Omega} \nonumber\\
&\quad = - \frac{6}{N} \sum_{\substack{p,q\in\Lsp\\p+q\neq0}} \frac{f(p,q)^2}{ \mathfrak{e}(p+q) + \mathfrak{e}(-p) + \mathfrak{e}(-q) }.
\end{align}
From Lemma~\ref{lemma:propF,G} we find that $\mathfrak{e}(p) \gs |p|^2$, and thus $(\mathfrak{e}(p+q) + \mathfrak{e}(-p) + \mathfrak{e}(-q))^{-1} \ls |p|^{-2}$. Using this, Cauchy--Schwarz, Lemma~\ref{lem:cp:sp}, $|\widetilde{c}_p| \ls 1$, $|\widetilde{s}_p| \ls |p|^{-4}$, and the symmetry of $f(p,q)$, we can estimate
\begin{align}
\bigg|\lr{\Chiz, \widetilde{\mathbb{G}}_1 \frac{\Q_0}{E_0-\FockGz} \widetilde{\mathbb{G}}_1 \Chiz}\bigg| \ls N^{-1} \sum_{\substack{p,q\in\Lsp\\p+q\neq 0}} |p|^{-2} \left(\widehat{v}_N^\beta(p)\right)^2 |q|^{-4} \ls N^{\beta-1},
\end{align}
where the last step follows from Lemma~\ref{lem:aux:potential}.

\paragraph{The second term.}

Note that Lemma~\ref{lem:G1}  implies
\begin{align}
\norm{\mathbb{R}_d  \Chiz}^2 
&= \scp{\mathbb{R}_d  \Chiz}{\mathbb{R}_d  \Chiz} 
\ls N^{- \frac{1}{2} } \norm{\mathbb{R}_d  \Chiz} \norm{\left( \Np +1 \right)^{\frac{3}{2}} \Chiz} 
\end{align}
and therefore
\begin{align}
\norm{\mathbb{R}_d  \Chiz} 
&\ls N^{- \frac{1}{2} }  \norm{\left( \Np +1 \right)^{\frac{3}{2}} \Chiz} .
\end{align}
Together with Lemma~\ref{lem:BT:K0} we get
\begin{align}
\abs{\lr{\Chiz, \mathbb{R}_d  \frac{\Q_0}{E_0-\FockGz} \mathbb{R}_d  \Chiz}}
&\ls N^{-1}  \norm{\left( \Np +1 \right)^{\frac{3}{2}} \Chiz}^2 
\ls N^{-1}  ,
\end{align}
showing that the left hand side only has a subleading contribution which can be absorbed in the overall error in our main theorem.

\paragraph{The third term.}
Having estimated the first and second term, the third term can immediately be bound by Cauchy--Schwarz, i.e.,
\begin{align}
\bigg|\lr{\Chiz, \widetilde{\mathbb{G}}_1 \frac{\Q_0}{E_0-\FockGz} \mathbb{R}_d  \Chiz}\bigg| \ls \sqrt{N^{-1}} \sqrt{N^{\beta-1}} =  N^{\beta/2 - 1} \ls N^{\frac{3}{2}(\beta-1)}
\end{align}
for $\beta\in(\frac{1}{2},1)$.

\paragraph{Summary.} We have shown that for $\beta\in(\frac{1}{2},1)$,
\begin{align}\label{E_tilde_pert}
E_\mathrm{pert} = \widetilde{E}_\mathrm{pert} + \mathcal{O}\big(N^{\frac{3}{2}(\beta-1)}\big), ~~\text{with}´~ \widetilde{E}_\mathrm{pert} := - \frac{6}{N} \sum_{\substack{p,q\in\Lsp \\ p+q\neq0}} \frac{f(p,q)^2}{ \mathfrak{e}(p+q) + \mathfrak{e}(p) + \mathfrak{e}(q)}\,.
\end{align}

Next, we simplify $\widetilde{E}_\mathrm{pert}$. Recall that $f(p,q)$ is defined in \eqref{f_def}, and $\mathfrak{e}(p) = \sqrt{F_p^2-G_p^2}$, with $F_p$ and $G_p$ defined in \eqref{def:F,G}. Our goal is to prove that
\begin{align}
\widetilde{E}_\mathrm{pert} = C_{N,\beta}^{(2)} \left( - \frac{1}{2N} \sum_{p\in\Lsp} \frac{\widehat{v}_N^\beta(p)^2}{2p^2} \right) + O(N^{-1}(\ln N)^2),
\end{align}
with
\begin{align}
C_{N,\beta}^{(2)} =  \sum_{q \in \Lsp} 4 \Big( c_q\widetilde{s}_q + 2 \widetilde{c}_q s_q \Big)^2.
\end{align}
Note that $|C_{N,\beta}^{(2)}| \ls 1$. We first collect a few preparatory estimates. Lemma~\ref{lemma:propF,G} implies that
\begin{align}\label{e_p_estimate_p_squared}
p^2 \ls \mathfrak{e}(p) \ls p^2
\end{align}
for $p\in\Lsp$, and thus we have
\begin{align}\label{e_p_denominator_estimate}
\frac{1}{ \mathfrak{e}(p+q) + \mathfrak{e}(p) + \mathfrak{e}(q)} \ls \frac{1}{p^2}, \frac{1}{q^2}, \frac{1}{|p| \, |q|}.
\end{align}
Furthermore, \eqref{eq:fquare-gsquare} implies that
\begin{align}\label{e_p_p_squared_approximation}
\left| p^2 - \mathfrak{e}(p) \right| &= p^2 \left| 1 - \sqrt{1 + \frac{2\widehat{v}_N^\beta(p)}{p^2} + \frac{A_p}{p^4}} \right| \nonumber\\
&= p^2 \left| \frac{1}{2} \left(\frac{2\widehat{v}_N^\beta(p)}{p^2} + \frac{A_p}{p^4}\right) \int_0^1 \ds \left( 1 + s\left(\frac{\widehat{v}_N^\beta(p)}{p^2} + \frac{A_p}{p^4}\right) \right)^{-1/2} \right| \nonumber\\
&\ls 1\,,
\end{align}
for $A_p$ as in \eqref{def:Ap} with $|A_p|\ls N^{\beta-1}$ by \eqref{eq:boundAp}. We start to estimate \eqref{E_tilde_pert} by using Lemma~\ref{lem:cp:sp} and $|\widetilde{s}_p| \ls |p|^{-4}$, $|\widetilde{c}_p - 1| \ls |p|^{-8}$. Then many contributions of $f(p,q)$ can be summed up, and renaming summation indices if necessary leads to
\begin{align}\label{E_tilde_pert_simplified}
\widetilde{E}_\mathrm{pert} = - \frac{1}{4N} \sum_{\substack{p,q\in\Lsp \\ p+q\neq0}} \frac{\widehat{v}_N^\beta(p) \Big(\widehat{v}_N^\beta(p) + \widehat{v}_N^\beta(p+q)\Big) f(q)}{ \mathfrak{e}(p+q) + \mathfrak{e}(p) + \mathfrak{e}(q)} + \mathcal{O}((\ln N)^2 N^{-1}),
\end{align}
with
\begin{equation}
f(q) := 4\Big(c_q\tilde{s}_q+2\tilde{c}_qs_q\big)^2\,.
\end{equation}
Note that here for the terms which decay only like $|p|^{-3}|q|^{-3}$ we have used that for any $r\in\Lsp$,
\begin{align}\label{p_to_the_minus_three_summation}
\sum_{p \in \Lsp} \frac{|\widehat{v}_N^\beta(p+r)|}{|p|^3} &= \sum_{\substack{p \in \Lsp \\ |p| \leq N^{\beta}}} \frac{|\widehat{v}_N^\beta(p+r)|}{|p|^3} + \sum_{\substack{p \in \Lsp \\ |p| > N^{\beta}}} \frac{|\widehat{v}_N^\beta(p+r)|}{|p|^3} \nonumber\\
&\leq \left(\sup_{p\in \Lsp} |\widehat{v}_N^\beta(p)|\right) \sum_{\substack{p \in \Lsp \\ |p| \leq N^{\beta}}} \frac{1}{|p|^3} + \sqrt{\sum_{p \in \Lsp} |\widehat{v}_N^\beta(p)|^2} \sqrt{\sum_{\substack{p \in \Lsp \\ |p| > N^{\beta}}} \frac{1}{|p|^6}} \nonumber\\
&\ls \ln(N) + \sqrt{N^{3\beta}} \sqrt{N^{-3\beta}}.
\end{align}
Note that $|f(q)| \ls |q|^{-4}$. We then split the remaining double sum in \eqref{E_tilde_pert_simplified} as
\begin{subequations}
\begin{align}
&- \frac{1}{4N} \sum_{\substack{p,q\in\Lsp \\ p+q\neq0}} \frac{\widehat{v}_N^\beta(p) \Big(\widehat{v}_N^\beta(p) + \widehat{v}_N^\beta(p+q)\Big) f(q)}{ \mathfrak{e}(p+q) + \mathfrak{e}(p) + \mathfrak{e}(q)} \notag\\
&\qquad =- \frac{1}{4N} \sum_{\substack{p,q\in\Lsp\\p+q\neq0}} \frac{\widehat{v}_N^\beta(p) \Big(\widehat{v}_N^\beta(p+q) - \widehat{v}_N^\beta(p)\Big) f(q)}{ \mathfrak{e}(p+q) + \mathfrak{e}(p) + \mathfrak{e}(q)} \label{error_r_8} \\
&\quad\qquad - \frac{1}{2N} \sum_{\substack{p,q\in\Lsp\\p+q\neq0}} \widehat{v}_N^\beta(p)^2 f(q) \left( \frac{1}{ \mathfrak{e}(p+q) + \mathfrak{e}(p) + \mathfrak{e}(q)} - \frac{1}{2p^2} \right) \label{error_r_9} \\
&\quad\qquad - \frac{1}{2N} \sum_{p\in\Lsp} \frac{\widehat{v}_N^\beta(p)^2}{2p^2}  \sum_{q\in\Lsp} f(q)\label{main:term}\\
&\quad\qquad+\frac{1}{2N}\sum_{p\in\Lsp}\frac{\hvNb(p)^2}{2p^2}f(-p)\label{error_r_7}\,.
\end{align}
\end{subequations}
We will now prove that \eqref{main:term} is the leading contribution.
Clearly $|\eqref{error_r_7}| \ls N^{-1}$, and similarly to \eqref{r_5_estimate} we estimate
\begin{equation}
|\eqref{error_r_8}| \ls N^{-1} \sum_{\substack{p,q\in\Lsp \\ |q| \leq N^{\beta}}} \frac{|\widehat{v}_N^\beta(p)| \, |q| N^{-\beta} |q|^{-4}}{p^2} + N^{-1} \sum_{\substack{p,q\in\Lsp \\ |q| > N^{\beta}}} \frac{|\widehat{v}_N^\beta(p)| |q|^{-4}}{p^2} \ls N^{-1} \ln(N),
\end{equation}
where we used \eqref{Taylor_of_v_hat} for momenta $|q| \leq N^{\beta}$ and $\norm{\hat{v}}_{\l^\infty} \ls 1$ for momenta $|q| \geq N^{\beta}$, together with Lemma~\ref{lem:aux:potential}. Finally, in order to estimate \eqref{error_r_9}, we first split into momenta $|q| \leq N^{\beta}$ and $|q| > N^\beta$, and then use \eqref{e_p_estimate_p_squared}, \eqref{e_p_denominator_estimate}, and \eqref{e_p_p_squared_approximation}. We find
\begin{align}
&|\eqref{error_r_9}|\notag\\
&= \Bigg| \frac{1}{2N} \sum_{\substack{p,q\in\Lsp \\p+q\neq0\\ |q| \leq N^{\beta}}} \widehat{v}_N^\beta(p)^2 f(q) \left( \frac{1}{ \mathfrak{e}(p+q) + \mathfrak{e}(p) + \mathfrak{e}(q)} - \frac{1}{2p^2} \right) \nonumber\\ 
&\quad + \frac{1}{2N} \sum_{\substack{p,q\in\Lsp\\p+q\neq0 \\ |q| > N^{\beta}}} \widehat{v}_N^\beta(p)^2 f(q) \left( \frac{1}{ \mathfrak{e}(p+q) + \mathfrak{e}(p) + \mathfrak{e}(q)} - \frac{1}{2p^2} \right) \Bigg| \nonumber\\
&\quad \ls N^{-1} \sum_{\substack{p,q\in\Lsp \\p+q\neq0\\ |q| \leq N^{\beta}}} \widehat{v}_N^\beta(p)^2 |f(q)| \left( \frac{|p^2 - \mathfrak{e}(p)| + |\mathfrak{e}(q)| + |p^2 - (p+q)^2| + |(p+q)^2 - \mathfrak{e}(p+q)|}{ p^2 \big( \mathfrak{e}(p+q) + \mathfrak{e}(p) + \mathfrak{e}(q)\big)}\right) \nonumber\\
&\quad\quad + \mathcal{O}(N^{-1}) \nonumber\\
&\quad \ls N^{-1} \sum_{\substack{p,q\in\Lsp \\ |q| \leq N^{\beta}}} \widehat{v}_N^\beta(p)^2 |q|^{-4} \left( \frac{1}{|p|^4} + \frac{|q|^2}{p^2|p||q|} + \frac{|p||q|}{p^2 p^2} + \frac{|q|^2}{p^2 |p||q|} \right) + \mathcal{O}(N^{-1}) \nonumber\\
&\quad \ls N^{-1} (\ln N)^2.
\end{align}
Note that in the last step we have additionally used  \eqref{p_to_the_minus_three_summation}.

\section{Proof of Theorem \ref{thm:state}}\label{sec:proof:thm:state}

Let $p_0 = \ket{\varphi_0} \bra{\varphi_0}$, $q_0 = 1 - p_0 = \sum_{p \in \Lambda_+^*} \ket{\varphi_p} \bra{\varphi_p}$ and $A \in \mathcal{L}(L^2(\Lambda))$ with $\norm{A}_{\mathcal{L} \left( L^2(\Lambda) \right) } = 1$. Then,
\begin{align}
\Tr \left( A \gamma^{(1)} \right) 
&= \Tr \left( A \left( p_0 \gamma^{(1)} p_0 +  q_0 \gamma^{(1)} q_0
+ p_0 \gamma^{(1)} q_0 + q_0\gamma^{(1)} p_0   \right) \right)
\nonumber \\
&= \frac{1}{N} \Tr \left( A  p_0 \right) 
\scp{\Psi_N}{\ad(\varphi_0) a(\varphi_0) \Psi_N}
\nonumber \\
&\quad 
+ \frac{1}{N} \sum_{p,q \in \Lambda_+^*} \Tr \left( A \ket{\varphi_p} \bra{\varphi_q}  \right)
\scp{\Psi_N}{\ad(\varphi_q) a(\varphi_p) \Psi_N}
\nonumber \\
&\quad +
 \frac{1}{N} \sum_{p \in \Lambda_+^*} \Big(\Tr \left( A  \ket{\varphi_0} \bra{\varphi_p} \right)
\scp{\Psi_N}{\ad(\varphi_p) a(\varphi_0) \Psi_N} 
\nonumber \\
&\quad\qquad\qquad\quad +
\Tr \left( A  \ket{\varphi_p} \bra{\varphi_0} \right)
\scp{\Psi_N}{\ad(\varphi_0) a(\varphi_p) \Psi_N}\Big) .
\end{align}
By means of \eqref{eqn:substitution:rules} this can be written as
\begin{subequations}
\begin{align}
\label{eq:trace norm distance estimate a}
\Tr \left( A \left( \gamma^{(1)} - p_0 \right) \right)
&=  -  \frac{1}{N}
\Tr \left( A  p_0 \right) 
\scp{\Chi}{ \FockT \Np \FockT^* \Chi} 
\\
\label{eq:trace norm distance estimate b}
&\quad 
+ \frac{1}{N} \sum_{p,q \in \Lambda_+^*} \Tr \left( A \ket{\varphi_p} \bra{\varphi_q}  \right)
 \scp{\Chi}{\FockT \ad_q a_p \FockT^* \Chi}
\\
\label{eq:trace norm distance estimate c}
&\quad +
\frac{1}{N} \sum_{p \in \Lambda_+^*} \Tr \left( A  \ket{\varphi_0} \bra{\varphi_p} \right)
\scp{\Chi}{\FockT \ad_p \sqrt{N - \Np} \FockT^* \Chi} 
\\
\label{eq:trace norm distance estimate d}
&\quad +
\Tr \left( A  \ket{\varphi_p} \bra{\varphi_0} \right)
\scp{\Chi}{\FockT \sqrt{N - \Np} a_p \FockT^* \Chi}  .
\end{align}
\end{subequations}
Note that 
\begin{align}
\label{eq:estimate number operator with T wrt to chi}
\scp{\Chi}{ \FockT \Np \FockT^* \Chi}  \ls 1
\end{align}
because of Lemma \ref{lem:T:Number} and Lemma \ref{lem:kinetic:energy:excitations}.
With \eqref{eq:estimate number operator with T wrt to chi}, we obtain
\begin{align}
\bigg| \sum_{p,q \in \Lambda_+^*} \Tr \left( A \ket{\varphi_p} \bra{\varphi_q}  \right)
 \scp{\Chi}{\FockT \ad_q a_p \FockT^* \Chi}
\bigg|
&=\lr{\Chi,\FockT\d\Gamma(q_0Aq_0)\FockT^*\Chi}\\
&\ls \norm{A}_{\mathcal{L} \left( \l^2(\Lambda) \right) } 
\norm{\left( \Np + 1 \right)^{1/2} \FockT^* \Chi}^2 
\ls 1 .
\end{align}
Hence,
\begin{align}
\abs{\eqref{eq:trace norm distance estimate a} + \eqref{eq:trace norm distance estimate b}} 
\ls \frac{1}{N} .
\end{align}
Using
\begin{align}
\label{eq:estimate for l-2 norm of trace of operator with varphi-0 and varphi-p}
 \sum_{p \in \Lambda_+^*} \abs{\Tr \left( A  \ket{\varphi_0} \bra{\varphi_p} \right) }^2
&\leq  \norm{A \varphi_0}^2
\leq \norm{A}_{\mathcal{L} \left( L^2(\Lambda) \right) }^2 \leq 1 ,
\end{align}
Lemma \ref{lem:T:Number}, Lemma \ref{lem:kinetic:energy:excitations} and
\eqref{eq:estimate for the projector P},
we estimate
\begin{align}
\abs{\eqref{eq:trace norm distance estimate c}}
&= \frac{1}{N} \bigg| \sum_{p \in \Lambda_+^*} \Tr \left( A  \ket{\varphi_0} \bra{\varphi_p} \right)
\scp{\Chi}{\FockT \ad_p \left( \sqrt{N - \Np} - N^{1/2} + N^{1/2} \right) \FockT^* \Chi} \bigg|
\nonumber \\
&\ls 
N^{-3/2} \norm{\left( \Np + 1 \right)^{3/4} \FockT^* \Chi}^2
+
N^{-1/2}  \bigg| \sum_{p \in \Lambda_+^*} \Tr \left( A  \ket{\varphi_0} \bra{\varphi_p} \right)
\scp{\Chi}{\FockT \ad_p \FockT^* \Chi} \bigg|
\nonumber \\
&\ls N^{-3/2} 
+ N^{-1/2} \norm{\Chi - \left( \Chi_0 + \Chi_1 + \Chi_2 \right) }
\Big[ \norm{\left( \Np + 1 \right)^{1/2} \FockT^* \Chi}
\nonumber \\
&\qquad \qquad \qquad \qquad \qquad \qquad \qquad \qquad \qquad  
+ \norm{\left( \Np + 1 \right)^{1/2} \FockT^* \left( \Chi_0 + \Chi_1 + \Chi_2 \right)}
\Big]
\nonumber \\ 
&\quad +
N^{-1/2}  \bigg| \sum_{p \in \Lambda_+^*} \Tr \left( A  \ket{\varphi_0} \bra{\varphi_p} \right)
\scp{\left( \Chi_0 + \Chi_1 + \Chi_2 \right)}{\FockT \ad_p \FockT^* \left( \Chi_0 + \Chi_1 + \Chi_2 \right)} \bigg|
\nonumber \\
&\ls N^{-3/2} 
+ N^{\frac{3}{2} \beta - 2} 
\Big[ 1
+ \norm{\left( \Np + 1 \right)^{1/2} \FockT^* \left( \Chi_0 + \Chi_1 + \Chi_2 \right)}
\Big]
\nonumber \\ 
&\quad +
N^{-1/2}  \bigg| \sum_{p \in \Lambda_+^*} \Tr \left( A  \ket{\varphi_0} \bra{\varphi_p} \right)
\scp{\left( \Chi_0 + \Chi_1 + \Chi_2 \right)}{\FockT \ad_p \FockT^* \left( \Chi_0 + \Chi_1 + \Chi_2 \right)} \bigg| .
\end{align}
Note that
\begin{align}
\scp{\Chi_0 }{\FockT \ad_p \FockT^* \Chi_0 } 
=  \scp{\Chi_0 }{\FockT \ad_p \FockT^* \Chi_2 } 
= \scp{\Chi_1 }{\FockT \ad_p \FockT^* \Chi_1} = \scp{\Chi_2 }{\FockT \ad_p \FockT^* \Chi_2}=
0
\end{align}
because all these expressions are vacuum expectation values of operators with an odd number of creation/annihilation operators.
Moreover, recall that
\begin{align}
\Chi_1 &=\REz\FockG_1\Chiz
\nonumber \\
&= \BogUz^* \BogUz \frac{\Qz}{\Ez-\FockGz}\BogUz^* \BogUz \FockG_1 \BogUz^* \vac
\nonumber \\
&= -
\frac{1}{\sqrt{N}} \sum_{\substack{p,q\in\Lsp\\p+q\neq0}} \frac{f(p,q)}{\mathfrak{e}(p+q) + \mathfrak{e}(p) + \mathfrak{e}(q)}
 \BogUz^*  \ad_{p+q} \ad_{-p}   \ad_{-q} \vac 
\end{align}
with $f(p,q)$ being defined as in Section \ref{section:explicit calculation of E-pert}.
Consequently,
\begin{align}
\lr{\Chi_0,\FockT\ad_p\FockT^*\Chi_1}
&=-
\frac{1}{\sqrt{N}} \sum_{\substack{p,q\in\Lsp\\p+q\neq0}} \frac{f(p,q)}{\mathfrak{e}(p+q) + \mathfrak{e}(p) + \mathfrak{e}(q)}\lr{\Omega,\BogUz\FockT\ad_p\FockT^*\BogUz^*\ad_{p+q} \ad_{-p}   \ad_{-q} \Omega}=0
\end{align}
because $\BogUz\FockT\ad_p\FockT^*\BogUz^*$ contributes only an annihilation operator, hence the inner product vanishes. 
Finally, 
\begin{align}
\norm{\left( \Np + 1 \right)^{1/2} \FockT^* \Chi_0 }
&\ls 1 ,
\quad 
\norm{\left( \Np + 1 \right)^{1/2} \FockT^*  \Chi_1 }
\ls N^{\frac{\beta - 1}{2}}  
, 
\quad 
\norm{\left( \Np + 1 \right)^{1/2} \FockT^*  \Chi_2}
\ls  N^{\beta - 1}
\end{align}
by Lemma \ref{lem:T:Number}, Lemma \ref{lem:BT:K0} and similar estimates as used in the proofs of Lemma \ref{lem:aux} and Lemma \ref{lem:aux:2}.
Using \eqref{eq:estimate for l-2 norm of trace of operator with varphi-0 and varphi-p} again and the Cauchy--Schwarz inequality leads to
\begin{align}
\abs{\eqref{eq:trace norm distance estimate c}}
&\ls N^{\frac{3}{2} (\beta - 1 )} .
\end{align}
By similar means we get
\begin{align}
\abs{\eqref{eq:trace norm distance estimate d}}
&\ls N^{\frac{3}{2} (\beta - 1 )} .
\end{align}
In total this proves the claim because the space of compact operators is the dual of the space of trace-class operators.\qed

\section*{Acknowledgements}
It is a pleasure to thank Chiara Boccato, Robert Seiringer, Christian Hainzl, Alessandro Olgiati and Phan Th\`anh Nam for helpful discussions. 
We are very grateful to the  Oberwolfach Research Institute for Mathematics for their hospitality during a Research in Pairs fellowship, where the first part of this project was completed.
L.B.\ was supported by the German Research Foundation within the Munich Center of Quantum Science and Technology (EXC 2111). N.L.\ acknowledges funding from the Swiss National Science Foundation through the NCCR SwissMap and support from the European Union's Horizon 2020 research and innovation programme under the Marie Sk{\textl}odowska-Curie grant agreement No.\ 101024712. S.P.\ acknowledges funding by the Deutsche Forschungsgemeinschaft (DFG, German Research Foundation) - project number 512258249.

\bibliographystyle{abbrv}

\begin{thebibliography}{10}

\bibitem{basti2023}
G.~Basti, S.~Cenatiempo, A.~Olgiati, G.~Pasqualetti, and B.~Schlein.
\newblock A second order upper bound for the ground state energy of a
  hard-sphere gas in the {Gross--Pitaevskii} regime.
\newblock {\em Commun. Math. Phys.}, 399(1):1--55, 2023.

\bibitem{boccato2017}
C.~Boccato, C.~Brennecke, S.~Cenatiempo, and B.~Schlein.
\newblock Complete {Bose--Einstein} condensation in the {Gross--Pitaevskii}
  regime.
\newblock {\em Commun. Math. Phys.}, 359(3):975--1026, 2018.

\bibitem{boccato2018}
C.~Boccato, C.~Brennecke, S.~Cenatiempo, and B.~Schlein.
\newblock Bogoliubov theory in the {Gross--Pitaevskii} limit.
\newblock {\em Acta Math.}, 222(2):219--335, 2019.

\bibitem{boccato2017_2}
C.~Boccato, C.~Brennecke, S.~Cenatiempo, and B.~Schlein.
\newblock The excitation spectrum of {Bose} gases interacting through singular
  potentials.
\newblock {\em J. Eur. Math. Soc.}, 22(7):2331--2403, 2020.

\bibitem{boccato2018_2}
C.~Boccato, C.~Brennecke, S.~Cenatiempo, and B.~Schlein.
\newblock Optimal rate for {Bose--Einstein} condensation in the
  {Gross-Pitaevskii} regime.
\newblock {\em Commun. Math. Phys.}, 376:1311--1395, 2020.

\bibitem{bogoliubov1947}
N.~N. Bogoliubov.
\newblock On the theory of superfluidity.
\newblock {\em Izv. Akad. Nauk Ser. Fiz.}, 11:23--32, 1947.

\bibitem{proceedings}
L.~Bo{\ss}mann.
\newblock Low-energy spectrum and dynamics of the weakly interacting bose gas.
\newblock {\em J.\ Math.\ Phys.}, 2022.

\bibitem{corr}
L.~Bo{\ss}mann, N.~Pavlovi\'c, P.~Pickl, and A.~Soffer.
\newblock Higher order corrections to the mean-field description of the
  dynamics of interacting bosons.
\newblock {\em J.\ Stat.\ Phys.}, 178(6):1362--1396, 2020.

\bibitem{QF}
L.~Bo{\ss}mann, S.~Petrat, P.~Pickl, and A.~Soffer.
\newblock Beyond {Bogoliubov} dynamics.
\newblock {\em Pure Appl.\ Anal.}, 3(4):677--726, 2021.

\bibitem{spectrum}
L.~Bo{\ss}mann, S.~Petrat, and R.~Seiringer.
\newblock Asymptotic expansion of low-energy excitations for weakly interacting
  bosons.
\newblock {\em Forum Math.\ Sigma}, 9:e28, 2021.

\bibitem{brennecke2022_2}
C.~Brennecke, B.~Schlein, and S.~Schraven.
\newblock Bogoliubov theory for trapped bosons in the {Gross--Pitaevskii}
  regime.
 \newblock {\em Ann.\ H.\ Poincar{\'e}}, 23(5):1583--1658, 2022.

\bibitem{brennecke2022}
C.~Brennecke, B.~Schlein, and S.~Schraven.
\newblock Bose--Einstein condensation with optimal rate for trapped bosons in
  the Gross--Pitaevskii regime.
\newblock {\em Mathematical Physics, Analysis and Geometry}, 25(2):12, 2022.

\bibitem{brooks2023}
M.~Brooks.
\newblock Diagonalizing Bose gases in the Gross-Pitaevskii regime and beyond.
\newblock {\em arXiv:2310.11347}, 2023.


\bibitem{caraci2023}
C.~Caraci, S.~Cenatiempo, and B.~Schlein.
\newblock The excitation spectrum of two-dimensional {Bose} gases in the
  {Gross--Pitaevskii} regime.
\newblock {\em Ann. Henri Poincar{\'e}}, pages 1--52, 2023.

\bibitem{caraci2023_2}
C.~Caraci, A.~Olgiati, D.~Saint Aubin, and B.~Schlein.
\newblock Third order corrections to the ground state energy of a {Bose} gas in the {Gross--Pitaevskii} regime.
\newblock {\em arXiv:2311.07433}, 2023.

\bibitem{derezinski2014}
J.~Derezi{\'n}ski and M.~Napi{\'o}rkowski.
\newblock Excitation spectrum of interacting bosons in the mean-field
  infinite-volume limit.
\newblock {\em Ann. Henri Poincar{\'e}}, 15(12):2409--2439, 2014.

\bibitem{grech2013}
P.~Grech and R.~Seiringer.
\newblock The excitation spectrum for weakly interacting bosons in a trap.
\newblock {\em Commun. Math. Phys.}, 322(2):559--591, 2013.

\bibitem{hainzl2020}
C.~Hainzl.
\newblock Another proof of {BEC} in the {GP}-limit.
\newblock {\em J. Math. Phys.}, 62(5):051901, 2021.

\bibitem{hainzl2022}
C.~Hainzl, B.~Schlein, and A.~Triay.
\newblock Bogoliubov theory in the {Gross-Pitaevskii} limit: a simplified
  approach.
\newblock {\em Forum Math. Sigma}, 10:e90, 2022.

\bibitem{lewin2015_2}
M.~Lewin, P.~T. Nam, S.~Serfaty, and J.~P. Solovej.
\newblock {B}ogoliubov spectrum of interacting {B}ose gases.
\newblock {\em Commun. Pure Appl. Math.}, 68(3):413--471, 2015.

\bibitem{lieb2002}
E.~H. Lieb and R.~Seiringer.
\newblock Proof of {Bose--Einstein} condensation for dilute trapped gases.
\newblock {\em Phys. Rev. Lett.}, 88(17):170409, 2002.

\bibitem{nam2020_2}
P.~T. Nam and M.~Napi{\'o}rkowski.
\newblock Two-term expansion of the ground state one-body density matrix of a
  mean-field {Bose} gas.
\newblock {\em Calc. Var. Partial Differential Equations}, 60(3):1--30, 2021.

\bibitem{nam2020}
P.~T. Nam, M.~Napi{\'o}rkowski, J.~Ricaud, and A.~Triay.
\newblock Optimal rate of condensation for trapped bosons in the
  {G}ross--{P}itaevskii regime.
\newblock {\em Analysis \& PDE}, 15(6):1585--1616, 2022.

\bibitem{nam2023}
P.~T. Nam and S.~Rademacher.
\newblock Exponential bounds of the condensation for dilute {Bose} gases.
\newblock {\em arXiv:2307.10622}, 2023.

\bibitem{nam2014}
P.~T. Nam and R.~Seiringer.
\newblock Collective excitations of {Bose} gases in the mean-field regime.
\newblock {\em Arch. Ration. Mech. Anal.}, 215:381--–417, 2015.

\bibitem{nam2021}
P.~T. Nam and A.~Triay.
\newblock Bogoliubov excitation spectrum of trapped {Bose} gases in the
  {Gross--Pitaevskii} regime.
\newblock {\em J.\ Math.\ Pures et Appliqu{\'e}es}, 2023.

\bibitem{pizzo2015}
A.~Pizzo.
\newblock Bose particles in a box {I}. {A} convergent expansion of the ground
  state of a three-modes {Bogoliubov Hamiltonian}.
\newblock {\em arXiv:1511.07022}, 2015.

\bibitem{seiringer2011}
R.~Seiringer.
\newblock The excitation spectrum for weakly interacting bosons.
\newblock {\em Commun. Math. Phys.}, 306(2):565--578, 2011.

\end{thebibliography}

\end{document}